\documentclass[a4paper,reqno]{amsart}
\usepackage[parfill]{parskip}  
\usepackage[T1]{fontenc}
\usepackage[english]{babel}
\usepackage{amssymb,bbm,enumerate}
\usepackage{bm}
\usepackage{xcolor}
\usepackage[linktocpage=true,colorlinks=true, linkcolor=blue, citecolor=red, urlcolor=green]{hyperref}
\usepackage{multirow} 
\usepackage{graphicx}

\usepackage{float}
\restylefloat{table}
\usepackage{tabu}
\usepackage[makeroom]{cancel}
\usepackage{cases}

\renewcommand{\ker}{\Ker}
\renewcommand{\Re}{\operatorname{Re}}
\renewcommand{\Im}{\operatorname{Im}}

\newcommand{\bb}[1]{\mathbb{#1}}
\newcommand{\mc}[1]{\mathcal{#1}}
\newcommand{\mf}[1]{\mathfrak{#1}}
\newcommand{\mr}[1]{\mathring{#1}}
\newcommand{\rmf}[1]{\mathring{\mathfrak{#1}}}
\newcommand{\tmf}[1]{\tilde{\mathfrak{#1}}}
\newcommand{\bmf}[1]{\bar{\mathfrak{#1}}}

\newcommand{\mb}[1]{\mathbb{#1}}

\newcommand{\beq}{\begin{equation}}
\newcommand{\eeq}{\end{equation}}
\newcommand{\e}{\varepsilon}
\newcommand{\wg}{\widetilde{\mathfrak g}}

\newcommand{\la}{\lambda}

\newcommand{\scaling}{{\bf d}}

\DeclareMathOperator{\End}{End}
\DeclareMathOperator{\Aut}{Aut}

\DeclareMathOperator{\ad}{ad}

\DeclareMathOperator{\Ker}{Ker}

\DeclareMathOperator{\rank}{rank}

\DeclareMathOperator{\op}{op}
\DeclareMathOperator{\Op}{Op}

\theoremstyle{plain}
\newtheorem{theorem}{Theorem}[section]
\newtheorem{lemma}[theorem]{Lemma}

\newtheorem{proposition}[theorem]{Proposition}
\newtheorem{corollary}[theorem]{Corollary}

\newtheorem{asu}{Assumption}

\theoremstyle{definition}
\newtheorem{definition}[theorem]{Definition}

\theoremstyle{remark}
\newtheorem{remark}[theorem]{Remark}

\numberwithin{equation}{section}

\definecolor{light}{gray}{.9}

\usepackage{float}
\restylefloat{table}
\usepackage{mathtools}
\usepackage{tikz}
\usetikzlibrary{chains}
\usetikzlibrary{cd}
\usepackage{tkz-euclide}

\tikzset{node distance=2em, ch/.style={circle,draw,on chain,inner sep=2pt},chj/.style={ch,join},every path/.style={shorten >=4pt,shorten <=4pt},line width=1pt,baseline=-1ex}

\usepackage{datenumber}

\newcounter{dateone}
\newcounter{datetwo}

\newcommand{\difftoday}[3]{%
\setmydatenumber{dateone}{\the\year}{\the\month}{\the\day}%
\setmydatenumber{datetwo}{#1}{#2}{#3}%
\addtocounter{datetwo}{-\thedateone}%
\textcolor{red}{\the\numexpr-\thedatetwo day(s) late}
}

\title{Feigin-Frenkel-Hernandez Opers  and the $QQ-$system}

\author{D.Masoero, A.Raimondo}

\address{Grupo de F\'isica Matem\'atica da Universidade de Lisboa,
Av. Prof. Gama Pinto 2, 1649-003 Lisboa, Portugal.}
\email{dmasoero@gmail.com}

\address{Dipartimento di Ingegneria Gestionale, dell’Informazione e della Produzione, Universit\`a di Bergamo, Italy, and INFN, Sezione di Milano-Bicocca, Piazza della Scienza 3, 20126 Milano, Italy.}
\email{andrea.raimondo@unibg.it}

\begin{document}

\pagestyle{plain}

\begin{abstract}
 This paper represents the completion of
 our work on the ODE/IM correspondence for the generalised quantum Drinfeld-Sokolov models.
We present a unified and general mathematical theory, encompassing all
particular cases that we had already addressed, and we fill important analytic and algebraic gaps in the literature on the ODE/IM correspondence.
For every affine Lie algebra $\mathfrak{g}$ -- whose Langlands dual $\mathfrak{g}'$ is the untwisted affinisation of a simple Lie algebra -- we study a class of affine twisted parabolic Miura $\mathfrak{g}$-opers,
 introduced by Feigin, Frenkel and Hernandez.  The Feigin-Frenkel-Hernandez opers are defined by fixing the singularity structure at $0$ and $\infty$, and by allowing a finite number of additional singular terms with trivial monodromy.
 We define the central connection matrix and Stokes matrix for these opers, and prove that the coefficients of the former satisfy the the $QQ$ system of the quantum $\mathfrak{g}'$-Drinfeld-Sokolov (or quantum $\mathfrak{g}'$-KdV) model.
 If $\mf{g}$ is untwisted, it is known that the trivial monodromy conditions are equivalent to a complete system of algebraic equations for the additional singularities. We prove
 a suprising negative result in the case $\mf{g}$ is twisted: in this case, the trivial monodromy conditions have no non-trivial solutions.
\end{abstract}

\maketitle

\tableofcontents

\section{Introduction}
In this paper we study a class of affine opers, introduced by Feigin and Frenkel \cite{FF11}, and Frenkel and Hernandez \cite{FH16}, and we prove that their monodromy data provides solutions to the Bethe equations of the quantum $\mf{g}'$-Drinfeld-Sokolov (or quantum $\mf{g}'$-KdV) model, where $\mf{g}'$ is an untwisted affine Kac-Moody algebra.
 This paper thus belongs to a research field called \emph{ODE/IM correspondence}, which encompasses a large family of (conjectural) relations between linear differential operators and quantum integrable models  \cite{doreytateo98,bazhanov01,BLZ04,dorey07,FF11,bazhanov14,dorey13,lukyanov13,lukyanov10,marava15,
 marava17,mara18,bazhanov18,gaiotto22}.

The FFH (Feigin-Frenkel-Hernandez) opers belongs to the class of affine twisted parabolic Miura opers, a concept that will introduce later in the paper. Here for sake of definiteness, we introduce them as concrete partial differential operator that we name FFH connections \footnote{In Theorem \ref{pro:quasinormal} below, we will show that any FFH oper admits a unique
representative as a FFH connection.}.
The Lie algebraic setting is as  follows - see Table \ref{table:affine} for more details.
We let $\mf{g}=\,^L(\mf{g}')$ be the Langlands dual Lie algebra of $\mf{g}'$ so that $\mf{g}$ is an affine Kac-Moody algebra of type $\tmf{g}^{(r)}$, where $\tmf{g}$ is a simply-laced simple Lie algebra and $r\in\{1,2,3\}$ is the order of a Dynkin diagram automorphism $\sigma$ of $\tmf{g}$.
We denote by $\rmf{g}$ the simple Lie algebra whose Dynkin diagram is obtained from that of $\mf{g}$ by removing the $0-$th node. If $r=1$ then $\rmf{g}=\tmf{g}$ while if $r>1$ then $\rmf{g}$ is the fixed-point subalgebra of $\tmf{g}$ under $\sigma$. In addition, we denote by $\rmf{g}=\rmf{n}^-\oplus\rmf{h}\oplus\rmf{n}^+$ the triangular decomposition of $\rmf{g}$, with $\rmf{h}$ a Cartan subalgebra and $\rmf{n}^+$  (resp. $\rmf{n}^-$) a positive (resp. negative) maximal nilpotent subalgebra of $\rmf{g}$.

Fixed global coordinates $(z,\la)$ on $\bb{C}^2$, the FFH connections for the algebra
$\tmf{g}^{(r)}$ are, by definition, the following
differential operators
whose coefficients are meromorphic functions with values in $\tmf{g}$,
\begin{align}
 \mc{L}=  z\,\partial_z+ k\; \la\, \partial_{\la} + & \mr{f}+\ell +  (z+\la )v_\theta + \sum_{j \in J} \frac{ r z^{r}}{z^r-w_j^r}\left(-\theta^\vee +X(j)
  \right) 
 \label{ffhintro} .
 \end{align}
Formula \eqref{ffhintro} describes a family of partial differential operators,
written in terms of the fixed elements $\mr{f},v_\theta,\theta^\vee \in \tmf{g}$ and of the parameters $\ell, k$, $\lbrace w_j,X(j),y(j)\rbrace_{j \in J}$, where:
\begin{itemize}
 \item $\mr{f} \in \rmf{n}^-$ is a principal nilpotent element of $\rmf{g}$,
$\theta^\vee\in\rmf{h}$ is the coroot corresponding to the highest short root $\theta$ of $\rmf{g}$, and $v_\theta$ is a highest weight vector, of weight $\theta$, for the $\rmf{g}-$module $\tmf{g}$.
\item $\ell \in \rmf{h}$ and $k \in \bb{R}$, $0<k<1$, are free parameters.
\item $J$ is a (possibly empty) finite-set of indexes and for every $j \in J$, $w_j \in \bb{C}^\times $, $X(j) \in \rmf{n}^+$, and $y(j) \in \bb{C}$. The parameters $\lbrace w_j,X(j),y(j)\rbrace_{j \in J}$ are not free; they are constrained by the requirements that i)
$w_i^r = w_j^r$ if and only if $i =j$, and ii) that \eqref{ffhintro} has trivial monodromy at each pole in $\bb{C}^\times$ of the coefficients of $\mc{L}$, for every value of $\lambda \in \bb{C}$.
\end{itemize}

\subsection{Description of the main results and Structure of the paper}
This paper builds on our previous works \cite{marava15,marava17,mara18} (\cite{marava15,marava17} in collaboration with D. Valeri) -- where the case $J= \emptyset$ or  $r=1$ was considered. A fundamental role for the results of the present work is also played by  \cite{FF11}  by Fegin and Frenkel, where the relation between quantum $\mf{g}'$-Drinfeld Sokolov systems and afffine $\mf{g}$-opers was suggested for the first time,  as well as  \cite{FH16} by Frenkel and Hernandez, where it was conjectured that the appropriate $\mf{g}$-opers to consider should be twisted parabolic Miura $\mf{g}$-opers. In the present paper, we provide
unified and general analytic, algebraic and geometric theory of FFH opers, and, as a by-product, we complete
our previous works and fill important gaps in the mathematical literature. In detail:

\subsubsection*{Section \ref{sec:liealgebra1}. Lie algebraic preliminaries} This section contains preliminary material on the simple Lie algebras $\tmf{g}$ and $\rmf{g}$, as well as on the affine algebra $\mf{g}$. In addition, we describe those $\tmf{g}$-modules and those properties of \textit{cyclic element}
$\mr{f}+ v_{\theta} $, which are relevant in the rest of the paper.
\subsubsection*{Section \ref{opers and ba}. Analytic Theory and QQ system}
On their face values FFH connections are abstract partial differential operators whose coefficients are meromorphic functions with value in $\tmf{g}$. There are, in principle, vast options of differential equations one may consider. In this section we describe the differential equations relevant for the ODE/IM correspondence and we study them.

Given a finite dimensional $\tmf{g}$-module $V$, we denote by $V(\la)$ the space of entire function
with values in $V$ and we let the differential operator \eqref{ffhintro} act as a (meromorphic) connection of the trivial bundle $ \bb{C}^* \times V(\la)$. We are thus led to consider the
ODE
\begin{equation}\label{eq:gendiffinto}
  \mc{L}\psi=0, \quad \psi: \left(\widetilde{\bb{C}}^\times \setminus \bigcup_{j \in J} \bigcup_{l=0}^{r-1}\Pi^{-1} (e^{\frac{2\pi i\, l}{r}} w_j)\right) \to V(\la),
\end{equation}
where $\Pi:\widetilde{\bb{C}^\times} \to \bb{C}^\times$ is the universal covering map.
The above differential equation is a linear ODE with values in a infinite-dimensional Frechet space and we address the
study of local and global solution in Proposition \ref{prop:localsolution} and Theorem \ref{prop:AVfree}.
In Proposition \ref{prop:localsolution}, we prove that the the Cauchy problem for (\ref{eq:gendiffinto}) is well-posed if the equation is restricted to any \textit{simply-connected} open subset $D$ of the domain.
In Theorem \ref{prop:AVfree}, we prove a structure theorem for the space of global solutions:
We notice that the space of global solutions
is a module of the ring
\begin{equation}\label{eq:O'introd}
 \mc{O}'=\lbrace f: \widetilde{\bb{C}}^\times \times \bb{C} \to \bb{C},\; f(z;\la)= Q(z^{-k}\la), \mbox{ with } Q: \bb{C} \to \bb{C} \mbox{ entire } \rbrace.
\end{equation}
We prove that a connection of the form \eqref{ffhintro} has trivial monodromy about all singularities in $\bb{C}^\times$ if and only if the space of global solution is a free-module of rank $\dim V$ of the ring of functions $\mc{O}'$ \footnote{In other words, among all connection of the form \eqref{ffhintro}, FFH connections are singled-out by the following property:
the space of global solution is a free-module of rank $\dim V$ of the ring of functions $\mc{O}'$.}.

We then study $\mc{O}'$-bases of solutions, with distinguished asymptotic behaviour at $0$ and $\infty$, in
Proposition \ref{prop:frobeniusolutions} and Theorem \ref{prop:basisatinfinity}, respectively.
In Proposition \ref{prop:frobeniusolutions}, following \cite{mara18}, we show, under generic non-resonant conditions on $\ell$,
that there exists a basis of Frobenius solutions. These are solutions which admits the following expansion in a neighborhood of $z=0$
\begin{equation}\label{eq:frobeniusintro}
 \chi_{\gamma}(z;\la) = z^{-\gamma} \sum_{m\geq 0} g_m(\la) z^m, \quad g_0(0)=\chi_{\gamma} ,
\end{equation}
where $g_m:\bb{C}\to V$ is an entire function and $\chi_{\gamma}\in V$ is an eigenvector of $\mr{f}+\ell $ with eigenvalue $\gamma$.
In Theorem \ref{prop:basisatinfinity}, using the recent result of one of the authors with Cotti and Guzzetti \cite{cotti23}, we study the asymptotic behaviour at $\infty$ of solution to \eqref{eq:gendiffinto}.
\begin{enumerate}
 \item We prove that for each (generic) sector of amplitude $\pi h$, there exists a basis of solutions with the asymptotic behaviour
 \begin{equation}\label{eq:asymptoticinfty}
  \Psi_{\nu}(z;\la)= \left(z^{\frac{1}{h}\mr{\rho}^\vee}\, \psi_{\nu}\right)e^{- h \nu z^{\frac{1}{h}} + C(z;\la) } \left(1+ O(z^{-\delta})\right), \delta >0,
 \end{equation}
as $z \to \infty$ in the given sector. Here $\psi_{\nu}$ is an eigenvector
of $\mr{f}+ v_{\theta}$ with eigenvalue $\nu$, $\mr{\rho}^\vee$ is the dual Weyl vector, $h$ the Coxeter number of $\mf{g}$, and $C(z;\la)=O( z^{\frac{1}{h}-k})$ certain polynomial expression in $z^{-k}$.
\item We define the central connection matrix $Q$ and the Stokes matrix $\mc{T}$. These are the matrix of change of basis, respectively between the basis of Frobenius solution and a basis at $\infty$ (in a sector containing the ray $\arg z=0$), and the matrix of change of basis between two \textit{consecutive bases} at $\infty$.
The coefficients of such matrices take values in the ring $\mc{O}'$, thus are eintire functions.
 They are denoted by symbol $Q$ and $\mc{T}$, because, according to the ODE/IM correspondence, they correspond to the
 Baxter $Q$ operators and the Transfer Matrix of the quantum Drinfeld-Sokolov model.
 \item  We prove that if $\nu$ is a subdominant eigenvalue, namely $\nu$ is such that
$e^{- h \nu z^{\frac{1}{h}}}$ goes to $0$ as fast as possible along a ray of the sector, then
$\Psi_{\nu}(z;\la)$ admits the
asymptotic \eqref{eq:asymptoticinfty} in a larger sector, of amplitude at least
$2 \pi h$.
\end{enumerate}
We remark here that the above results fill major gaps in the literature on the ODE/IM correspondence. In fact, the existence of bases at $\infty$, outside the case $\mf{sl}_2^{(1)}$, was never proven due to the lack, prior
to \cite{cotti23}, of a theory of ODEs with not meromorphic coefficients. Therefore, the matrices $Q,\mc{T}$ could not be defined. Moreover, the subdominant solutions was proven
to exist in \cite{marava15,marava17,mara18} (case $J=\emptyset$ or $r=1$) but its asymptotic was shown to hold only in a small sector, of amplitude $\pi h$; the fact that the asymptotics holds in a large sector is crucial for some applications, such as the extended $QQ$-system \cite{ekhammar20}.

Having studied distinguished solutions at $0$ and $\infty$,
we apply the machinery developed in \cite{marava15,marava17} -- following partial results in \cite{dorey07,Sun12} -- to construct solution of the Bethe Equations.
First we select $\rank \rmf{g}$ fundamental $\tmf{g}$-modules in such a way that in each module there exists a subdominant solution along the real positive axis and that the collection of these $\rank \rmf{g}$  subdominant satisfy a system of non-linear relations known as $\Psi$-system \eqref{eq:Psisystem}.
The coefficients in the expansion of the $\rank\rmf{g}$ subdominant solutions with respect to the basis of Frobenius solutions
belong to the ring $\mc{O}'$ -- hence are entire functions -- and are called generalised spectral determinants.
Substituting these expansions into the $\Psi$-system, and fixing an element $w$ of the Weyl group of $\rmf{g}$
\footnote{We consider the coefficients with respect to those elements of the basis of Frobenius solutions \eqref{eq:frobeniusintro} such that
$\gamma=w(\tilde{\omega}_i)(\ell)$ where $\tilde{\omega}_i$ is the $i-th$ fundamental weight of $\tmf{g}$. See \cite{ekhammar20} for
a detailed discussion on this point.}, one obtains a closed system of functional equations for $2\rank\rmf{g}$ of these entire functions, say $Q_w^{(i)}(\lambda)$, $\widetilde{Q}_w^{(i)}(\lambda)$, $i=1,\dots,\rank\rmf{g}$. This system goes under the name of \emph{$QQ$-system}, and it is given by:
\begin{align}
\begin{split}\label{QQsystemintro}
\prod_{j=1}^{\rank\rmf{g}}\prod_{u=0}^{B_{sj}-1}Q_w^{(j)}(q^{\frac{B_{sj}-1-2u}{r}} \la)
&=e^{\pi i D_s\langle \ell,w(\tilde{\alpha}_s)\rangle}Q_w^{(s)}(q^{D_s} \la)\widetilde{Q}_w^{(s)}(q^{- D_s}\la)\\
&-e^{- \pi i D_s\langle \ell,w(\tilde{\alpha}_s)\rangle}Q_w^{(s)}(q^{-D_s}\la)\widetilde{Q}_w^{(s)}(q^{D_s}\la)
\,,
\end{split}
\end{align}
for $s=1,\dots,\rank\rmf{g}$. In the above formula,  $q=e^{\pi i k}$, $B_{ij}=2\delta_{ij}-C_{ij}$ is the incidence matrix of $\rmf{g}$, while $\tilde{\alpha}_i$ are simple roots of $\tmf{g}$, and $w$ is an arbitrary element of the Weyl group of $\tmf{g}$.
As it was proven in \cite{marava17}, the $QQ$-system implies, under some genericity assumptions, the $^L\rmf{g}$ Bethe Equations:
\beq\label{BAintro}
\prod_{j=1}^{\rank\rmf{g}}e^{i\pi\overline{C}_{sj}\theta^j_w}\frac{Q_w^{(j)}(q^{\overline{C}_{sj}}\la^\ast)}{Q_w^{(j)}(q^{-\overline{C}_{sj}}\la^\ast)}=-1,\qquad s=1,\dots,\rank\rmf{g},
\eeq
for every zero $\lambda^\ast$ of $Q^{(i)}(\la)$. Here $\overline{C}_{sj}=C_{js}D_s$ is the symmetrized Cartan matrix of $^L\rmf{g}$ and $\theta^j_w=\langle \ell,w(\tilde{\omega}_j)\rangle$, with $\tilde{\omega}_j$ the $j$-th fundamental weight of $\tmf{g}$.
We remark here that in other contexts, Bethe Equations different from \eqref{BAintro} are associated to the algebra $^L\rmf{g}$, see e.g. \cite{frhere22,varchenko19}.

\subsubsection*{Section \ref{sec:trivialmonodromy}, trivial monodromy equations.}
In this section we study the trivial monodromy conditions for connections of the form \eqref{ffhintro}.
Following \cite{mara18}, we show that, fixed a basis of $\rmf{n}_+$, the trivial monodromy conditions are equivalent to a system of algebraic equations for the coefficients $\lbrace X(j),y(j),w_j\rbrace_{j \in J}$.
More precisely, in Theorem \ref{thm:trivialmonodromy} we prove that
\begin{itemize}
 \item If $r=1$, the trivial monodromy conditions are equivalent to a complete system of
 $|J|(2h-2)$ equations in$|J|(2h-2)$ scalar unknowns \cite{mara18};
\item If $r>1$, the trivial monodromy conditions are satisfied if and only if $J=\emptyset$.
\end{itemize}
The latter result is rather surprising. We comment on it in the general discussion of the ODE/IM literature below.
We add here that, in order to obviate this negative result, we compute trivial monodromy conditions for a larger class of connections which would lead in principle
to solutions of the QQ system (see \eqref{eq:ffopersourgauge-1} and the discussion therein), still obtaining no solutions.

\subsubsection*{Sections \ref{atpmo} and \ref{sec:opers}, twisted affine parabolic opers}
In the last part of the paper we study the geometric structure underlying FFH connections \eqref{ffhintro},
which is the notion of twisted parabolic Miura $\mf{g}-$opers, as suggested in \cite{FH16}.

In Section \ref{atpmo} we study Dynkin diagram automorphisms for affine algebras.
In Section \ref{sec:opers}, we develop the theory of twisted parabolic Miura $\mf{g}-$opers, which were first defined in \cite{FH16}.
This theory builds on the theory of parabolic Miura $\mf{g}-$opers, introduced (both in the finite and in the affine case) in \cite{FF11} as a generalisation of the notion of Miura opers \cite{frebenzvi01}, as well as on the theory of twisted opers, introduced by Frenkel and Gross in \cite{frgr09} (in the finite dimensional case).
Afterwards, we present the axiomatic definition of FFH opers, as a special class of twisted parabolic Miura
$\mf{g}-$opers defined by local conditions at the singular points \cite{FH16}, and in Theorem \ref{pro:quasinormal} we prove a normal form/rigidity theorem for FFH opers: we show that any FFH admits an essentially unique  (i.e. up to the action of the Weyl group of $\rmf{g}$ on $\ell$) representative as a FFH connection \eqref{ffhintro}.

 \subsection{Brief discussion of the literature on the ODE/IM correspondence}
 We conclude our introduction by contextualising our results within  the literature on the ODE/IM correspondence.
 First, we notice that  FFH connections \eqref{ffhintro} reduce, in particular cases, to connections that were already studied in the previous literature on the ODE/IM correspondence:
 \begin{itemize}
  \item The case $\tmf{g}^{(r)}=\mf{sl}_2^{(1)}$ marked the beginning of the ODE/IM correspondence in the seminal papers
  \cite{doreytateo98,BLZ04}. As shown in \cite{FF11},
  when $J=\emptyset$ connections \eqref{ffhintro} are equivalent of the anharmonic oscillators studied by Dorey and Tateo \cite{doreytateo98}, while in the case $\tmf{g}^{(r)}=\mf{sl}_2^{(1)}$ and arbitrary $J$, they are equivalent to the monster potentials studied by Bazhanov-Lukyanov-Zamolodchikov \cite{BLZ04}.
  \item When $r=1$ and $J=\emptyset$, the connections \eqref{ffhintro} were studied in \cite{marava15}, when $r>1$ and $J=\emptyset$, they were studied in \cite{marava17}, and when $r=1$ and $J \neq \emptyset$, they were studied in \cite{mara18} -- see Remark \ref{rem:xtoz} for the precise identification.
 \end{itemize}
A thorough discussion of the ODE/IM conjecture for FFH connections \eqref{ffhintro} can be found in
\cite{BLZ04,dorey07,marava17,FF11,FH16,mara18,coma20,coma21}, in various particular cases, and in \cite{FH16}, in the general case. Here, without repeating the detailed analysis from the cited literature, we simply recall that the quantum $\mf{g}'$-Drinfeld-Sokolov model is a family of Integrable Quantum Field Theories parameterised by the highest weight $p \in(^L\rmf{h})^\ast$ and central charge $c$ of a representation of the $\mc{W}$-algebra $\mc{W}_{\mf{g}'}$. The (conjectured) duality between \eqref{ffhintro} and $\mf{g}'$-Drinfeld-Sokolov model is as follows:
\begin{itemize}
 \item $\ell \in \rmf{h}$ corresponds to the highest weight/vacuum parameter $p\in(^L\rmf{h})^\ast$, where $(^L\rmf{h})^\ast$ is the dual of a Cartan subalgebra of $^L\rmf{g}$ - note  that $(^L\rmf{h})^\ast\simeq \rmf{h}$;
 \item $k$ corresponds to the central charge $c$;
 \item  Fixed the cardinality $|J|$ of $J$, the solutions $\lbrace w_j,X(j)\rbrace_{j \in J}$ of the trivial monodromy conditions should be in finite number and correspond to Bethe states of degree $|J|$, which are, by definition, the eigenvectors of the Baxter $\mc{Q}$ operators, whose conformal dimension is the conformal dimension of the vacuum plus $|J|$. In particular, when $J =\emptyset$ the oper should correspond to the ground state of the theory. The correspondence between solutions of the trivial monodromy conditions for FFH opers and Bethe states should be such that the generalised spectral determinants of the FFH connections coincide with the eigenvalues of the Baxter $\mc{Q}$ operators.
\end{itemize}
We stress here that while the first two relations are settled once the $QQ$-system \eqref{QQsystemintro} is derived, the latter relation has never been addressed in generality -- it was verified so far only in the case $\tmf{g}=\mathfrak{sl}_2$ (duality monster potentials/Quantum KdV) in the large momentum limit, see \cite{coma20,coma21}.
In fact, even though Baxter $\mc{Q}$ operators
can be defined in principle for quantum Drinfeld-Sokolov models \cite{FH16,wang23}, they were constructed explicitly
only in the case $\mf{g}=\mf{g}'=A_n^{(1)}$ \cite{baluzaI,bazhanov02integrable}, as formal power series.
The only claim that can be verified is that the FFH opers are complete
in the sense that any solution of the Bethe Equation \eqref{BAintro}, in a class of entire functions  satisfying certain boundary conditions expected to define the Quantum Drinfeld-Sokolov model
(see \cite{BLZ04,dorey07}), can be expressed as the spectral determinant of a FFH connection \eqref{ffhintro}.
Therefore, given the state-of-the-art of the mathematical theory of quantum Drinfeld-Sokolov models, our discovery that
the trivial monodromy conditions of FFH opers have no non-trivial solutions in the case of a twisted affine Kac-Moody algebra $\mf{g}$ might provide an indication  that in these cases the FFH $\mf{g}$-connections are not complete, or that in these cases the quantum ($\mf{g}'=\,^L\mf{g}$) theory is not complete. In the former case, one needs to define a different modification of the ground state oper \footnote{The FFH opers with $J=\emptyset$ or $r=1$ were already defined before the paper \cite{FH16}. Due to our proof that FFH opers with $J \neq \emptyset$ and $r>0$ do not exist, the opers that we consider should be a-posteriori be better called FF (Feigin-Frenkel) opers. However, since some of our results do not depend on the trivial monodromy conditions, we prefer to stick to the more general name.}.

The research about the ODE/IM correspondence has been rapidly evolving.
Besides the correspondence between FFH connections and quantum Drinfeld-Sokolov models, many more instances of the ODE/IM correspondences have been later found, including massive deformations of quantum Drinfeld-Sokolov model, the conformal limit of inhomogeneous XXZ-type chain, the sausage model, Kondo lines defects in product of chiral WZW models, see e.g. \cite{bazhanov14,dorey13,lukyanov13,lukyanov10,marava15,bazhanov18,gaiotto21,ekhammar20,
ito19,gaiotto22,fioravanti23,aramini23,kotousov22,bazhanov21scaling,fioravanti23}.
We notice that beyond the theoretical relevance of the ODE/IM correspondence, on the practical level the most efficient, if not the only, way to study integrable quantum field theory is to study the corresponding ODEs, see e.g. \cite{gaiotto22,kotousov22,bazhanov21scaling} for this point.
All these conjectural correspondences form an impressive body of evidence for the existence of a, yet mysterious,  fundamental duality between Integrable Quantum Field Theory and linear differential operators.
Assuming its existence, to any Bethe state of an integrable quantum field theory there corresponds a differential operator in such a way that the solutions to Bethe Equations coincide with the generalised spectral determinants, and the conserved charges can be computed from the coefficients of the operator,  see e.g. \cite{FF11,FH16,mara18,coma21,litvinov13,BLZ04,prochazka23}.

We conclude the introduction with a remark on the $QQ$-system \eqref{QQsystemintro}, which, prior to \cite{dorey07}, was known by other names in the Integrable Model literature, for example as \lq quantum Wronskian\rq\, in \cite{baluzaI} ($\mf{sl}_2$ case) and \cite{bazhanov02integrable} ($\mf{sl}_3$ case), and as  \textit{reproduction procedure} in \cite{mukhin04,mukhin08,mukhin14} (general case).
The $QQ$-system has played a fundamental role in the ODE/IM correspondence since the seminal paper of Dorey and Tateo \cite{doreytateo98}, and, after our works \cite{marava15,marava17}, it has become the cornerstone of the ODE/IM correspondence. Inspired by our derivation of the $QQ$-system, Frenkel and Hernandez \cite{FH16} proved that this is not a coincidence: the $QQ$-system holds as a universal system of relations in the commutative Grothendieck ring $K_0(\mc{O})$ of the category $\mc{O}$ of representations of the Borel subalgebra of the quantum affine algebra $U_q(^L\mf{g})$. This discovery has in turn inspired many developments in the theory of $QQ$-system, which has become an object of interest in itself, see e.g \cite{ekhammar20,fioravanti23,fksz23,frhere22,gaiotto21,korzei21,pinet23,rouven21,wang23,zeitlin22,zullo15}.
It is tempting to claim that the solution of \emph{the mystery of the ODE/IM correspondence} could be unveiled by a deeper understanding of the QQ system and its possible `representations'.

\subsection*{Acknowledgements.}
We are very grateful to Edward Frenkel
for insight and support. We also thank Anton Zeitlin, Dmytro Volin, Evgeny Mukhin, Peter Koroteev, Rouven Frassek and Tom\'as Proch\'azka for many useful discussions.
Davide Masoero thanks Giordano Cotti and Davide Guzzetti for the collaboration which led to the proof of Theorem \ref{thm:CGM}.\\
The authors gratefully acknowledge support from the Department of Mathematics of Genova University (DIMA),
where some of the research for the paper was performed.
Davide Masoero gratefully acknowledges support from the Simons Center for Geometry and Physics, Stony Brook University at which some of the research for this paper was performed. Andrea Raimondo thanks the Group of Mathematical
Physics of Lisbon university for the kind hospitality, during his many visits.
Davide Masoero is supported by the FCT projects UIDB/00208/2020, 2021.00091.CEECIND, and 2022.03702.PTDC (GENIDE), and he is a member of the COST Action CA21109 CaLISTA.
Andrea Raimondo is partially supported by funds of INFN (Istituto Nazionale di Fisica Nucleare) by IS-CSN4 Mathematical Methods of Nonlinear Physics,  and by the INdAM–GNFM Project CUP-E53C22001930001.


\section{Lie algebra preliminaries.}\label{sec:liealgebra1}
Let $\mf{g}'$ be an untwisted affine Kac-Moody algebra, and denote by $\mf{g}=\,^L\!\mf{g}'$ the (Langlands) dual Lie algebra. Then $\mf{g}$ is isomorphic to an affine algebra $\tmf{g}^{(r)}$, where $\tmf{g}$ is a simple and simply-laced Lie algebra and $r$ is the order of a Dynkin automorphism $\sigma$ of $\tmf{g}$. In this section we first review this basic construction, with the purpose of fixing the notation, and we consider some results on the representation theory of $\tmf{g}$, which is needed for the ODE/IM correspondence.

We index the nodes of the Dynkin diagram of $\mf{g}$ by the set $I=\{0,1,\dots,n\}$ as in Table \ref{table:dynkin(r=1)} and Table \ref{table:dynkin(r>1)}, and we let $C=(C_{ij})_{i,j\in I}$ be the Cartan matrix of $\mf{g}$. There exists integers $(a^\vee_0,a^\vee_1,\dots,a^\vee_n)$ and  $(a_0,a_1,\dots,a_n)$ which satisfy
\beq\label{220729-2/0}
\sum_{i\in I}a^\vee_i C_{ij}=0,\qquad \sum_{i\in I} C_{ji}a_i=0,\qquad j\in I.
\eeq  
and are uniquely specified by setting $a_0=a^\vee_0=1$. Removing the $0-$th node from the Dynkin diagram of $\mf{g}$ one gets the Dynkin diagram of a simple Lie algebra $\rmf{g}$, see  Table \ref{table:dynkinfd(r=1)}, Table \ref{table:dynkinfd(r>1)}, and Table \ref{table:affine}. The numbers
$$h=\sum_{i\in I}a_i,\qquad h^\vee=\sum_{i\in I}a_i^\vee,$$
are, respectively, the Coxeter number and the dual Coxeter number of $\mf{g}$.

\subsection{The simple Lie algebras $\tmf{g}$ and $\rmf{g}$.}
Let $\tmf{g}$ one of the simply-laced simple Lie algebras as in Table \ref{table:affine}, with Dynkin diagram as in Table \ref{table:dynkinfd(r=1)} and Table \ref{table:dynkinfd(r>1)}.  Let $\tilde{n}$ be the rank of $\tmf{g}$ and set $\tilde{I}=\{1,\dots,\tilde{n}\}$. Let $\{\tilde{e}_i,\tilde{\alpha}_i^\vee, \tilde{f}_i|i\in\tilde{I}\}$ be Chevalley generators of $\tmf{g}$, satisfying the relations $(i,j\in\tilde{I})$:
$$[\tilde{\alpha}_i^\vee,\tilde{\alpha}_j^\vee]=0,\quad [\tilde{\alpha}_i^\vee,\tilde{e}_j]=\tilde{C}_{ij}\tilde{e}_j,\quad [\tilde{\alpha}_i^\vee,\tilde{f}_j]=-\tilde{C}_{ij}\tilde{f}_j,\quad [\tilde{e}_i,\tilde{f}_j]=\delta_{ij}\tilde{\alpha}_i^\vee,$$
where $\tilde{C}=(\tilde{C}_{ij})_{i,j\in\tilde{I}}$ is the Cartan matrix of $\tmf{g}$. Fix a Cartan subalgebra 
 \beq\label{220718-6/0}
 \tmf{h}=\langle\tilde{\alpha}^\vee_i\,|\,i\in\tilde{I}\rangle, 
\eeq
such that $\{\tilde{\alpha}^\vee_i | i\in\tilde{I}\}\subset \tmf{h}$ are simple coroots, and denote by $\{\tilde{\alpha}_i | i\in\tilde{I}\}\subset \tmf{h}^\ast$ the corresponding set of simple roots, such that $\langle\tilde{\alpha}_i^\vee,\tilde{\alpha}_j\rangle =\tilde{C}_{ij}$. The algebra $\tmf{g}$ admits the root space decomposition
\beq\label{220728-3/0}
\tmf{g}=\tmf{h}\oplus\bigoplus_{\alpha\in\tilde{\Delta}}\tmf{g}_\alpha,
\eeq
where $\tilde{\Delta}\subset\tmf{h}^\ast$ is the set of roots and  $\tmf{g}_\alpha\subset\tmf{g}$ are the root spaces.  Let  $\tmf{n}^+$ (resp. $\tmf{n}^-$) be the  nilpotent subalgebra of $\tmf{g}$ generated by $\{\tilde{e}_i\,|\,i\in\tilde{I}\}$ (resp. $\{\tilde{f}_i\,|\,i\in\tilde{I}\}$), so that
\beq\label{triangulartildeg}
\tmf{g}=\tmf{n}^-\oplus\tmf{h}\oplus\tmf{n}^+.
\eeq
Let $\sigma$ be a Dynkin diagram automorphism of $\tmf{g}$, namely a permutation of the set $\tilde{I}$ such that $\tilde{C}_{\sigma(i)\sigma(j)}=\tilde{C}_{ij}$. Extend $\sigma$ to a Lie algebra automorphism (still denoted by $\sigma$) $\sigma\in\Aut(\tmf{g})$ defined on Chevalley generators by ($i\in\tilde{I}$) 
\beq\label{220723-2/0}
\sigma(\tilde{e}_i)=\tilde{e}_{\sigma(i)},\quad \sigma(\tilde{\alpha}^\vee_i)=\tilde{\alpha}^\vee_{\sigma(i)},\quad \sigma(\tilde{f}_i)=\tilde{f}_{\sigma(i)}.
\eeq
The automorphism $\sigma$ induces the following gradation on $\tmf{g}$:
\beq\label{221102-2/0}
\tmf{g}=\bigoplus_{\ell=0}^{r-1}\tmf{g}_\ell,\qquad \tmf{g}_\ell=\left\{x\in \tmf{g}\,|\,\sigma(x)=\e^\ell x\right\}.
\eeq
where
\beq\label{epsilon/0}
\e=e^\frac{2\pi i}{r}
\eeq
and $r$ is the order of $\sigma$. We denote by $\rmf{g}=\tmf{g}_0$ the fixed-point subalgebra of $\tmf{g}$ under the action of $\sigma$. Then $\rmf{g}$ is a simple Lie algebra, whose Dynkin diagram is obtained from the Dynkin diagram of $\tmf{g}$ by the folding induced by $\sigma$. 

Due to the numbering of the Dynkin diagrams in Table \ref{table:dynkinfd(r>1)} we have that the nodes $1,\dots,n$ lie in different orbits, so we can represent the set of orbits by $\mr{I}=\{1\dots,n\}\subset \tilde{I}$. Note, incidentally, that the Dynkin diagram of $\rmf{g}$ can also be obtained from that of the affine algebra $\mf{g}$ by removing the $0$-th node, so that $\mr{I}=I\setminus\{0\}=\{1,\dots,n\}$. For $i\in\tilde{I}$ we denote by $\langle i\rangle\in\bb{Z}^+$ the cardinality of the $i-$th $\sigma$-orbit, and we set
\beq\label{Di/0}
D_i=\frac{\langle i\rangle}{r},\qquad i\in\tilde{I}.
\eeq
The Cartan matrix $(C_{ij})_{i,j\in \mr{I}}$ of $\rmf{g}$ can be obtained summing over the columns of $\tilde{C}$ along the orbits of $\sigma$:
$$C_{ij}=\sum_{\ell=1}^{\langle i\rangle}\tilde{C}_{\sigma^\ell(i)j},\qquad i,j\in \mr{I}.$$
 The elements
\beq\label{220717-4/0}
e_i=\sum_{\ell=1}^{\langle i\rangle}\tilde{e}_{\sigma^\ell(i)},\quad \alpha^\vee_i=\sum_{\ell=1}^{\langle i\rangle}\tilde{\alpha}^\vee_{\sigma^\ell(i)},\quad f_i=\sum_{\ell=1}^{\langle i\rangle}\tilde{f}_{\sigma^\ell(i)},\qquad i\in \mr{I}.
\eeq
satisfy the relations ($i,j\in \mr{I}$)
$$[\alpha^\vee_i,\alpha^\vee_j]=0,\quad[\alpha^\vee_i,e_j]=C_{ij}e_j,\quad [\alpha^\vee_i,f_j]=-C_{ij}f_j,\quad [e_i,f_j]=\delta_{ij}\alpha^\vee_i,$$
together with the Serre relations, and are therefore Chevalley generators of $\rmf{g}$. In addition, the elements
\beq\label{221026-1/0}
\alpha_i=\frac{1}{\langle i\rangle}\sum_{\ell=1}^{\langle i\rangle}\tilde{\alpha}_{\sigma^\ell(i)},\qquad i\in \mr{I}
\eeq
are simple roots of $\rmf{g}$, namely they satisfy $\langle \alpha_i^\vee,\alpha_j\rangle=C_{ij}$ ($i,j\in \mr{I}$). 
Denoting
\beq\label{220720-1/0}
\rmf{n}^-=\tmf{n}^-\cap\rmf{g},\qquad\rmf{h}=\tmf{h}\cap\rmf{g},\qquad\rmf{n}^+=\tmf{n}^+\cap\rmf{g},
\eeq
then we obtain the triangular decomposition
$$\rmf{g}=\rmf{n}^-\oplus\rmf{h}\oplus\rmf{n}^+.$$
We denote by $\rmf{b}^\pm=\rmf{h}\oplus\rmf{n}^\pm$ the corresponding Borel subalgebras. In particular, we have:
\beq\label{220720-8}
\rmf{h}=\langle\alpha^\vee_i\,|\,i\in \mr{I}\rangle=\langle\alpha^\vee_1,\dots,\alpha^\vee_n\rangle,
\eeq
and similarly
$$\rmf{h}^\ast=\langle\alpha_i\,|\,i\in \mr{I}\rangle=\langle\alpha_1,\dots,\alpha_n\rangle.$$
We introduce the following elements:
\begin{enumerate}[i)]
\item Denote by $\mr{f}\in\rmf{g}$ the element
\beq\label{221026-2/0}
\mr{f}=\sum_{j\in \mr{I}}f_j,
\eeq
which is a principal nilpotent element of $\rmf{g}$.
\item Denote by $\theta^\vee\subset\rmf{h}$ and $\theta\subset\rmf{h}^\ast$ the elements
\beq\label{221026-3/0}
\theta^\vee=\sum_{i\in\mr{I}}a^\vee_i\alpha^\vee_i,\qquad \theta=\sum_{i\in\mr{I}}a_i\alpha^\vee_i,
\eeq
where  $a_i$ and $a^\vee_i$ ($i\in \mr{I}$) are the integer coefficients satisfying \eqref{220729-2/0}. The element $\theta$ is the highest short root of $\rmf{g}$ (i.e. the highest root of $\rmf{g}$ if $r=1$); its height is $h-1$, where $h$ is the Coxeter number of $\rmf{g}$.
\item Let $\{\mr{\omega}^\vee_i | i\in\mr{I}\}$ be fundamental coweights of $\rmf{g}$, defined by the relations $\langle\mr{\omega}^\vee_j,\alpha_j\rangle=\delta_{ij}$ ($i,j\in\mr{I}$). The element
\beq\label{mrrho}
\mr{\rho}^\vee=\sum_{i\in\mr{I}}\mr{\omega}^\vee_i
\eeq
satisfies $\langle\mr{\rho}^\vee,\alpha_i\rangle=1$ ($i\in\mr{I}$), and therefore its adjoint action induces the principal gradation on $\rmf{g}$, defined as
\beq\label{221102-1/0}
\rmf{g}=\bigoplus_{j=1-h^\vee}^{h^\vee-1}\rmf{g}^j,\qquad \rmf{g}^j=\{x\in\rmf{g}\,|\,[\mr{\rho}^\vee,x]=jx\},
\eeq
where $h^\vee$ is the dual Coxeter number of $\mf{g}$. In particular,
\beq
\rmf{g}^{-1}=\bigoplus_{j\in\mr{I}}\bb{C}\mr{f}_j,
\eeq
so that $\mr{f}\in\rmf{g}^{-1}$. 
\end{enumerate}
The principal gradation \eqref{221102-1/0} of $\rmf{g}$ induces a gradation on $\mf{s}$ by setting $\mf{s}^{i}=\mf{s}\cap\rmf{g}^i$. Moreover, one has $\dim\mf{s}^i=\dim\rmf{g}^{i}-\dim\rmf{g}^{i+1}$, and if $\dim\mf{s}^i>0$ then $i$ is said to be an exponent of $\rmf{g}$. Counting multiplicities there are $n$ exponents, which we will denote by $d_1,\dots,d_n$.

The gradation \eqref{221102-2/0} induced by $\sigma$ decomposes $\tmf{g}$ into a direct sum of $\rmf{g}-$modules: the subspace $\tmf{g}_0\simeq\rmf{g}$ is a $\rmf{g}-$module via the adjoint representation, if $r>1$ $\tmf{g}_{1}$ is the (unique) quasi-minuscule $\rmf{g}-$module, and if $r>2$ then $\tmf{g}_{2}$ is isomorphic to $\tmf{g}_{1}$. The weights of $\tmf{g}_1$ (and of $\tmf{g}_2$ for $r>2$) are the zero weight and the short roots of $\rmf{g}$.
\begin{definition}\label{221108-2/0}
Recall that the element $\theta$ given in \eqref{221026-3/0} is the highest short root of $\rmf{g}$. We denote by $V_\theta$ the unique quasi-minuscule $\rmf{g}$-module, namely the irreducible $\rmf{g}-$module with highest weight $\theta$, and we let $v_\theta\in V_\theta$ be a highest weight vector. If $r=1$ then  $V_\theta\simeq \rmf{g}\simeq \tmf{g}$ (the adjoint representation) and we identify $v_\theta$ with $e_\theta$, the highest root vector of $\tmf{g}\simeq \rmf{g}$. If $r>1$ then $V_\theta\simeq \tmf{g}_1$, and in particular  $v_\theta\in\tmf{g}\setminus\rmf{g}$.
\end{definition}
The element $v_\theta$ is defined up to a nonzero scalar multiple, which will be fixed in the next section (see Theorem \ref{thm:genlamba}). Note, incidentally, that denoting by $\mr{\mc{N}}$ the group 
\beq\label{221108-1/0}
\mr{\mc{N}}=\{\exp(y)\,|\,y\in\rmf{n}^+\},
\eeq
then $v_\theta$ is fixed by $\mr{\mc{N}}$, that is
\beq\label{221115-2/0}
gv_\theta=v_\theta,\qquad g\in\mr{\mc{N}}.
\eeq
This follows immediately from the fact that $v_\theta$ is a highest weight vector for a $\rmf{g}-$module. 
 %


\begin{center}
\begin{table}[h]
\caption{Dynkin diagram of $\tmf{g}\simeq \rmf{g}$ (case $r=1$).}\label{table:dynkinfd(r=1)}
\begin{center}
\begin{tabular}{c c}

\begin{tikzpicture}
\coordinate (N1) at (0,0);
\coordinate (N2) at (1,0) ; 
\coordinate (Nn-1) at (3,0) ;
\coordinate (Nn) at (4,0);
\foreach \i in {1,2,n-1,n}{\filldraw (N\i) circle[radius=2pt] node[below]{\small{$\i$}};}


\draw[thick] (Nn) -- (Nn-1);
\draw[thick] (N2)--(N1);
\draw[thick] (N2) -- +(.7,0);
\draw[dotted,thick ] (N2)--(Nn-1);
\draw[thick] (Nn-1) -- +(-.7,0);


\node at ($(N1)!.5!(Nn)+(0,1)$) {$\tmf{g}=\rmf{g}=A_{n}$, $(n\geq 1)$};
\end{tikzpicture}

\\
\\
\\

\begin{tikzpicture}
\coordinate (N1) at (0,0);
\coordinate (N3) at ($(N1)+(1,0)$) ; 
\coordinate (Nn-3) at ($(N3)+(2,0)$);
\coordinate (Nn-2) at ($(Nn-3)+(1,0)$);

\coordinate (N0) at ($(N1)+(135:1)$);
\coordinate (Nn-1) at ($(Nn-2)+(45:1)$);
\coordinate (Nn) at ($(Nn-2)+(-45:1)$);

\foreach \i in {1,2,n-3,n-2,n}{\filldraw (N\i) circle[radius=2pt] node[below]{\small{$\i$}};}
\foreach \i in {n-1}{\filldraw (N\i) circle[radius=2pt] node[above]{\small{$\i$}};}

\draw[thick] (N1)--(N2);
\draw[dotted,thick] (N2)--(Nn-3);
\draw[thick] (N3) -- +(.7,0);
\draw[thick] (Nn-3) -- +(-.7,0);
\draw[thick] (Nn-3)--(Nn-2)--(Nn-1);
\draw[thick] (Nn-2)--(Nn);

\node at ($(N2)!.5!(Nn-2)+(0,1)$){$\tmf{g}=\rmf{g}=D_n$, $(n\geq 4)$};

\end{tikzpicture}

\\
\\

\begin{tikzpicture}
\coordinate (N1) at (0,0);
\coordinate (N2) at ($(N1)+(1,0)$);
\coordinate (N3) at ($(N2)+(1,0)$);
\coordinate (N4) at ($(N3)+(0,1)$);
\coordinate (N5) at ($(N3)+(1,0)$);
\coordinate (N6) at ($(N5)+(1,0)$);

\foreach \i in {1,2,...,6}{\filldraw (N\i) circle[radius=2pt] node[below right]{$\i$};}

\draw[thick] (N1)--(N2) -- (N3) -- (N5) --(N6);
\draw[thick] (N3) -- (N4);

\node at ($(N1)!.5!(N6)+(0,2)$) {$\tmf{g}=\rmf{g}=E_{6}$};
\end{tikzpicture}

\\
\\

\begin{tikzpicture}
\coordinate (N1) at (0,0);
\coordinate (N2) at ($(N1)+(1,0)$);
\coordinate (N3) at ($(N2)+(1,0)$);
\coordinate (N4) at ($(N3)+(1,0)$);
\coordinate (N5) at ($(N4)+(0,1)$);
\coordinate (N6) at ($(N4)+(1,0)$);
\coordinate (N7) at ($(N6)+(1,0)$);

\foreach \i in {1,2,...,7}{\filldraw (N\i) circle[radius=2pt] node[below right]{$\i$};}

\draw[thick] (N1)--(N2) -- (N3) -- (N4) --(N6) -- (N7);
\draw[thick] (N4) -- (N5);

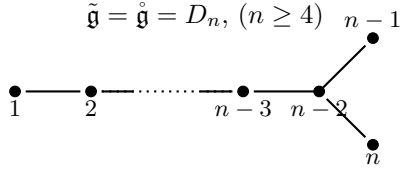
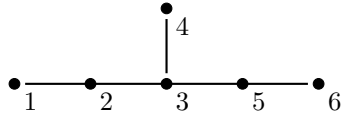
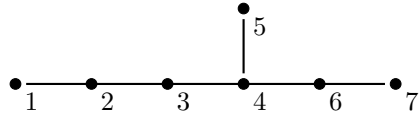
\node at ($(N1)!.5!(N7)+(0,2)$) {$\tmf{g}=\rmf{g}=E_{7}$};
\end{tikzpicture}

\\
\\

\begin{tikzpicture}
\coordinate (N1) at ($(N0)+(1,0)$);
\coordinate (N2) at ($(N1)+(1,0)$);
\coordinate (N3) at ($(N2)+(1,0)$);
\coordinate (N4) at ($(N3)+(1,0)$);
\coordinate (N5) at ($(N4)+(1,0)$);
\coordinate (N6) at ($(N5)+(0,1)$);
\coordinate (N7) at ($(N5)+(1,0)$);
\coordinate (N8) at ($(N7)+(1,0)$);

\foreach \i in {1,2,...,8}{\filldraw (N\i) circle[radius=2pt] node[below right]{$\i$};}

\draw[thick] (N1)--(N2) -- (N3) -- (N4) --(N5) -- (N7) -- (N8);
\draw[thick] (N6) -- (N5);

\node at ($(N1)!.5!(N8)+(0,2)$) {$\tmf{g}=\rmf{g}=E_{8}$};
\end{tikzpicture}

\end{tabular}
\end{center}
\end{table}
\end{center}


\begin{center}
\begin{table}[h]
\caption{Dynkin diagram of $\tmf{g}$ with the automorphism $\sigma$ (of order $r>1$), and the Dynkin diagram of the folded algebra $\rmf{g}$. In the table, we set $D_3=A_3$.}\label{table:dynkinfd(r>1)}
\begin{center}
\begin{tabular}{c c}

\begin{tikzpicture}
\coordinate (N1) at (0,0);
\coordinate (N2) at (1,0) ; 
\coordinate (Nn_2) at (3,0) ;
\coordinate (Nn_1) at (4,0);
\filldraw (N1) circle[radius=2pt] node[above right]{\small{$1$}};
\filldraw (N2) circle[radius=2pt] node[above]{\small{$2$}};
\filldraw (Nn_2) circle[radius=2pt] node[below]{\small{$n-2$}};
\filldraw (Nn_1) circle[radius=2pt] node[above]{\small{$n-1$}};
\draw[thick] (N1)--(N2);
\draw[thick] (N2) -- +(.7,0);
\draw[dotted,thick ] (N2)--(Nn_2);
\draw[thick] (Nn_2) -- +(-.7,0);
\draw[thick] (Nn_2) -- (Nn_1);

\coordinate (Nn) at ($(Nn_1)+(45:1)$);
\filldraw (Nn) circle[radius=2pt] node[above]{\small{$n$}};
\draw[thick] (Nn_1)--(Nn);

\coordinate (Nn+1) at ($(Nn_1)+(-45:1)$);
\filldraw (Nn+1) circle[radius=2pt] node[below]{\small{$n+1$}};
\draw[thick] (Nn_1)--(Nn+1);

\coordinate (Nmed) at ($(Nn)!.5!(Nn+1)$);
\node at ($(N1)!.5!(Nmed)+(0,1.5)$) {$\tmf{g}=D_{n+1}$, $(n\geq 2)$};

\draw [<->, >=stealth, dashed, thin](Nn+1) .. controls ($(Nn+1)!.5!(Nn)+(.2,0)$)   .. (Nn) ;
\node[right] at ($(Nn+1)!.5!(Nn)+(.2,0)$){$\sigma$};
\end{tikzpicture}

&

\begin{tikzpicture}
\coordinate (N1) at (0,0);
\coordinate (N2) at (1,0) ; 
\coordinate (Nn_1) at (3,0) ;
\coordinate (Nn) at (4,0);
\node at ($(N1)!.5!(Nn)+(0,1.5)$) {$\rmf{g}=B_{n}$, $(n\geq 2)$};
\filldraw (N1) circle[radius=2pt] node[above]{\small{$1$}};
\filldraw (N2) circle[radius=2pt] node[above]{\small{$2$}};
\filldraw (Nn_1) circle[radius=2pt] node[above]{\small{$n-1$}};
\filldraw (Nn) circle[radius=2pt] node[above]{\small{$n$}};
\draw[thick] (N1)--(N2);
\draw[thick] (N2) -- +(.7,0);
\draw[dotted,thick ] (N2)--(Nn_1);
\draw[thick] (Nn_1) -- +(-.7,0);
\draw[->,-stealth,thick,double] (Nn_1) -- (Nn);
\end{tikzpicture}
\\
\\
\\

\begin{tikzpicture}
\coordinate (Nn) at (0,0);
\coordinate (Nn+1) at ($(Nn)+(120:1)$);
\coordinate (N2n_2) at ($(Nn+1)+(-2,0)$);
\coordinate (N2n_1) at ($(N2n_2)+(-1,0)$);
\coordinate (Nn_1) at ($(Nn)+(240:1)$);
\coordinate (N2) at ($(Nn_1)+(-2,0)$);
\coordinate (N1) at ($(N2)+(-1,0)$);
\filldraw (Nn) circle[radius=2pt] node[above right]{\small{$n$}};
\filldraw (Nn+1) circle[radius=2pt] node[above]{\small{$n+1$}};
\filldraw (N2n_2) circle[radius=2pt] node[above]{\small{$2n-2$}};
\filldraw (N2n_1) circle[radius=2pt] node[above]{\small{$2n-1$}};
\filldraw (Nn_1) circle[radius=2pt] node[below]{\small{$n-1$}};
\filldraw (N2) circle[radius=2pt] node[below]{\small{$2$}};
\filldraw (N1) circle[radius=2pt] node[below]{\small{$1$}};

\draw[thick] (Nn) -- (Nn_1);
\draw[thick] (Nn_1) -- +(-.7,0);
\draw[dotted,thick] (Nn_1) -- (N2);
\draw[thick] (N2) -- +(.7,0);
\draw[thick] (N1) -- (N2);
\draw[thick] (Nn) -- (Nn+1);
\draw[thick] (Nn+1) -- +(-.7,0);
\draw[dotted,thick] (Nn+1) -- (N2n_2);
\draw[thick] (N2n_2) -- +(.7,0);
\draw[thick] (N2n_1) -- (N2n_2);

\draw [<->, >=stealth, dashed, thin](Nn+1) .. controls ($(Nn+1)!.5!(Nn_1)+(-.2,0)$)   .. (Nn_1) ;
\draw [<->, >=stealth, dashed, thin](N2n_2) .. controls ($(N2n_2)!.5!(N2)+(-.2,0)$)   .. (N2) ;
\node[right] at ($(N2n_2)!.5!(N2)+(-.2,0)$){$\sigma$};
\draw [<->, >=stealth, dashed, thin](N2n_1) .. controls ($(N2n_1)!.5!(N1)+(-.2,0)$)   .. (N1) ;

\node at ($($(N2n_1)!.5!(N1)+(-.7,0)$)!.5!(Nn)+(0,2)$) {$\tmf{g}=A_{2n-1}$, $(n\geq 3)$};

\end{tikzpicture}

&

\begin{tikzpicture}
\coordinate (Nn) at (0,0);
\coordinate (Nn_1) at ($(Nn)+(-1,0)$);
\coordinate (N3) at ($(Nn_1)+(-2,0)$);
\coordinate (N1) at ($(N3)+(-1,0)$);
\filldraw (Nn) circle[radius=2pt] node[above]{\small{$n$}};
\filldraw (Nn_1) circle[radius=2pt] node[above]{\small{$n-1$}};
\filldraw (N3) circle[radius=2pt] node[above]{\small{$2$}};
\filldraw (N1) circle[radius=2pt] node[above right]{\small{$1$}};

\draw[->,-stealth,thick,double] (Nn) -- (Nn_1);
\draw[thick] (Nn_1) -- +(-.7,0);
\draw[dotted,thick] (Nn_1) -- (N3);
\draw[thick] (N3) -- +(.7,0);
\draw[thick] (N1) -- (N3);

\node at ($(N1)!.5!(Nn)+(0,2)$) {$\rmf{g}=C_{n}$, $(n\geq 3)$};

\end{tikzpicture}
\\
\\

\begin{tikzpicture}
\coordinate (N3) at (0,0);
\coordinate (N4) at ($(N3)+(1,0)$);
\coordinate (N2) at ($(N3)+(135:1)$);
\coordinate (N1) at ($(N2)+(-1,0)$);
\coordinate (N5) at ($(N3)+(225:1)$);
\coordinate (N6) at ($(N5)+(-1,0)$);
\filldraw (N1) circle[radius=2pt] node[above]{\small{$1$}};
\filldraw (N2) circle[radius=2pt] node[above]{\small{$2$}};
\filldraw (N3) circle[radius=2pt] node[above]{\small{$3$}};
\filldraw (N4) circle[radius=2pt] node[below]{\small{$4$}};
\filldraw (N5) circle[radius=2pt] node[below]{\small{$5$}};
\filldraw (N6) circle[radius=2pt] node[below]{\small{$6$}};

\draw[thick] (N1) -- (N2);
\draw[thick] (N2) -- (N3);
\draw[thick] (N3) -- (N4);
\draw[thick] (N3) -- (N5);
\draw[thick] (N5) -- (N6);

\draw [<->, >=stealth, dashed, thin](N2) .. controls ($(N2)!.5!(N5)+(-.2,0)$)   .. (N5) ;
\draw [<->, >=stealth, dashed, thin](N1) .. controls ($(N1)!.5!(N6)+(-.2,0)$)   .. (N6) ;
\node[left] at ($(N1)!.5!(N6)+(-.2,0)$){$\sigma$};

\node at ($($(N1)!.5!(N6)+(-.8,0)$)!.5!(N4)+(0,1.5)$) {$\tmf{g}=E_{6}$};

\end{tikzpicture}

&

\begin{tikzpicture}
\coordinate (N3) at (0,0);
\coordinate (N4) at ($(N3)+(1,0)$);
\coordinate (N2) at ($(N3)+(-1,0)$);
\coordinate (N1) at ($(N2)+(-1,0)$);

\foreach \i in {1,...,4}{\filldraw (N\i) circle[radius=2pt] node[above]{\small{$\i$}};}

\draw[thick] (N1) -- (N2);
\draw[->,-stealth,thick,double] (N3) -- (N2);
\draw[thick] (N3) -- (N4);

\node at ($(N1)!.5!(N4)+(0,1.5)$) {$\rmf{g}=F_{4}$};

\end{tikzpicture}
 \\
 \\

\begin{tikzpicture}
\coordinate (N2) at (0,0);
\coordinate (N1) at ($(N2)+(-1,0)$);
\coordinate (N4) at ($(N2)+(60:1)$);
\coordinate (N3) at ($(N2)+(300:1)$);
\foreach \i in {2,3,4}{\filldraw (N\i) circle[radius=2pt] node[right]{\small{$\i$}};}
\foreach \i in {1}{\filldraw (N\i) circle[radius=2pt] node[left]{\small{$\i$}};}
\draw[thick] (N2) -- (N3);
\draw[thick] (N2) -- (N4);
\draw[thick] (N1) -- (N2);
\draw [->, >=stealth, dashed, thin](N4) .. controls ($(N4)!.5!(N1)+(120:0.3)$)   .. (N1) ;
\draw [->, >=stealth, dashed, thin](N1) .. controls ($(N1)!.5!(N3)+(240:0.3)$)   .. (N3) ;
\draw [<-, >=stealth, dashed, thin](N4) .. controls ($(N4)!.5!(N3)+(.3,0)$)   .. (N3) ;
\node[above left] at ($(N4)!.5!(N1)+(-.2,0)$){$\sigma$};
\node at ($(N1)!.5!($(N4)!.5!(N3)$)+(0,1.5)$) {$\tmf{g}=D_{4}$};

\end{tikzpicture}

&

\begin{tikzpicture}
\tikzset{Rightarrow/.style={double equal sign distance,>={Implies},->},
triple/.style={-,preaction={draw,Rightarrow}},
quadruple/.style={preaction={draw,Rightarrow,shorten >=0pt},shorten >=1pt,-,double,double
distance=0.2pt}}

\coordinate (N2) at (0,0);
\coordinate (N1) at ($(N2)+(-1,0)$);

\filldraw (N1) circle[radius=2pt] node[above]{\small{$1$}};
\filldraw (N2) circle[radius=2pt] node[above]{\small{$2$}};

\draw[thick,->,-stealth,triple] (N2) -- (N1);

\node at ($(N1)!.5!(N2)+(0,1.5)$) {$\rmf{g}=G_{2}$};
\end{tikzpicture}

\end{tabular}
\end{center}
\end{table}
\end{center}

\subsection{The affine Kac-Moody algebra $\mf{g}=\tmf{g}^{(r)}$.} 
Let $\mf{g}=\tmf{g}^{(r)}$ be one of the affine Kac-Moody algebra listed in  Table \ref{table:affine}. We recall the loop presentation of $\mf{g}$, the construction is standard \cite{Kac90}. Let $\mc{L}(\tilde{\mf{g}})=\tilde{\mf{g}}[\lambda,\lambda^{-1}]$ denote the loop algebra of $\tilde{\mf{g}}$, with the natural Lie algebra structure which extends the one of  $\tilde{\mf{g}}$.   Extend also the action of $\sigma$ from  $\tilde{\mf{g}}$ to $\mc{L}(\tilde{\mf{g}})$ by letting
$$\sigma(\lambda^m x)=
\e^{-m}
\lambda^m\sigma(x),\qquad m\in\bb{Z}, x\in \tilde{\mf{g}},$$
with $\e$ given by \eqref{epsilon/0}. Note that this action reduces to the identity when $r=1$. The affine Kac-Moody algebra $\mf{g}=\tilde{\mf{g}}^{(r)}$ is defined as
\beq\label{loopg}
\mf{g}=\tilde{\mf{g}}^{(r)}=\bigoplus_{k=0}^{r-1}\lambda^k\tilde{\mf{g}}_k[\lambda^{r},\lambda^{-r}]\oplus \bb{C}K\oplus \bb{C}{\bf d},
\eeq
where $K$ is central and the element $\scaling$ satisfies $[{\bf d},\lambda^mx]=-m\lambda^mx$, for $m\in\bb{Z}$ and $x\in\tilde{\mf{g}}$. Defining $\sigma(K)=K$ and $\sigma({\bf d})={\bf d}$ then $\tilde{\mf{g}}^{(r)}$ is $\sigma$-invariant. Let $I=\{0,1,\dots,n\}=\{0\}\cup\mr{I}$. The Chevalley generators $\{e_i,\alpha^\vee_i,f_i \, |\, i\in I\}$ of $\tilde{\mf{g}}^{(r)}$ can be written in terms of $\tilde{\mf{g}}$ as follows. The elements $\{e_i,\alpha^\vee_i,f_i \, |\, i\in\mr{I}\}$ are given by \eqref{220717-4/0}, in particular, they generate the simple Lie algebra $\tilde{\mf{g}}_0\simeq\mathring{\mf{g}}$ described above. On the other hand, the generators $e_0$, $\alpha^\vee_0$ and $f_0$ are given by
\beq\label{221116-1/0}
f_0=\lambda v_\theta,\qquad \alpha^\vee_0=K-\theta^\vee,\qquad e_0=\lambda^{-1}v_{-\theta},
\eeq
where $\theta^\vee\in\rmf{h}$ is given by \eqref{221026-3/0}, the element $v_{\theta}\in\tmf{g}$ was introduced in Definition \ref{221108-2/0},  and $v_{-\theta}\in\tmf{g}$ is uniquely defined by the relation $(v_\theta|v_{-\theta})=1$, where $(\cdot|\cdot)$ is the normalized invariant bilinear form on $\tmf{g}$. Finally, the scaling element $\scaling$ is realized  as
\beq\label{221116-2}
\scaling=-\lambda\partial_\lambda.
\eeq

\begin{remark}\label{rmk:negativesign}
The choice of the negative sign in \eqref{221116-2} together with \eqref{221116-1/0} implies the relations $[\scaling,e_0]=e_0$, $[\scaling,f_0]=-f_0$.
\end{remark}
We define
\beq\label{f}
f=\sum_{i\in I}f_i=\sum_{i\in\mr{I}}f_i+f_0=\mr{f}+f_0,
\eeq
where $\mr{f}$ is the principal nilpotent element \eqref{221026-2/0}. Equivalently, we can write
\beq\label{fla}
f=\mr{f}+\lambda v_\theta.
\eeq
We extend the action of the group $\mr{\mc{N}}$ to $\mf{g}$ by setting
\beq\label{230913-1}
g(\lambda^mx)=\lambda^m gx,\qquad gK=K,\qquad g\scaling=\scaling,
\eeq
for $g\in\mr{\mc{N}}$, $m\in\bb{Z}$, $x\in\tmf{g}$. In particular, due to \eqref{221115-2/0} we have
\beq\label{230913-2}
gf_0=g(\la v_\theta)=\la gv_\theta=\la v_\theta=f_0.
\eeq

\subsection{ $\tmf{g}$-modules}
Let $\tmf{g}$ be the simple Lie algebra \eqref{220728-3/0} with Cartan subalgebra \eqref{220718-6/0}.  Denote by $P(\tmf{g})\subset\tmf{h}^\ast$ the weight lattice of $\tmf{g}$ and by $P^+(\tmf{g})$ the set of dominant integral weights. For every $\omega\in P^+(\tmf{g})$ we denote by $L(\omega)$ the irreducible highest weight $\tmf{g}$-module with highest weight $\omega$, and we let $P_\omega(\tmf{g})\subset P(\tmf{g})$ be the set of weights of $L(\omega)$.   Let $\{\tilde{\omega}_i\,|\,i\in\tilde{I}\}\subset P^+(\tmf{g})$ be fundamental weights of $\tmf{g}$, satisfying $\langle\bar{\alpha}^\vee_i,\tilde{\omega}_j\rangle=\delta_{ij}$ ($i,j\in\tilde{I}$).  The corresponding $\tmf{g}$-modules $L(\tilde{\omega}_i)$  ($i\in\tilde{I}$) are known as fundamental $\tmf{g}$-modules. For every $i\in\tilde{I}$ we fix a highest weight vector $v_i$ of $L(\tilde{\omega}_i)$. Let $\widetilde{B}=(\widetilde{B}_{ij})_{i,j\in\tilde{I}}$ be the incidence matrix of the Dynkin diagram of $\tmf{g}$, namely $\widetilde{B}_{ij}=2\delta_{ij}-\widetilde{C}_{ij}$ ($i,j\in\tilde{I}$). Note that $\widetilde{B}_{ij}\geq0$. Define the dominant weights
$$\eta_i=\sum_{j\in\tilde{I}}\widetilde{B}_{ij}\tilde{\omega}_j\qquad i\in\tilde{I}.$$
Then as proved in \cite{marava15} for every $i\in\tilde{I}$ the weight $\eta_i$ is a highest weight (of multiplicity one) of the following two $\tmf{g}-$modules:
$$\bigwedge^2 L(\tilde{\omega}_i),\qquad \bigotimes_{j\in\tilde{I}}L(\tilde{\omega}_j)^{\otimes\widetilde{B}_{ij}},$$ 
with highest weight vector given, respectively, by $\bar{f}_iv_i\wedge v_i$ and $\otimes_{j\in\tilde{I}}v_j^{\widetilde{B}_{ij}}$ ($i\in\tilde{I}$). By complete reducibility, for every $i\in\tilde{I}$ there exists a $\tmf{g}-$module $U_i$ such that $\bigwedge^2 L(\tilde{\omega}_i)=L(\eta_i)\oplus U_i$, and an homomorphism of Lie algebras
\beq\label{mtilde}
\widetilde{m}_i:\bigwedge^2 L(\tilde{\omega}_i)\to \bigotimes_{j\in\mr{I}}\bigotimes_{\ell=0}^{\langle j\rangle -1}L(\tilde{\omega}_{\sigma^\ell(j)})^{\otimes\widetilde{B}_{i\sigma^\ell(j)}}
\eeq
uniquely fixed by the conditions $\ker \widetilde{m}_i=U_i$ and $\widetilde{m}_i(\bar{f}_iv_i\wedge v_i)=\otimes_{j\in\tilde{I}}v_j^{\widetilde{B}_{ij}}.$ \\

We now consider twisted $\tmf{g}-$modules, and we extend the action of $\sigma$ on $\tmf{g}-$modules.
If $V$ is a $\tmf{g}$-module and $\Phi:\tmf{g}\to\End(V)$ the corresponding representation, then we define the twisted  representation as
\beq\label{220811-1}
\Phi^\sigma=\pi\circ\sigma^{-1}:\tmf{g}\to\End(V), 
\eeq
and we denote by $V^\sigma$ the vector space $V$ with the $\tmf{g}-$module structure induced by $\Phi^\sigma$.   We define the action of $\sigma$ on $\tmf{h}^\ast$ by
$$\langle h,\sigma\omega \rangle=\langle \sigma^{-1}h,\omega \rangle, \qquad h\in\tmf{h},\omega\in\tmf{h}^\ast,$$
so that in particular $\sigma\tilde{\omega}_i=\tilde{\omega}_{\sigma(i)}$ for every fundamental weight $\tilde{\omega}_i$ ($i\in\tilde{I}$).
The following lemma is an elementary extension of a result proved in \cite{marava17}:
\begin{lemma}\label{221111-1}
Let $L(\omega)$ be an irreducible $\tmf{g}-$module with highest weight $\omega$. Then $L(\omega)^\sigma$ is irreducible and there  exists an isomorphism of $\tmf{g}-$modules  $L(\omega)^\sigma\simeq L(\sigma\omega).$ In particular, for every $i\in \tilde{I}$ the  $\tmf{g}-$module $L(\tilde{\omega}_i)^\sigma$ is isomorphic to the fundamental $\tmf{g}-$module $L(\tilde{\omega}_{\sigma(i)})$.
\end{lemma}
\begin{proof}
It is easy to show that that a finite dimensional $\tmf{g}-$module $V$ is irreducible if and only if $V^\sigma$ is irreducible. If $v\in L(\omega)$ is an highest weight vector, namely $\Phi(\bar{e}_i)v=0$ and $\Phi(\bar{\alpha}^\vee_i)v=\langle\bar{\alpha}^\vee_i,\omega\rangle v$, then $\Phi^\sigma(\bar{e}_i)v=\Phi(\sigma^{-1}\bar{e}_i)v=\Phi(\bar{e}_{\sigma^{-1}(i)})v=0$ and $\Phi^\sigma(\bar{\alpha}^\vee_i)v=\Phi(\sigma^{-1}\bar{\alpha}^\vee_i)v=\langle\sigma^{-1}\bar{\alpha}^\vee_i,\omega\rangle v=\langle \bar{\alpha}^\vee_i,\sigma\omega\rangle v$.
\end{proof}
For every finite dimensional $\tmf{g}-$module $V$ we define a $\bb{C}-$linear bijective map, denoted $\sigma$ by abuse of notation
\beq\label{V->Vsigma}
\sigma:V\to V^\sigma,\qquad v\mapsto \sigma(v),
\eeq
satisfying
\beq\label{221111-2}
\sigma(\Phi(x)v)=\Phi^\sigma(\sigma x)\sigma(v),\qquad x\in\tmf{g}, v\in V.
\eeq
In particular, \eqref{V->Vsigma} maps highest weight vectors to highest weight vectors. In the case when $V=L(\tilde{\omega}_i)$ is the $i-$th fundamental $\tmf{g}-$module, using Lemma \ref{221111-1} we thus obtain, for each $i\in\tilde{I}$, a $\bb{C}-$linear map 
\beq\label{sigmaLi}
\sigma:L(\tilde{\omega}_i)\to  L(\tilde{\omega}_{\sigma(i)}),
\eeq
uniquely specified, up to a nonzero scalar multiple, by condition \eqref{221111-2}. If $v_i\in L(\tilde{\omega}_i)$ ($i\in\tilde{I}$) is a highest weight vector, then we can choose a normalization so that \eqref{221111-2} reads
\beq\label{221111-3}
\sigma(\Phi_i(x)v_i)=\Phi_{\sigma(i)}(\sigma x)v_{\sigma(i)},\qquad x\in\tmf{g}.
\eeq
Furthermore, if $\sigma(i)=i$ then the map $v\mapsto \sigma(v)$  is  an automorphism of $L(\tilde{\omega}_i)$, and \eqref{221111-3} reduces to 
$$\sigma(\Phi_i(x)v_i)=\Phi_{i}(\sigma x)v_{i},\qquad x\in\tmf{g}.$$
For every $i\in\tilde{I}$ we finally introduce the $\bb{C}-$linear map $R_i:L(\tilde{\omega}_i)\to L(\tilde{\omega}_i)$ as
\beq\label{22111-3}
R_i(v)=
\begin{cases}
v & i\neq\sigma(i),\\
\sigma(v) &i=\sigma(i).
\end{cases}
\eeq
These maps will be useful later to construct the so called $\Psi-$system for  $\mf{g}$.

\subsection{Cyclic elements of $\tmf{g}$}
Given the set $\tilde{I}$ of the vertices of the Dynkin diagram of $\tmf{g}$ we introduce a bipartition
\cite{moody87} of the form $\tilde{I}=\tilde{I}_1\cup\tilde{I}_2$ such that $1\in \tilde{I}_1$ and all edges of the Dynkin diagram of $\tmf{g}$ lead
from $\tilde{I}_1$ to $\tilde{I}_2$. Then, we define the function $p:\tilde{I}\longrightarrow \bb{Z}/2\mb Z$ as
$$p(i)=
\begin{cases}
0\quad i\in\tilde{I}_1,\\
1\quad i\in\tilde{I}_2.
\end{cases}
$$
Note that $p(\sigma(i))=p(i)$ ($i\in\tilde{I}$). 
\begin{table}[H]\label{table:ki}
\caption{The values of the scalars $\kappa_i$. Note that $\kappa_{\sigma(i)}=\kappa_i$ ($i\in \tilde{I}$).}
{\tabulinesep=1.2mm
\begin{tabu}{ |c|c|c| }
\hline
 $\wg$   & $r$  &   $\kappa_i$, \;$i\in\widetilde{I}$    \\
\hline
\hline
$ADE$  &  $1$ & $\kappa_i=\frac{p(i)}{2}$\\
\hline
$ A_{2n-1}$ & $2$ &
$\kappa_i=\frac{p(i)}{2}$\\
\hline
$ D_{n+1}$,\; $n$ even & $2$ & 
$\kappa_i=\frac{p(i)}{4}$, \;$1\leq i\leq n-1$,\qquad $\kappa_n=\kappa_{n+1}=\frac{1}{2}$ \\
\hline
$ D_{n+1}$,\; $n$ odd &  $2$ &
$\kappa_i=\frac{p(i)}{4}$,\; $1\leq i\leq n-1$,\qquad $\kappa_n=\kappa_{n+1}=-\frac{1}{4}$\\
\hline
$E_6$ &  $2$ & $\kappa_i=\frac{p(i)}{2}$,\; $i\neq 4$,\qquad $\kappa_4=\frac{1}{4}$\\
\hline
$D_4$  &  $3$ & $\kappa_i=\frac{p(i)}{2}$\\
\hline
\end{tabu}
}
\end{table}

 For $t\in\bb{R}$ introduce the following cyclic elements of $\tmf{g}$
\beq\label{lambdat}
\Lambda(t)=\mr{f}+ e^{2 \pi it}v_\theta,
\eeq
where $\mr{f}$ is given in \eqref{221026-2/0} and $v_\theta\in\tmf{g}$ is given in Definition \ref{221108-2/0}.
\begin{definition}\label{def:maximal}
Let $A$ be an endomorphism of a vector space $V$.
We say that a eigenvalue $\mu$ of $A$ is maximal if it is real,  its algebraic multiplicity
is one, and $\mu> \Re\mu'$ for every eigenvalue $\mu'\neq \mu$ of $A$.
\end{definition}
Recall that  $v_\theta\in\tmf{g}$  is defined up to a nonzero scalar multiple. 
\begin{theorem}\label{thm:genlamba}\cite[4.7]{marava17}
Let $\Lambda(t)$ be given by \eqref{lambdat} and let  $\kappa_i$ ($i\in\tilde{I}$) be defined  as in Table \ref{table:ki}.   We can always choose the element $v_\theta$ such that for all $i\in \mr{I}$, the following facts hold true:
\begin{enumerate}[(i)]
\item  For $\ell=0,\dots,\langle i\rangle-1$, the matrix representing the element
$$
\Lambda(\kappa_i-\tfrac{\ell}{r}) 
$$
in the fundamental $\tmf{g}-$module $L(\omega_{\sigma^\ell(i)})$
has a maximal eigenvalue
$\mu^{(\sigma^\ell(i))}$ and, in particular,  $\mu^{(1)}=1$.
 We denote by $\psi^{(\sigma^\ell(i))}$ the corresponding unique (up to a constant factor) eigenvector. Moreover, 
 $$\mu^{(\sigma^\ell(i))}=\mu^{(i)},\qquad \psi^{(\sigma^\ell(i))}=\sigma^\ell\left(\psi^{(i)}\right),$$
 where the action of $\sigma$ on $L(\tilde{\omega}_i)$ is defined in \eqref{sigmaLi}.
\item given $D_i$ as in \eqref{Di/0}, the matrix representing the element 
\beq\label{221111-4}
\Lambda(\kappa_i-\tfrac{D_i}{2}),\qquad i\in \mr{I}.
\eeq
 in the $\tmf{g}-$module $\bigwedge^2 L(\tilde{\omega}_i)$
has a maximal eigenvalue 
\beq\label{mui}
\left(e^{-\frac{\pi\sqrt{-1}D_i}{h}}+e^{\frac{\pi\sqrt{-1}D_i}{h}}\right)\lambda^{(i)}.
\eeq
The corresponding eigenvector is given by
\beq
\psi^{(i)}_\wedge= R_i\left(e^{\frac{\pi \sqrt{-1} D_i}{h}\mr{\rho}^\vee}\psi^{(i)}\right)\wedge e^{-\frac{\pi \sqrt{-1} D_i}{h}\mr{\rho}^\vee}\psi^{(i)},
\eeq
where $R_i$ is given by \eqref{22111-3}.
\item for the matrix representing the element \eqref{221111-4} in the $\tmf{g}-$module $\bigotimes_{j\in\tilde{I}}L(\tilde{\omega}_j)^{\otimes\bar{B}_{ij}}$
the scalar \eqref{mui} is maximal eigenvalue. The corresponding eigenvector is
\beq
\psi^{(i)}_\otimes=\bigotimes_{j\in \mr{I}}\bigotimes_{\ell=0}^{\langle j\rangle-1}\left(e^{\frac{\pi \sqrt{-1} (\kappa_j-\kappa_i+D_i/2-\ell/r)}{h}\mr{\rho}^\vee}\psi^{(\sigma^\ell(j))}\right)^{\otimes \widetilde{B}_{i\sigma^{\ell}(j)}}
\eeq
\item We can normalize the maps $\widetilde{m}_i$ ($i\in\mr{I}$) given in \eqref{mtilde} such that the following identity, known as \emph{algebraic $\Psi-$system}, holds:
\beq
\widetilde{m}_i(\psi^{(i)}_\wedge)=\psi^{(i)}_\otimes
\eeq
 \end{enumerate}
\end{theorem}
From now on, we will fix the element $v_\theta$ such that the above theorem holds.

\begin{table}[H]
\caption{}\label{table:affine}
{\tabulinesep=0.6mm 
\begin{tabu}{ |c||c|c|c|c|c|c|} 
\hline 
$\mf{g}'=\,^L\mf{g}=(^L\rmf{g})^{(1)}$ & $\bmf{g}$ & $\mf{g}=\tmf{g}^{(r)}$ & $\tmf{g}$   &  $\mathring{\mf{g}}$ & $\mbox{spec}(\theta^\vee)$ in $\tmf{g}$     \\ 
\hline 
\hline
&&&&&\\
$A^{(1)}_n$, $n\geq1$ & $A^{(1)}_n$ & $A^{(1)}_n$ & $A_n$   & $A_n$ & \begin{tabular}[x]{@{}c@{}}$\{-2,0,2\}$ if $n=1$\\ $\{-2,-1,0,1,2\}$ if $n>1$\end{tabular}    \\
&&&&&\\
\hline
&&&&&\\
$B^{(1)}_n$, $n\geq3$ & $D^{(1)}_{2n}$ & $A^{(2)}_{2n-1}$ & $A_{2n-1}$  & $C_n$ & $\{-2,-1,0,1,2\}$  \\
&&&&&\\
\hline
&&&&&\\
$C^{(1)}_n$, $n\geq2$ & $D^{(1)}_{n+2}$ & $D^{(2)}_{n+1}$ & $D_{n+1}$  & $B_n$ & $\{-2,0,2\}$  \\
&&&&&\\
\hline
&&&&&\\
$D^{(1)}_n$, $n\geq4$ & $D^{(1)}_n$ & $D^{(1)}_{n}$ & $D_n$  & $D_n$ & $\{-2,-1,0,1,2\}$   \\
&&&&&\\
\hline
&&&&&\\
$E^{(1)}_6$ & $E^{(1)}_{6}$  & $E^{(1)}_6$ & $E_6$  & $E_6$ & $\{-2,-1,0,1,2\}$   \\
&&&&&\\
\hline
&&&&&\\
$E^{(1)}_7$ & $E^{(1)}_{7}$   & $E^{(1)}_7$ &  $E_7$  & $E_7$ &  $\{-2,-1,0,1,2\}$    \\
&&&&&\\
\hline
&&&&&\\
$E^{(1)}_8$ & $E^{(1)}_{8}$ & $E^{(1)}_8$  & $E_8$ & $E_8$ & $\{-2,-1,0,1,2\}$     \\
&&&&&\\
\hline
&&&&&\\
$F^{(1)}_4$ & $E^{(1)}_{7}$ & $E^{(2)}_{6}$   & $E_{6}$   & $F_4$ &  $\{-2,-1,0,1,2\}$    \\
&&&&&\\
\hline
&&&&&\\
$G^{(1)}_2$ & $E^{(1)}_{6}$ & $D^{(3)}_{4}$  & $ D_{4}$    & $G_2$ & $\{-3,-2,-1,0,1,2,3\}$    \\
&&&&&\\
\hline
\end{tabu}
}
\end{table}


\section{FFH connections and the Bethe Ansatz}\label{opers and ba}
In this Section, fixed an affine Kac-Moody algebra $\mf{g}=\tmf{g}^{(r)}$ as in Table \ref{table:affine}, we develop the analytic theory of FFH connections and we show that their monodromy data provide solutions to the Bethe equations for the quantum $^L\mf{g}$-Drinfeld-Sokolov model.
The result of this section builds on our previous paper \cite{marava15,marava17,mara18}.
Here however, we give a complete and unified theory, and we fill important gaps in the literature.

\subsection{Feigin-Frenkel-Hernandez connections}
Let $z$ be a local coordinate over the punctured complex plane $\bb{C}^\times$. Let $\mf{g}=\tmf{g}^{(r)}$ be an affine Kac-Moody algebra as listed in Table \ref{table:affine}. Recall the elements $f,\scaling\in\mf{g}$ as well as $v_\theta\in\tmf{g}$ introduced in Definition \ref{221108-2/0}.  We consider the following family of meromorphic connections with values in  $\tmf{g}^{(1)}$:
\begin{align}\label{eq:ffopersourgauge}
 \mc{L}(z)= & \partial_z+\frac{1}{z}\left(f+\ell-k \scaling + zv_\theta+ \sum_{j \in J} \frac{ r z^{r}}{z^r-w_j^r}\left(-\theta^\vee +X(j)
 \right)\right).
 \end{align}
Here $J$ is a possibly empty finite set, and
the parameters $\ell, k$, and $(w_j,X(j),y(j))$, $j\in J$ satisfy the following requirements
\begin{itemize}
\item $\ell\in\rmf{h}$;
\item $0<k<1$;
 \item $w_j \in \mathbb{C}^*, j \in J$ and $\big(w_j/w_i\big)^r \neq 1$ if $i \neq j$;
 \item $X(j) \in \mathring{\mf{n}}^+,\,\, j \in J$.
\end{itemize}
The FFH connections will be later defined as the subclass of connections \eqref{eq:ffopersourgauge} with trivial monodromy at each $w_j$ $(j\in J)$, see Definition \ref{def:trivialmonodromy} below.
We remark that since $\scaling = -\la \partial_\la$, \eqref{eq:ffopersourgauge} can be also tought of as a first-order linear partial
differential operator
whose coefficients are meromorphic functions with values in $\tmf{g}$, namely
\begin{align*}
z \mc{L}=  z\,\partial_z+ k\; \la\, \partial_{\la} + & \mr{f}+\ell +  (z+\la )v_\theta + \sum_{j \in J} \frac{ r z^{r}}{z^r-w_j^r}\left(-\theta^\vee +X(j)
  \right)  .
 \end{align*}

\subsubsection*{Rotated and twisted connections.}
Given $t\in\bb{R}$, the \emph{rotated connection} is the connection induced by the map $z\mapsto e^{2\pi it}z$, that is:
\begin{align}
 \mc{L}_t(z)=  \partial_z&+\frac{1}{z}\Bigg(f+\ell-k \scaling + e^{2\pi it}zv_\theta+ \sum_{j \in J} \frac{ re^{2\pi irt} z^{r}}{e^{2\pi irt}z^r-w_j^r}\left(\alpha_0^\vee + X(j)
  \right)\Bigg).\label{rotated}
\end{align}
On the other hand, the twisted connection is the operator obtain from \eqref{eq:ffopersourgauge} by acting with the automorphism $\sigma$. Since all terms in \eqref{eq:ffopersourgauge} are fixed by $\sigma$ except for $v_\theta$ (which satisfies $\sigma(v_\theta)=\e v_\theta$), the \emph{twisted operator} is the operator
\beq\label{twisted}
 \mc{L}^\sigma(z)=  \partial_z+\frac{1}{z}\left(f+\ell-k \scaling +\e zv_\theta+ \sum_{j \in J} \frac{ r z^{r}}{z^r-w_j^r}\left(\alpha_0^\vee +X(j)
  \right)\right),
\eeq
where $\e=e^{\frac{2\pi i}{r}}$.
A direct computation then shows that the connections of the form \eqref{eq:ffopersourgauge} satisfy the identity
\beq\label{Lsigma1/r0}
 \mc{L}^\sigma_{-\frac{1}{r}}(z)=\mc{L}(z),
\eeq
which is crucial in the deduction of the Bethe Equations.

\subsubsection*{Loop realization}
It is often useful to consider a different realization of the connection \eqref{eq:ffopersourgauge}, taking values  in the loop algebra $\tmf{g}[\la,\la^{-1}]$ (or, more precisely, in the current algebra $\tmf{g}[\la]$). Given $\mc{L}(z)$ as in \eqref{eq:ffopersourgauge}, we then define
\begin{align}
\mc{L}(z;\la):=\partial_z+\frac{1}{z}\left(\mr{f}+\ell +(z+z^k\lambda)v_\theta+ \sum_{j \in J} \frac{ r z^{r}}{z^r-w_j^r}\left(-\theta^\vee +X(j)
 \right)\right).\label{looprealization}
\end{align}
and we call it  the \emph{loop realization} of \eqref{eq:ffopersourgauge}. 
At least formally, $\mc{L}(z;\la)$ is obtained from $\mc{L}$ by the action
of the Gauge $z^{k\scaling}$, namely 
$$\mc{L}(z;\la)=z^{k\scaling}\mc{L}(z).$$ 
The precise analytic and algebraic meaning of this Gauge will be elucidated, respectively, in Lemma \ref{lem:AisoB} and in \eqref{L(z,la)defalg}. While the connection $\mc{L}$ is a meromorphic connection over $\bb{C}$, its loop realization is not, since $z=0$ is a branch point of its coefficients. Given $\mc{L}(z;\la)$ as in \eqref{looprealization} and $t\in\bb{R}$, we define the rotated connection $ \mc{L}_t(z;\la)$ as
\begin{align}
 \mc{L}_t(z;\la)=  \partial_z+\frac{1}{z}&\Bigg(\mr{f}+\ell + (e^{2\pi it}z+z^k\la)v_\theta+ \sum_{j \in J} \frac{ re^{2\pi irt} z^{r}}{e^{2\pi irt}z^r-w_j^r}\left(-\theta^\vee + X(j)
 \right)\Bigg).\label{rotatedloop}
\end{align}
Note that the map $\mc{L}(z;\la) \mapsto \mc{L}_t(z;\la)$ is not induced by the map $z\mapsto e^{2\pi it}z$,
but rather by
\begin{equation}\label{eq:symanzik}
(z,\la)\mapsto (e^{2\pi it}z,e^{-2\pi ikt}\la),\qquad t\in\bb{R}.
\end{equation}
\subsubsection*{Dorey-Tateo symmetry}
Since the connection $\mc{L}$ \eqref{eq:ffopersourgauge} is meromorphic at $z=0$ then $\mc{L}_1(z)=\mc{L}(z)$.
This is equivalent to the following identity for the loop realisation of $\mc{L}$,
\beq\label{DT}
\mc{L}_1(z;\la)=\mc{L}(z;\la).
\eeq
In the context of the ODE/IM correspondence, the identity \eqref{DT} is known as
 \emph{Dorey-Tateo symmetry} or \emph{Symanzik rescaling}. We have just shown that the Dorey-Tateo symmetry is nothing but the \lq loop counterpart\rq of the fact that the connection (\ref{eq:ffopersourgauge}) is meromorphic at $z=0$.

\begin{remark}\label{rem:xtoz} In what follows, we will import results from our previous papers \cite{marava15,marava17,mara18}, where we studied the loop realization (\ref{looprealization}) of connections of the form (\ref{eq:ffopersourgauge}) respectively in the cases $r=1$ and $J=\emptyset$ (simply-laced, ground state), $r>1$ and $J=\emptyset$ (non simply-laced, ground state), and $r=1$
and $J \neq \emptyset $ (simply laced, higher states). We note here that in \cite{marava15,marava17} we
used a different coordinate system on $\widetilde{\mathbb{C}^*}$.
In fact, while in the case  $J = \emptyset$, the connection (\ref{looprealization}) reads
\beq\label{LG(z;la)}
\mc{L}(z;\la)=\partial_z+\frac{1}{z}\left(\mr{f}+\ell+\left(z+z^{k}\la\right)v_\theta\right),
\eeq
in \cite{marava15,marava17}, we considered connections of the form
\beq\label{L(x;E)}
L(x;E)=\partial_x+\mr{f}+\frac{l}{x}+(x^{Mh}-E)v_\theta,
\eeq
with $l\in\mr{\mf{h}}$, $M>0$, and $E$  a complex parameter.
With the following change of coordinates and parameters \cite{FF11}
\begin{align*}
 & x=\varphi(z)=(h(M+1))^{\frac{1}{M+1}}z^{\frac{1}{h(M+1)}}, \\
&\ell=\frac{1}{h(M+1)}(l+\mr{\rho}^\vee), \; k=\frac{1}{M+1},
\;\la=-\left(\frac{1}{h(M+1)}\right)^{\frac{hM}{M+1}}E,
\end{align*}
one readily verifies that
$$\mc{L}(z;\la)=z^{\ad\rmf{\rho}^\vee}\varphi'(z)^{\ad\rmf{\rho}^\vee}(\varphi^\ast L)(z;E).$$
\end{remark}

\subsection{The space of global solutions}
We denote by $\widetilde{\bb{C}}^\times$
the universal cover of $\bb{C}^\times$ and let $\Pi$ be the corresponding canonical projection.
For every finite dimensional $\widetilde{\mf{g}}$ module $V$ we define
\beq\label{Vla}
V(\la)=V\otimes \mc{O}_{\la},
\eeq
where $\mc{O}_\la$ is the ring of entire functions in the variable $\lambda$.
We let the differential operator \eqref{eq:ffopersourgauge} act as a (meromorphic) connection of the trivial bundle $ \bb{C}^* \times V(\la)$. We are thus led to consider the
ODE
\begin{equation}\label{eq:psiVla}
 \mc{L}(z)\Psi(z)=0,  \qquad \Psi:\widetilde{\bb{C}}^\times_J \to V(\la).
\end{equation}
where
\begin{equation}\label{eq:C*J}
 \widetilde{\bb{C}}^\times_J:=\left(\widetilde{\bb{C}}^\times \setminus \bigcup_{j \in J} \bigcup_{l=0}^{r-1}\Pi^{-1} (\e^l w_j)\right),
\end{equation}
and, by abuse of notation, we will denote from now on by $z$ both the global co-ordinate on $\bb{C}^\times$ the local co-ordinate on  $\widetilde{\bb{C}}^\times$.

\begin{definition}
Given a connection $\mc{L}$ of the form \eqref{eq:ffopersourgauge}
and a finite dimensional $\widetilde{\mf{g}}$ module $V$,
we say that $\Psi:\widetilde{\bb{C}}^\times_J \to V(\la)$ is a global solution (or simply, a solution) if $\mc{L}(z)\Psi(z)=0$, for all $z \in \widetilde{\bb{C}}^\times_J$.
We denote by $\mc{A}_V$ the $\bb{C}$-vector space of global solutions of the ODE (\ref{eq:psiVla}).
\end{definition}
\begin{definition}
Let $D \subset \widetilde{\bb{C}}^\times_J$ be open and simply-connected. We say that $\psi: D \to V(\la)$ is a local solution if $\mc{L}(z)\psi(z)=0$ for all $z \in D$.
\end{definition}

Before addressing the study of global solutions,
we remark that a solution $\psi:D \to V(\la)$ is also an analytic function of two variables with domain
  $D \times \bb{C}$: given $\psi$, we denote
  by $\psi(z;\la) \in V$ the evaluation at the point $\la$ of $\psi(z)$.
  This view-point allows us to make a bridge between a FFH connection and its loop realisation. We have in fact the following
  \begin{lemma}\label{lem:loopcauchy}
  Let $\mc{L}$ a connection of the form \eqref{eq:ffopersourgauge}, $\mc{L}(z,\la)$ ist loop realisation, $D \subset \widetilde{\bb{C}}^\times_J$ be open and simply-connected,  and $V$ a finite dimensional $\widetilde{\mf{g}}$.
  The function $\psi:D \to V(\la)$ is a local solution of $\mc{L}\psi=0$ if and only if the function $ \widetilde{\psi}: D \times \bb{C} \to V$, with $\widetilde{\psi}(z;\la)=\psi(z)(\la z^{-k})$, satisfies the differential equation
  \begin{equation}\label{eq:loopcauchy}
   \mc{L}(z;\la)\widetilde{\psi}(z,\la)=0, \; \psi(z_0;\la)= g(\la z^{-k}),
  \end{equation}
   for all $(z,\la) \in D\times \bb{C}$, and it is analytic with respect to the parameter $\la$.
  \end{lemma}
  \begin{proof}
   It follows from a direct computation.
  \end{proof}

\begin{proposition}\label{prop:localsolution}
Let $V$ be a finite dimensional $\widetilde{\mf{g}}$-module, $D\subset \widetilde{\bb{C}}^\times_J$ open and
simply-connected, $z_0 \in D$, and $g \in V(\la)$. The Cauchy problem
\begin{equation}\label{eq:localLpsi}
 \mc{L}(z)\Psi(z)=0, \qquad   \Psi(z_0)=g, z \in D.
\end{equation}
where $\mc{L}(z)$ is given by \eqref{eq:ffopersourgauge}, admits a unique solution $\Psi:D \to V(\la)$.
\begin{proof}
Due to Lemma \ref{lem:loopcauchy}, the above thesis is equivalent to the statement that for every $\la\in \bb{C}$ the Cauchy problem \eqref{eq:loopcauchy} admits a unique solution $\widetilde{\psi}(\cdot;\la): D \to V$, and that such a solution depends analytically on $\la$. This is a standard result in the theory of linear ODEs in the complex plane; see e.g. \cite[Theorem 1.1]{Ilya08}
\end{proof}

\end{proposition}

\begin{definition}\label{def:trivialmonodromy}
 We say that the the connection (\ref{eq:ffopersourgauge}) has trivial monodromy or is monodromy-free
 if for any finite dimensional $\tmf{g}$-module $V$ any local solution extends to a global solution.
\end{definition}
\begin{definition}\label{def:FFH}
 We call Feigin-Frenkel-Hernandez (FFH) connections the class of connections of the form (\ref{eq:ffopersourgauge}) which have trivial monodromy.
\end{definition}

\begin{definition}
We denote by
$\mc{O}'$ the sub-ring of the ring of holomorphic function
on $\widetilde{\bb{C}}^\times \times \bb{C} \ni (z;\la)$ given by
\begin{equation}
 \mc{O}'=\lbrace f(z;\la)= Q(z^{-k}\la), \mbox{ for some entire function } Q: \bb{C} \to \bb{C} \rbrace.
\end{equation}
\end{definition}
\begin{theorem}\label{prop:AVfree}
The connection $\mc{L}$ of the form (\ref{eq:ffopersourgauge}) has trivial monodromy if and only if for every finite $\tmf{g}$-module $V$ the space $\mc{A}_V$ of global solution
is a free $\mc{O}'$-module of rank $\dim V$ (i.e. $\mc{A}_V \cong V \otimes \mc{O}'$ as an $\mc{O}'$ module).
\end{theorem}

\begin{proof}
 The fact that $\mc{A}_V$ is an $\mc{O}'$-module follows by the fact that $\partial_z -k\scaling/z$ annihilates every function in $\mc{O}'$.

 Fixed a basis $\lbrace v_1, \dots, v_{\dim V} \rbrace$ of $V$ and a point $z_0 \in \widetilde{\bb{C}}^\times_J$, we
 define the local solutions
 $\psi_i(z)$ via the Cauchy problem $\psi(z_0)=v_i$, $ i=1 \dots \dim V $, where we consider $v_i\in V(\lambda)$ as a constant function (in fact, by Proposition \ref{prop:localsolution}, the local Cauchy problem is well-posed).

 Assume that $\mc{L}$ has trivial monodromy. By hypothesis  the solutions $\lbrace \psi_1, \dots, \psi_{\dim V} \rbrace$ extend to global solutions and, as we prove below, they form an $\mc{O}'$ basis of $\mc{A}_V$. Let in fact
 $\psi \in \mc{A}_V$. By hypothesis $\psi(z_0) \in V(\la)$, namely $\psi(z_0)=\sum_i g_i(\la) v_i$ for some
 $g_i \in \mc{O}_{\la}$; hence, by the well-posedness of the local Cauchy problem,
 $\psi(z)=\sum_{i} g_i(\la z^{-k} z_0^{k}) \psi_i(z)$. This prove the \textit{only if} part of the thesis.

 Now assume that $V$ is such that $\mc{A}_V$ is a free-module of rank $\dim V$. We show that in this
 case any local solution extends to a global solution. Let in fact $\psi$ be the solution of the
 local Cauchy problem $\mc{L}\psi=0, \psi(z_0)=g(\la)$, for an arbitrary pair $(z_0,g(\la))$.
 Fixed a basis $\lbrace \varphi_1, \dots, \varphi_{\dim V}\rbrace$ of global solutions,
 the linear system  $\sum_i  g_i(\la) \varphi_i(z_0) = g(\la)$ admits a unique solution. Therefore
 locally $\psi(z)=\sum_{i} g_i(\la z^{-k} z_0^{k}) \varphi_i(z)$, whence it extends to a global solution. This proves the \textit{if} part of the thesis.
 \end{proof}
\begin{definition}
 Let $\mc{L} (z;\la)$ be the loop realization of $\mc{L}(z)$ as per (\ref{looprealization}).
 We denote by $\mc{B}_V$ the $\bb{C}$-vector space of analytic functions
 $$\psi: \widetilde{\bb{C}}^\times_J\to V$$
 such that
 $$\mc{L}(z;\la)\psi(z;\la)=0,\qquad \forall (z,\la) \in  \widetilde{\bb{C}}^\times_J \times \bb{C} .$$
 \end{definition}
\begin{lemma}\label{lem:AisoB}
1. The space $\mc{B}_V$ is an $\mc{O}_{\la}-$module and the map
 \begin{equation}
  z^{-k \scaling}: \mc{A}_V \to \mc{B}_V , \qquad \psi \mapsto \widetilde{\psi},
 \end{equation}
where $ \widetilde{\psi}(z;\la)=\psi(z)(\la z^{-k})$, is an isomorphism of $\bb{C}-$vector spaces.

2. The following properties are equivalent
\begin{itemize}
 \item[i)] $\mc{L}(z)$ has trivial monodromy.
 \item[ii)] For every finite dimensional $\tmf{g}$-module $V$, the space $\mc{B}_V$ is a free $\mc{O}_{\la}-$module
 of rank $\dim V$ (i.e. $\mc{B}_V\cong V(\la)$ as $\mc{O}_{\la}-$modules).
 \item[iii)] For every finite dimensional $\tmf{g}$-module $V$, the loop realisation $\mc{L}(z;\la)$ of $\mc{L}(z)$  has trivial monodromy at the singular points $\e^l w_j, j \in J$, $l=1,\dots,r$  for every $\la \in \mathbb{C}$.
 \item[iv)] For every finite dimensional $\tmf{g}$-module $V$, the loop realisation $\mc{L}(z;\la)$ of $\mc{L}(z)$  has trivial monodromy at the singular points $w_j, j \in J$, for every $\la \in \mathbb{C}$.
\end{itemize}
\begin{proof}
Part 1 follows directly from Lemma \ref{lem:loopcauchy}.

Part 1. and Proposition \ref{prop:AVfree} imply the equivalence between property $i)$ and property $ii)$.
The equivalence between property $ii)$ and property $iii)$ is
proved in \cite[Lemma 5.4]{mara18} (for the case $r=1$). The same proof applies verbatim to the case $r>1$.
The equivalence between (iii) and (iv) is a consequence of the identity \eqref{Lsigma1/r0}, namely
$\mc{L}^{\sigma}_{-\frac1r}=\mc{L}$:
The loop realisation of $\mc{L}$ has trivial monodromy at $w_j$ if and only if the loop realisation of $\mc{L}^{\sigma}$ has trivial monodromy at $w_j$ -- as indeed the action of $\mc{L}^{\sigma}$ on $V(\la)$ coincide with the action of $\mc{L}$ on $V^{\sigma}$. Using \eqref{Lsigma1/r0}, we deduce that the loop realisation of $\mc{L}$ has trivial monodromy at $w_j$ if and only if it has trivial at $\e w_j$. This concludes the proof.
\end{proof}

\end{lemma}

\subsubsection{Rotated and twisted solutions}
Let $\mc{L}$ be a Feigin-Frenkel-Hernandez connection, and for $t\in\bb{R}$ consider the  rotated connection $\mc{L}_t$, given by \eqref{rotated}.
Fixed $V(\la)$ as in \eqref{Vla} consider the ODE
$$ \mc{L}_t(z)\Psi(z)=0,  \quad \Psi:\left(\widetilde{\bb{C}}^\times \setminus
\bigcup_{j \in J} \bigcup_{l=0}^{r-1}\Pi^{-1} (e^{-2\pi i t} \e^l w_j)\right) \to V(\la),\; \e=e^{\frac{2\pi i}{r}}$$
and denote by $\mc{A}_{V,t}$ be the  corresponding space of solutions. Denoting
$$\psi_t(z)= \psi(e^{2\pi i t}z),$$
then the map
$$\mc{A}_V \to \mc{A}_{V,t}, \qquad \psi \mapsto \psi_t,$$
is an isomorphism of $\bb{C}$-vector spaces.
Moreover, since $\mc{L}$ is single-valued, for every $t\in \bb{R}$ the spaces $\mc{A}_{V,t}$ and $\mc{A}_{V,t+1}$ actually coincide.
Therefore there is a $\bb{C}$-automorphism $\mc{M}$ of $\mc{A}_{V,t}$, called the monodromy
(or monodromy operator) and defined by
\begin{equation}\label{monodromy}
\mc{M}: \mc{A}_{V,t} \to \mc{A}_{V,t} ,\qquad\psi \mapsto \psi_{1} .
\end{equation}
If $L(\tilde{\omega}_i)$ ($i \in \widetilde{I}$) is the $i$-th fundamental $\widetilde{\mf{g}}$-module, we denote
\begin{equation}
\mc{A}^{(i)}:=\mc{A}_{L(\tilde{\omega}_i),\kappa_i}, \quad \mc{A}^{(i)}_t=\mc{A}_{L(\tilde{\omega}_i),\kappa_i+t}
\end{equation}
 where $\kappa_i$'s are the rational numbers defined in Table \ref{table:ki}. 
 
Now recall the algebra $\widetilde{\mf{g}}$ carries a Dynkin automorphism $\sigma$ and that for any representation $V$ we have defined the  linear isomorphism \eqref{V->Vsigma} from $V$ to the twisted module $V^{\sigma}$. Since FFH connections  satisfy \eqref{Lsigma1/r0}, we have an isomorphism between the spaces of solutions
$\mc{A}_V$ and $\mc{A}_{V^\sigma,-\frac{1}{r}}$. More precisely, we have
\begin{lemma}
1. For every finite dimensional module $V$ and every $t\in\bb{R}$, the map \eqref{V->Vsigma} induces a $\mc{O}'-$isomorphism between $\mc{A}_{V,t}$ and $\mc{A}_{V^{\sigma},t+\frac{1}{r}}$.
In particular, for every $i \in \tilde{I}$,
\begin{align}
 & \Psi\in\mc{A}^{(i)}_t  \Longleftrightarrow \sigma(\Psi) \in \mc{A}^{(\sigma(i))}_{t+\frac{1}{r}} \label{eq:twistedandrotatedA}\\
 & \Psi \in  \mc{A}^{(i)}_t \Longleftrightarrow  R_i(\Psi) \in  \mc{A}^{(i)}_{t+ D_i} 
\end{align}
In the latter equation $R_i: L(\tilde{\omega}_i) \to L(\tilde{\omega}_i)$ is given by \eqref{22111-3} and the coefficients   $D_i$ by \eqref{Di/0}.

2. For $i\in \mr{I}$ and $t\in\bb{R}$ the map $\widetilde{m}_i$ defined in \eqref{mtilde} induces a  $\bb{C}-$linear map
$$
\widetilde{m}_i:\bigwedge^2 \mc{A}^{(i)}_t\to \bigotimes_{j\in\mr{I}}\bigotimes_{\ell=0}^{\langle j\rangle -1}\left(\mc{A}^{\sigma^\ell(j)}_{\kappa_i-\kappa_j+t}\right)^{\otimes\bar{B}_{i\sigma^\ell(j)}}
$$
Using \eqref{eq:twistedandrotatedA} there exists a $\bb{C}-$linear map
\begin{equation}\label{eq:laverami}
  m_i:\bigwedge^2 \mc{A}^{(i)}_{\frac{D_i}{2}}\longrightarrow \bigotimes_{j\in \mr{I}}
  \bigotimes_{l=0}^{B_{ij}-1}\mc{A}^{(j)}_{\frac{B_{ij}-1-2l}{2r}}\,, \qquad i\in \mr{I},
 \end{equation}
which has the property that 
\begin{equation}\label{eq:mi}
 m_i(v_i \wedge f_iv_i) = \otimes_{j \in \mr{I}} v^{\otimes B_{ij}}_j,
\end{equation}
where $v_j$ is a highest weight vector of $L(\tilde{\omega}_j)$.
\begin{proof} 
1. is a straightforward check and 2. is proved in  \cite[Proposition 4.5]{marava17}.
\end{proof}
\end{lemma}

\subsubsection{Frobenius solutions: a basis of monodromy eigenvectors}
We discuss Frobenius-like solutions of the ODE (\ref{eq:psiVla}). From these solutions, under some genericity conditions on the parameters $\ell, k$  we will construct a basis of eigenvectors for the monodromy operator \eqref{monodromy}. We closely follow  \cite[Section 5.2]{mara18}, to which we refer for further details and proofs.

A Frobenius solution will be a global solution which admits at zero the convergent series expansion
\begin{equation*}
 \Psi(z)=z^{-\gamma} \sum_{m,n} c_{m} z^m \, g_m, \quad g_m \in V(\la),
\end{equation*}
for some $\gamma \in \bb{C}$.
We write the above sereis more conveniently as
\begin{equation}\label{eq:frobeniuseries}
 \Psi(z)(\la)=z^{-\gamma} \sum_{m,n} c_{m,n} z^m \la^n, \;  c_{m,n} \in V.
\end{equation}
Inserting the above ansatz in the differential equation $\mc{L}\psi=0$, the coefficients $c_{m,n}$ are seen to necessarily
satisfy the following recurrence
\begin{equation}
 \left( \mathring{f}+\ell-\gamma+m+k n \right) c_{m,n} + v_\theta c_{m,n-1}+ \sum_{l=1}^m A_l c_{m-l,n}=0
\end{equation}
for some $A_l \in \widetilde{\mf{g}}$. In order to solve the recurrence, it is natural to impose the following two conditions:
\begin{enumerate}
 \item $\gamma$ is an eigenvalue of $\mathring{f}+\ell$ in the representation $V$; 
 \item $\mathring{f}+\ell-\gamma+m+k n $ is invertible for every $m,n \geq 0, (m,n) \neq (0,0)$.
\end{enumerate}

Regarding condition (1), we notice that in any finite dimensional representation
the  spectra  of $\mathring{f}+\ell$ and of $\ell$ coincide, and that the condition that $\mathring{f}+\ell$ is semisimple
is generic in $\ell$.
Regarding condition (2), we notice that in any finite dimensional representation $V$, fixed $k$,
the condition (2) is  generic in $\ell$.
Therefore, we make the following definition.
\begin{definition}\label{genericpair}
For $i\in\mr{I}$ let $P(\tilde{\omega}_i)$ be the multi-set of weights -- with multiplicities -- of the fundamental representation $L(\tilde{\omega}_i)$. FIxed $k \in (0,1)$, $\ell\in \mathring{\mf{h}}$ is said to be generic if
$f+\ell$ is semisimple and if
for every $i\in\mr{I}$ and every $\omega \in P(\tilde{\omega}_i)$ the element
$$\mathring{f}+\ell-\omega(\ell)+m+k n $$
 is invertible  in $\End(L(\tilde{\omega}_i))$ for every $m,n \geq 0, (m,n) \neq (0,0)$.
\end{definition}

\begin{proposition}\label{prop:frobeniusolutions}
Fixed $k$ and let $\ell$ be generic as in Definition \ref{genericpair}. For any $i\in\mr{I}$ and any $\omega \in P(\tilde{\omega}_i)$ choose an eigenvector $\chi_{\omega}$ of $f+\ell$ with eigenvalue $\omega(\ell)$.

1. There exists a unique solution $\chi_{\omega} \in \mc{A}^{(i)}$, which at $0$ admits the Frobenius expansion
\begin{equation}\label{eq:chiomega}
 \chi_{\omega}(z,\la) = z^{-\omega(\ell)} \left( \chi_{\omega} + \sum_{(m,n) \neq (0,0)} c_{m,n} z^m \la^n \right),
\end{equation}
convergent in a neighborhood of $z=0$.\\
2. For every $i\in\mr{I}$, the collection of all Frobenius solutions
$\lbrace \chi_{\omega}(z)\rbrace_{\omega \in P(\tilde{\omega}_i)}$ is an
 $\mc{O}'$ basis of $\mc{A}^{(i)}$.\\
3. If $w$ in an element of the Weyl group $\mc{W}$ of $\rmf{g}$, we denote by $\chi^{(i)}_{w},\widetilde{\chi}^{(i)}_{w}\in\mc{A}^{(i)}$ the solutions corresponding respectively to the weights $w(\tilde{\omega}_i)$ and $w(\tilde{\omega}_i-\alpha_i)$.
We can find  a normalisation of these solutions such that they satisfy the following set of relations
\begin{align}\nonumber
 & m_i\left( R_i \big(\chi^{(i)}_{w,-\frac{D_i}{2}}\big) \wedge \widetilde{\chi}^{(i)}_{w,\frac{D_i}{2}} \right) =
 +e^{D_iw[\alpha_i] (\ell)} \otimes_{j\in \mr{I}}
  \otimes_{l=0}^{B_{ij}-1}\chi^{(j)}_{w,\frac{B_{ij}-1-2l}{2r}} \\ \label{eq:chisystem}
& m_i\left( R_i \big( \widetilde{\chi}^{(i)}_{w,\frac{-D_i}{2}} \big) \wedge \chi^{(i)}_{w,\frac{D_i}{2}} \right) = 
 -e^{-D_iw[\alpha_i] (\ell)} \otimes_{j\in \mr{I}}
  \otimes_{l=0}^{B_{ij}-1}\chi^{(j)}_{w,\frac{B_{ij}-1-2l}{2r}}
\end{align}
where $m_i$ is the linear map defined in (\ref{eq:mi}).

\begin{proof}
1. In \cite[Proposition 5.10]{mara18}, it is proven that
the Frobenius series \eqref{eq:frobeniuseries} and that its analytic continuation belongs to
$\mc{A}^{(i)}$, in the case $r=1$. The same proof applies to the general case.

The fact that $\chi_{\omega}$ is an eigenvalue of $\mc{M}$ is straightforward
\begin{equation}\label{eq:Mchi}
\mc{M}(\chi^{\omega})(z)=e^{-2\pi i \omega(\ell)}z^{-\omega(\ell)} \left( \chi_{\omega} + \sum_{(m,n) \neq (0,0)} c_{m,n} z^m \la^n \right) =
e^{-2\pi i \omega(\ell)} \chi^{\omega}(z) .
\end{equation}
2. The Frobenius solutions are a $\mc{O}'-$basis of $\mc{A}^{(i)}$ since by hypothesis $\lbrace \chi_{\omega}\rbrace_{\omega \in P(\tilde{\omega}_i)}$ is a basis
of $L(\tilde{\omega_i})$. \\
3. It follows directly from (\ref{eq:mi}) and (\ref{eq:chiomega}).
\end{proof}

\end{proposition}

\subsubsection{A basis at infinity. Subdominant solutions}
Here we study the asymptotic behaviour at $\infty$ of solutions in $\mc{A}^{(i)}$:
\begin{itemize}
 \item We provide the existence of distinguished bases, with prescribed exponential asymptotic behaviour
 in sectors of $\widetilde{\bb{C}^{\times}}$ of amplitude at least $\pi h$, with $h$ the Coxeter number of $\mf{g}$.
\item For every such a basis, we select a distinguished solution for its subdominant behaviour, and
we show that its asymptotic behaviour holds on a sector of amplitude at least $2 \pi h$.
\item We define the central connection matrix and the Stokes matrix, which, under the ODE/IM correspondence,
are identified with the $Q$ and $T$ operator-valued functions of the quantum Drinfeld-Sokolov model.
\end{itemize}

Our results follow rather directly from the following
\begin{theorem}[\cite{CGM23}]\label{thm:CGM}
 Let $V$ a finite dimensional vector space and, for $M$ positive, let $\widetilde{\mathbb{C}}^*_M=
 \lbrace z \in \widetilde{\mathbb{C}^*} , |z| > M \rbrace.$  Consider the differential equation
 \begin{equation}\label{eq:irregularnormal}
  \psi'(z)= \left(- A \, p(z;\la) + R(z;\la) \right) \psi(z), \quad \psi: \widetilde{\mathbb{C}}^*_M  \to V
 \end{equation}
 where $M$ is some positive number, and let
 \begin{itemize}
  \item  $A \in End(V)$ be a diagonalisable matrix with eigenvectors $\psi_j$
 and eigenvalues $\nu_j$, $j=1,\dots, \dim V$.
 \item $p(z;\la)=  z^{\sigma_0}+\sum_{n=1}^{N} c_n(\la)  z^{\sigma_n}$ where the exponents $\sigma_n$'s are real and ordered so that
 $$\sigma_0>\sigma_1 >\dots > \sigma_{N-1}>\sigma_{N} =-1.$$
 Moreover, the coefficients $c_n(\la)$ are analytic bounded functions of the parameter $\la$, which belongs to a domain $D \subset \mathbb{C}$. Given such a $p(z;\la)$ we define its primitive
$$P(z;\la)= \frac{z^{\sigma_0+1}}{\sigma_0+1} +\sum_{n=1}^{N-1} c_n(\la) \frac{z^{\sigma_n+1}}{\sigma_n+1}+ c_N \log z  .$$
 \item $R(z;\la)$ be a matrix valued function such that in an arbitrary closed sector of $\widetilde{\mathbb{C}}^*_M$ $|R(z;\la)|=O(z^{-1-\delta})$ uniformly with respect to $\la \in D$.
\end{itemize}
 1. Assume that the interval $[a,b]$ is such that the following condition holds:
 \begin{equation}\label{eq:nolines}
  \forall \varphi \in [a,b], \quad
  \Re (\nu_j e^{i\varphi})= \Re (\nu_{j'} e^{i\varphi}) \mbox{ if and only if } \nu_j=\nu_{j'}.
 \end{equation}
There exists a unique basis of solutions $\psi_j(z;\la)$ satisfying
the following asymptotics
\begin{align} \nonumber
 \psi_j(z,\la) & = \left(\psi_j + O(z^{-\delta}) \right) e^{-\nu_j P(z;\la)},\\ \label{eq:basesinfinity}
 & \mbox{as } z \to \infty \mbox{ in the sector }  \frac{a}{\sigma_0+1}\leq \arg z \leq\frac{b+\pi}{\sigma_0+1},
\end{align}
uniformly with respect to $\la \in D$.
Moreover, the functions $\psi_j(\cdot,\cdot)$'s are
analytic in $\widetilde{\mathbb{C}}^*_M \times D$.

2. Assume that $j_0$ is a subdominant index for the interval $[a,b]$, namely
\begin{equation}\label{eq:nosubdominantlines}
  \forall \varphi \in [a,b] \mbox{ and } j \neq j_0, \quad
  \Re (\nu_{j_0} e^{i\varphi}) > \Re (\nu_j e^{i\varphi}) .
 \end{equation}
There exists a unique solution $\Psi(z;\la)$, called subdominant, such that
\begin{align}\nonumber
 \Psi(z;\la) & =   \left(\psi_{j_0} + O(z^{-\delta}) \right) e^{-\nu_{j_0} P(z;\la)}, \\ \label{eq:subdominantgeneral}
 & \mbox{as } z \to \infty, \mbox{ in the sector }
 \frac{a-\pi}{\sigma_0+1} \leq \arg z
 \leq \frac{b + \pi}{\sigma_0+1},
\end{align}
uniformly with respect to $\la \in D$. Moreover, the function $\Psi(\cdot,\cdot)$ is
analytic in $\widetilde{\mathbb{C}}^*_M \times D$.
\end{theorem}
\begin{proof}
 See \cite[Theorem 3.21]{CGM23}.
\end{proof}

\begin{remark}
 Note that the interval $[a,b]$ satisfies (\ref{eq:nolines}) if and only if  $[a+m \pi,b+m \pi]$ satisfies the same
 property for every $m \in \mathbb{Z}$. Similarly,  the interval $[a,b]$ satisfies (\ref{eq:nosubdominantlines}) if and only if  $[a+2m \pi,b+2m \pi]$ satisfies the same
 property for every $m \in \mathbb{Z}$.
\end{remark}
We will also need the definition of a Stokes matrix
\begin{definition}\label{def:consecutive}
 Let $\psi_j$'s and $\psi'_{j'}$'s be bases of solution of the linear differential equation (\ref{eq:irregularnormal}), satisfying
 the asymptotic behaviour (\ref{eq:basesinfinity}) on the sectors  $\Sigma=\lbrace \frac{a}{\sigma_0+1}\leq \arg z \leq\frac{b+\pi}{\sigma_0+1}\rbrace$ and $\Sigma'=\lbrace \frac{a'}{\sigma_0+1}\leq \arg z \leq\frac{b'+\pi}{\sigma_0+1}\rbrace$. The bases
 $\psi_j$'s and $\psi'_j$'s are said to be consecutive bases if $\Sigma \cap \mc{S'}= \lbrace \frac{c}{\sigma_0+1} \leq \arg z \leq \frac{d}{\sigma_0+1} \rbrace$
 with $c<d$ such that condition (\ref{eq:nolines}) holds for all $\varphi \in [c,d]$.
The matrix $\mc{S}$ of change of basis,
\begin{equation}\label{eq:stokesm}
 \psi'_{j'}=\sum_{j=1}^{\dim V} \psi_{j} \mc{S}_{j,j'} , \; j' =1,\dots, \dim V,
\end{equation}
is called a Stokes matrix. 
 \end{definition}
 An important property of a Stokes matrix is that it is unipotent. In fact,
 we have the following lemma.
 \begin{lemma}\label{lem:Stokes}
  With the notation of Definition \ref{def:consecutive}. Assume that $\psi_j$'s and $\psi'_{j'}$'s are consecutive bases and the intersection of their respective sectors is
  $\lbrace \frac{c}{\sigma_0+1} \leq \arg z \leq \frac{d}{\sigma_0+1} \rbrace$ for some $c<d$.
  Then
  \begin{align}\label{eq:unipotent1}
 &  \mc{S}_{j,j}=1 , \; j=1,\dots, \dim V , \\ \label{eq:unipotent2}
 &  \mc{S}_{j,j'}=0, \mbox{ if } j\neq j' \mbox{ and } Re (\nu_{j} e^{i\varphi}) \leq \Re (\nu_{j'} e^{i\varphi}), \; \forall \varphi \in [c,d].
  \end{align}
   We remark that, due to (\ref{eq:nolines}), if the condition $Re (\nu_{j} e^{i\varphi}) \leq \Re (\nu_{j'} e^{i\varphi})$ holds for a $\varphi \in [c,d]$ then it holds for every $\varphi \in [c,d]$.
  \begin{proof}
  By definition of $\Sigma$ and by Theorem \ref{thm:genlamba},
\begin{align*}
\left( \psi_{j'} + O(z^{-\delta})\right) e^{-\nu_{j'}P(z;\la)}=& \sum_{j=1}^{\dim V} \left( \psi_{j} + O(z^{-\delta})\right)  e^{-\nu_{j}P(z;\la)}\mc{S}_{j,j'}, \\
& z \to \infty \mbox{ and }  \frac{c}{\sigma_0+1}\leq\arg z \leq \frac{d}{\sigma_0+1}.
\end{align*}
Fixed a $c>0$ and a $\varphi \in [c,d]$, we let
\begin{align*}
 l_{\varphi,c} & = \lbrace z, \Im \big(z^{\sigma_0+1}e^{-i\varphi}\big)= c\rbrace  \\
& = \lbrace |z|^{\sigma_0+1} \sin \big(-\varphi+(\sigma_0+1) \arg z\big) =c , \; \frac{\varphi}{\sigma_0+1}<\arg z <\frac{\varphi+\pi}{\sigma_0+1} \rbrace.
\end{align*}
We parameterise $l_{\varphi,c}$ by $t \in \bb{R}$ in such a way that
$z^{\sigma_0+1}\big(l_{\varphi,c}(t)\big)= (t+ i c) e^{+i\varphi}$. Clearly,
$l_{\varphi,c}$ belongs to the sector $\frac{c}{\sigma_0+1}\leq\arg z \leq \frac{d}{\sigma_0+1}$ if $t$ is positive and large enough. Hence, restricted to $l_{\varphi,c}$, the above estimate reads
\begin{equation*}
 \left( \psi_{j'} + O\big(t^{-\delta'}\big)\right) e^{(-\nu_{j'} e^{-i\varphi}) t}= \sum_{j=1}^{\dim V}
 \mc{S}_{j,j'}\left( \psi_{j} +  O\big(t^{-\delta'}\big)\right)
 e^{(-\nu_{j}e^{-i\varphi}) t} , t \to +\infty,
\end{equation*}
where $\delta'=\min \lbrace \frac{\delta}{\sigma_0+1}, 1-\frac{\sigma_1+1}{\sigma_0+1}\rbrace$.
Multiplying the two sides of the estimate by  $e^{\nu_{j'} e^{i\varphi} t}$ and comparing them as $t \to +\infty$, we obtain the thesis.
  \end{proof}
 \end{lemma}

We want to use Theorem \ref{thm:CGM} in order to study asymptotic solutions of global solution
$\mc{A}^{(i)}$, for $i\in\mr{I}$. To do that, we consider the differential equation $\mc{L}_{\kappa_i}\psi=0$ in the loop realisation
\begin{equation}\label{eq:loopLki}
 \mc{L}_{\kappa_i}(z;\la) \psi(z;\la) =0,
\end{equation}
where
$$
\psi(\cdot;\la): \left(\widetilde{\bb{C}}^\times \setminus \bigcup_{j \in J} \bigcup_{l=0}^{r-1}\Pi^{-1} (e^{-2\pi i \kappa_i} \e^l w_j)\right) \to L(\tilde{\omega}_i).
$$
Explicitly, the operators $\mc{L}_{\kappa_i}(z;\la) $ are given by
\begin{align}
 \mc{L}_{\kappa_i}(z;\la)&= \partial_z+ \frac{1}{z}\Biggl(\mathring{f}+  e^{2 \pi i \kappa_i}z\big( 1+ e^{-2 \pi i \kappa_i}\la z^{k-1}\big)v_\theta + \ell+   \notag\\
&+\sum_{j =1}^N \frac{ e^{2r\pi i \kappa_i } r z^{r}}{ e^{2r\pi i \kappa_i} z^r-w^r}\left(-\theta^\vee + X(j)
\right)\Biggl),\label{Lkappai}
\end{align}
and the coefficients  $\kappa_i$'s are as in Table \ref{table:ki}.
Notice that we cannot apply Theorem \ref{thm:CGM} to the differential equation \eqref{eq:loopLki}, as it is not of the form
\eqref{eq:irregularnormal}. However, we will find that \eqref{eq:loopLki} is Gauge equivalent to
an equation of the form \eqref{eq:irregularnormal}. More precisely, we fill find a gauge $G$ such that
\begin{equation}
 G. \mc{L}_{\kappa_i}(z;\la) =\partial_z+ q_i(z;\la) \Lambda(\kappa_i)+ o(z^{-1}),
\end{equation}
where  $\Lambda(\kappa_i)$ is the cyclic element
$\mr{f} + e^{2 \pi i \kappa_i} v_{\theta} $ introduced in \eqref{lambdat}, while
$$ q_i(z;\la)= z^{\frac{1}{h}} \left( 1 + \sum_{l=1}^{\lfloor  \frac{1}{(1-k)h} \rfloor}
 c_l \, e^{-2 l \pi i \kappa_i }\la^l z^{l(1-k)}\right) .$$
Here $c_l$ are the coefficients of the McLaurin expansion of $(1-w)^{\frac{1}{h}}$.

Before proceeding further, we recall some facts about the spectrum of $\Lambda(\kappa_i)$.
The cyclic element is regular and semisimple \cite{Kac90},
and its centralizer is a Cartan subalgebra, say $\underline{\tmf{h}}$ (in the case $r=1$, $\underline{\tmf{h}}$ is said to be in apposition with respect to $\tmf{h}$ \cite{Kos59,moody87}). Letting $ P(\underline{\tilde{\omega}_i})$ be the multi-set of weights of the representation $L(\tilde{\omega}_i)$, corresponding to the Cartan sub-algebra $\underline{\tmf{h}}$,
the eigenvalues and eigenvectors are denoted by
$\mu^{(i)}_{\underline{\omega}}=\underline{\omega} \left(\Lambda(\kappa_i) \right) $ and $\psi^{(i)}_{\underline{\omega}}$.
In particular, the maximal eigenvalue $\mu^{(i)}$ studied in Theorem \ref{thm:genlamba}
corresponds to the fundamental weight $\underline{\tilde{\omega}_i}$, namely $\mu^{(i)}=\mu^{(i)}_{\underline{\tilde{\omega}_i}} $ and
$\psi^{(i)}=\psi^{(i)}_{\underline{\tilde{\omega}_i}}$; see \cite{marava15,marava17}.

\begin{definition}\label{def:stokeslinesLambda}
 We say that $[a,b]$ is a good interval for $\Lambda(\kappa_i)$ if for every $\varphi \in [a,b]$:
 \begin{equation}\label{eq:goodlambdai}
\Re (\mu^{(i)}_{\underline{\omega}} e^{i\varphi})
 = \Re (\mu^{(i)}_{\underline{\omega}'}e^{i\varphi}), \mbox{ if and only if } \mu^{(i)}_{\underline{\omega}}=\mu^{(i)}_{\underline{\omega}'}.
 \end{equation}
 We let $\zeta^{(i)}$ be the supremum of all positive numbers $\zeta$
 \footnote{Since $\mu^{(i)}$ is the maximal eigenvalue as per Definition \ref{def:maximal}, the
 condition (\ref{eq:deltai}) is satisfied whenever $\zeta$ is small enough. In the case
 $\mf{g}=A_n$, it is simple to verify that $\zeta^{(1)}=\zeta^{(n)}=\frac\pi{n+1}$.}
 such that for every $\varphi \in [-\zeta,\zeta]$ and $\underline{\omega}'\neq \underline{\tilde{\omega}_i}$:
 \begin{equation}\label{eq:deltai}
 \Re (\mu^{(i)}_{\underline{\tilde{\omega}_i}} e^{i\varphi}) >  \Re (\mu^{(i)}_{\underline{\omega}} e^{i\varphi}).
 \end{equation}
\end{definition}

\begin{proposition}\label{prop:basisatinfinity}
Let $i \in  i\in \mr{I}$ and
$S_i$ be the primitive of $z^{-1}q_i(z,\la)$ given by
\begin{equation}
S_i(z;\la)=
\begin{cases}
 & h \, z^{\frac{1}{h}} \left( 1 + \sum_{l=1}^{\lfloor  \frac{1}{(1-k)h} \rfloor}
 \frac{c_l e^{-2 l \pi i \kappa_i }\la^l z^{-l(k-1)} }{1- h l (1-k)} \right) , \quad  \frac{1}{h(1-k)} \notin \bb{N}\\
 & h \, z^{\frac{1}{h}} \left( 1 + \sum_{l=1}^{m-1}
 \frac{c_l e^{2 l \pi i \kappa_i }\la^l z^{-l(k-1)}}{1- l\,m}  \right) +
  c_{m}\log z
  , \,   m:=\frac{1}{h(1-k)} \in \bb{N}
\end{cases} .
\end{equation}
There exists a $\delta>0$ such that:

1. If $[a,b]$ a good interval for $\Lambda(\kappa_i)$,
for every $\la \in \bb{C}$, the differential equation (\ref{Lkappai})
admits a unique basis of solutions $\widetilde{\Phi}_{\underline{\omega}}^{(i)}(\cdot;\la)$, with
$\underline{\omega} \in P(\underline{\tilde{\omega}_i})$, such that
 \begin{equation}\label{eq:wtildePsiom}
  \widetilde{\Phi}_{\underline{\omega}}^{(i)}(z;\la)=  z^{\frac{1}{h}\ad \mathring{\rho}^\vee} \left( \psi_{\underline{\omega}}^{(i)} + O\big(z^{-\delta}\big) \right) e^{-\mu_{\underline{\omega}}^{(i)} S_i(z;\la)},
 \end{equation}
 as $z \to + \infty$ in the closed sector $ a h \leq\arg z \leq h (\pi+ b)$.
 Moreover, the solutions $\widetilde{\Psi}_{\underline{\omega}}^{(i)}$'s are entire functions of the parameter
 $\la$. Therefore, the elements $\Phi^{(i)}_{\underline{\omega}} \in \mc{A}^{(i)}, \underline{\omega} \in P(\underline{\tilde{\omega}_i})$  corresponding to $\widetilde{\Psi}^{(i)}_{\underline{\omega}}(\cdot;\la)$ under the isomorphism $z^{k\scaling}$, i.e.
 $\Phi^{(i)}_{\underline{\omega}} (z)(\la)=\widetilde{\Phi}^{(i)}_{\underline{\omega}}(z;\la z^{k})$, form a
 $\mc{O}'$-basis
 of $\mc{A}^{(i)}$.

2. For every $\la \in \bb{C}$, the differential equation (\ref{Lkappai})
admits a unique solution $\widetilde{\Phi}^{(i)}(\cdot;\la)$ such that for all $\theta< \theta^{(i)}$,
 \begin{equation}\label{eq:wtildePsii}
  \widetilde{\Phi}^{(i)}(z;\la)=  z^{\frac{1}{h}\ad \mathring{\rho}^\vee} \left( \psi^{(i)} + O\big(z^{-\delta}\big) \right) e^{-\mu^{(i)} S_i(z;\la)},
 \end{equation}
 as $z \to + \infty$ in the closed sector $|\arg z|\leq h (\pi+\theta^{(i)})$.
 Moreover, the solution $\widetilde{\Phi}^{(i)}$ is an entire function of the parameter $\la$.
 We denote by $\Psi^{(i)}$ the element of $\mc{A}^{(i)}$ defined by
 $\Phi^{(i)}(z;\la)= \widetilde{\Phi}^{(i)}(z;\la z^{k})$.
\end{proposition}

\begin{proof}
We find a Gauge transformation $G=G_2 \circ G_1$ that transforms the connection $\mc{L}_{\kappa_i}(z;\la)$ to
\begin{equation*}
\partial_z+ q_i(z;\la) \Lambda(\kappa_i)+ O(z^{-1-\delta}), \mbox{ for some } \delta>0.
\end{equation*}
By definition of the function $q_i$, we have that
\beq\label{eq:delta0}
z^{\frac{1}{h}}\big((1+ e^{-2 \pi i \kappa_i} \lambda z^{k-1})\big)^{\frac{1}{h}}-q_i(z;\la)=
 O\left(z^{-\delta_0} \right),
\eeq
with
$$ \delta_0=(1-k)\left(\lfloor  \frac{1}{(1-k)h} \rfloor+1- \frac{1}{(1-k)h}\right)>0.$$
Therefore acting with $G_1=\big(q(z;\la)\big)^{-\ad \mathring{\rho}^\vee}$, we obtain
\begin{equation}
 G_1 \mc{L}_{\kappa_i}(z;\la)=\partial_z+ z^{-1}q_i(z;\la) \Lambda(\kappa_i) + \frac{\ell+ \frac{1}{h^\vee} \mathring{\rho}^\vee-
 r |J| r \theta^\vee}{z}+  O(z^{-1-\delta_0}).
\end{equation}
Now, we let
$\tilde{n} $ be the unique element in $\langle e_1,\dots,e_n \rangle$ such that
$[\tilde{n},\mathring{f}]=\ell+ \frac{1}{h^\vee} \mathring{\rho}^\vee- N r \theta^\vee$. Acting with
$G_2=\exp\left( q_i(z;\la)^{-1} \tilde{n}\right)$, we get
$$ \mc{L}_{\kappa_i,*}(z;\la):= G_2 \circ G_1 \mc{L}_{\kappa_i}(z;\la)=\partial_z+ z^{-1} q_i(z;\la) \Lambda(\kappa_i) +
 O(z^{-1-\delta}),$$
with $\delta=
\min \lbrace \frac{1}{h},\delta_0\rbrace>0$.

Applying Theorem \ref{thm:CGM} to $\mc{L}_{\kappa_i,*}(z;\la)$, we deduce the thesis.
\end{proof}

\subsection{Central connection matrix and Stokes matrix. Q and T functions.}
Let $\zeta_*>0$ be a sufficiently small numbers so that for every $\zeta, 0<\zeta <\zeta_*$, the interval
$[-\zeta_*,-\zeta]$ is a good interval for $\Lambda(\kappa_i)$ for every $i \in\mr{I} $.
After Proposition \ref{prop:basisatinfinity}, $\mc{A}^{(i)}$ admits the $\mc{O}'$-basis
 $\Psi^{(i)}_{\underline{\omega}}, \underline{\omega} \in P(\underline{\omega}_i)$, where
 $\Psi^{(i)}_{\underline{\omega}}$ is such that
 $\widetilde{\Psi}^{(i)}_{\underline{\omega}}=
z^{-k \scaling} \Psi^{(i)}_{\underline{\omega}} $ satisfies the
asymptotics (\ref{eq:wtildePsiom}) in the sector
$[-\zeta_*,\pi-\zeta]$ for all $0<\zeta<\zeta_*$.
Moreover, the solution $\Psi^{(i)}:=\Psi^{(i)}_{\underline{\tilde{\omega}_i}}$ is subdominant along the real positive axis,
and admits the asymptotic behaviour (\ref{eq:wtildePsii}) on the larger sector $[-\zeta_*-\pi,\zeta_*+\pi]$.

We also know from Proposition \ref{prop:frobeniusolutions}, that --
assuming that $\ell$ is generic -- $\mc{A}^{(i)}$ admits another $\mc{O}'$ distinguished basis -- the eigenbasis of the monodromy operator $\mc{M}$ -- whose elements are the Frobenius solutions
$\chi_{\omega}$, $\omega \in P(\tilde{\omega}_i)$.

The two bases are thefore related by invertible matrix $Q^{(i)}$, called \textit{central connection matrix}, whose coefficients belongs to $\mc{O}'$.  We write
\begin{equation}\label{eq:Qmatrix}
 \Psi^{(i)}_{\underline{\omega}}= \sum_{\omega \in P(\tilde{\omega}_i)}  Q^{(i)}_{\omega,\underline{\omega}} \chi_{\omega},
 \qquad \underline{\omega} \in P(\underline{\tilde{\omega}_i}).
\end{equation}
In particular, denoting by $Q^{(i)}_{\omega}:= Q^{(i)}_{\omega,\underline{\tilde{\omega}_i}}$, we have that
\begin{equation}\label{eq:PsiQchi}
 \Psi^{(i)}= \sum_{\omega \in P(\tilde{\omega}_i)}  Q^{(i)}_{\omega} \chi_{\omega}(z;\la).
\end{equation}
The coefficients $Q^{(i)}_{\omega,\omega'} \in \mc{O}' $ or $ Q^{(i)}_{\omega} \in \mc{O}'$ are called $Q$ functions.

We define now the Stokes matrix $\mc{T}$. Let us a fix a $0<\widetilde{\zeta}_*<\zeta$. It is straighforward to check that
the interval $[-\zeta,-\widetilde{\zeta}]$ is a good interval for $\Lambda(\kappa_i)$ for every $i$. Therefore, there exists a $\mc{O}'$ basis of $\mc{A}^{(i)}$,
$\Xi^{(i)}_{\omega}, \omega \in P(\underline{\tilde{\omega}_i})$, such that $\widetilde{\Xi}^{(i)}_{\omega}=z^{-k \scaling} \Xi^{(i)}_{\omega}$ has asymptotic behaviour (\ref{eq:wtildePsii}) in the sector
 $[-\zeta_*-\pi ,-\widetilde{\zeta_*}]$. Moreover, the bases
 $\widetilde{\Psi}^{(i)}_{\underline{\omega}}$'s -- which was defined above -- and $\widetilde{\Xi}^{(i)}_{\underline{\omega}'}$'s
 are consecutive bases of (\ref{Lkappai}), according to the Definition \ref{def:consecutive} above.

 We let therefore $\mc{T}^{(i)}$ be the matrix of change of basis, with coefficients in $\mc{O}'$, defined by
\begin{equation}\label{eq:Tmatrix}
 \Psi^{(i)}_{\underline{\omega}'}= \sum_{\underline{\omega} \in P(\tilde{\omega}_i)}
 \Xi^{(i)}_{\underline{\omega}} \mc{T}^{(i)}_{\underline{\omega},\underline{\omega}'}, \; \underline{\omega}' \in P(\tilde{\omega}_i).
\end{equation}
 The coefficients $\mc{T}^{(i)}_{\underline{\omega},\underline{\omega}'}$ are known as $\mc{T}$ functions.

 Since the bases $\widetilde{\Psi}^{(i)}_{\underline{\omega}}$'s and $\widetilde{\Xi}^{(i)}_{\underline{\omega}'}$'s $\mc{T}^{(i)}$ are consecutive then $\mc{T}^{(i)}$ is Stokes matrix. The following is direct corollary of Lemma \ref{lem:Stokes}
\begin{align}\label{eq:unipoT1}
 & \mc{T}^{(i)}_{\underline{\omega},\underline{\omega}}=1, \forall \underline{\omega} \in P(\underline{\tilde{\omega}_i}),\\ \label{eq:unipoT2}
 & \mc{T}^{(i)}_{\underline{\omega},\underline{\omega}'}=0, \mbox{ if } \underline{\omega} \neq  \underline{\omega}' \mbox{ and }
 \Re\big( e^{i\varphi}\mu_{\underline{\omega}}\big) \leq  \Re \big(e^{i\varphi}
 \mu_{\underline{\omega}'} \big) \mbox{ for } \varphi \in [-\zeta_*,-\widetilde{\zeta_*}].
\end{align}

\begin{remark}
As we have already stated in the introduction, the defintion of $Q^{(i)}$ matrix and of
the Stokes matrix $\mc{T}^{(i)}$ is, for the general Lie algebra, a novelty of the present work.
The thorough study of the matrices $\mc{T}^{(i)}$'s and $Q^{(i)}$'s is beyond the scope of the present paper.
In particular, we do not discuss the $\mc{T}Q$ relations. We only note here that in the case $\tmf{g}=A_n$ we have
$\Psi^{(1)}= \mc{M}^{h-1} \Xi^{(1)}$, where $\mc{M}$ is the monodromy operator. The $TQ$ relations for the algebra $A_n$ with $r=1$ follow from this \cite{dorey00}.
\end{remark}

\subsection{$\Psi$-system, QQ system and the Bethe Equations}

Let $\Psi^{(i)}\in\mc{A}^{(i)}$, $i \in \mr{I}$ be the elements defined in Proposition \ref{prop:Psisystem}.
They satisfy the following system of nonlinear relations known as $\Psi$-system
\begin{proposition}\label{prop:Psisystem}
 For every $i\in\mr{I}$, let $\Psi^{(i)}\in \mc{A}^{(i)}$ be defined by
$$\Psi^{(i)}(z)(\la)=\widetilde{\Psi}^{(i)}(z;\la z^{-k}).$$
Then the following identity, known as \emph{$\psi-$system} holds true:
 \begin{equation}\label{eq:Psisystem}
  m_i  \left(R_i \left(\Psi^{(i)}_{-\frac{D_i}2}\right) \wedge
  \Psi^{(i)}_{\frac{D_i}2} \right) = \bigotimes_{j \in \mr{I}} \bigotimes_{l=0}^{B_{ij}-1}
  \Psi^{(j)}_{\frac{B_{ij}-1-2l}{r}},
 \end{equation}
where $m_i$ is the linear map defined in (\ref{eq:mi}).
\begin{proof}
  This is proven in \cite[Theorem 4.7 (v)]{marava17} for the ground state oper. The same proofs hold for all FF opers, since it only
 depends on the asymptotic behaviour (\ref{eq:wtildePsii}), which, as Proposition \ref{prop:basisatinfinity} , is  independent of the additional singularities.
\end{proof}

\end{proposition}
As a corollary, we have the following
\begin{theorem}\label{thm:QQsystem}
Assume that the pair $(k,\ell)$ is generic as in Definition \ref{genericpair}, and that $\mathring{f}+\ell$ is semisimple.
Fixing an element $w$ of the Weyl group of $\rmf{g}$, denote $Q_w^{(i)}=Q_{w(\tilde{\omega}_i)}$ and
$\widetilde{Q}^{(i)}_w=Q_{w(\tilde{\omega}_i-\alpha_i)}$. \\
1. The following system, known as $QQ-$system, holds
for all $i \in \mr{I}$ and all $\la \in \bb{C}$
\begin{align}
\begin{split}\label{QQsystem}
\prod_{j\in \mr{I}}\prod_{s=0}^{B_{ij}-1}Q_w^{(j)}(q^{\frac{B_{ij}-1-2s}{r}} \la)
&=e^{\pi i D_i\langle \ell,w(\tilde{\alpha}_i)\rangle}Q_w^{(i)}(q^{D_i} \la)\widetilde{Q}_w^{(i)}(q^{- D_i}\la)\\
&-e^{- \pi i D_i\langle \ell,w(\tilde{\alpha}_i)\rangle}Q_w^{(i)}(q^{-D_i }\la)\widetilde{Q}_w^{(i)}(q^{D_i}\la)
\,, 
\end{split}
\end{align}
where $q=e^{\pi i k}$.\\
2. Let $\la^\ast$ be a zero of $Q_w^{(i)}$ such that
$$\prod_{j\in \mr{I}}\prod_{\ell=0}^{B_{ij}-1}Q^{(j)}(q^{\frac{B_{ij}-1-2\ell}{r}-D_i} \la^\ast)\neq 0.$$
Then, the following Bethe Equations hold
\beq
\prod_{j\in\mr{I}}e^{i\pi\overline{C}_{sj}\theta^j_w}\frac{Q^{(j)}(q^{\overline{C}_{sj}}\la^\ast)}{Q^{(j)}(q^{-\overline{C}_{sj}}\la^\ast)}=-1,
\eeq
where $\theta^j_w=\langle\ell,w(\tilde{\omega}_j)\rangle$, and 
$\overline{C}_{ij}=C_{ji}D_i$ is the symmetrized Cartan matrix of $^L\rmf{g}$.
\begin{proof}
Plug the decomposition (\ref{eq:PsiQchi}) into the $\Psi$-system (\ref{eq:Psisystem}) and use the relation (\ref{eq:chisystem}).
\end{proof}

\end{theorem}


\section{Trivial Monodromy Conditions}\label{sec:trivialmonodromy}
In this section, following \cite{mara18}, we study the trivial monodromy conditions the connections \eqref{eq:ffopersourgauge}. In Theorem \ref{thm:trivialmonodromy}, we prove that
\begin{itemize}
 \item If $r=1$, the trivial monodromy conditions are equivalent to a complete system of
 $|J|(2h-2)$ equations in$|J|(2h-2)$ scalar unknowns;
\item If $r>1$, the trivial monodromy conditions are satisfied if and only if $J=\emptyset$.
\end{itemize}
In order to prove  Theorem \ref{thm:trivialmonodromy}, we notice that if $x=z-w_j$ is a local coordinate centered at $w_j$, the connection \eqref{looprealization} admits an expansion of the form
\begin{equation*}
\mc{L}=\partial_x+\frac{-\theta^\vee+\eta}{x}+ O (x), \; \eta \in \tmf{n}^+
\end{equation*}
We show that the property that $\mc{L}$ has
trivial monodromy at $x=0$ is
equivalent to a set of polynomial constraints on the coefficients of the Laurent expansion at $x=0$, which we study
to obtain the main theorem of the section. As a preliminary step, we first address the study of the gradation induced by
 $\theta^\vee$.

\subsection{The gradation induced by $\theta^\vee$}   Let $\theta^\vee\in\rmf{h}\in\rmf{g}\subset\tmf{g}$ be the element introduced in \eqref{221026-3/0} and appearing in the loop realization \eqref{looprealization}. Recall that  under the action of $\sigma$ the algebra $\tmf{g}$ decomposes as a direct sum of $\rmf{g}-$modules: the adjoint $\rmf{g}-$module $\tmf{g}_0=\rmf{g}$, whose weights are the zero weight and the roots of $\rmf{g}$, and $r-1$ copies of the $\rmf{g}-$module $\tmf{g}_1=V_\theta$, whose weights are the zero weight and the short roots of $\rmf{g}$. It follows from this that  the adjoint action of $\theta^\vee$ on $\tmf{g}$ induces a $\bb{Z}-$gradation:
\beq\label{230905-1}
\tmf{g}=\bigoplus_{j\in\bb{Z}}\tmf{g}_j[\theta^\vee],\qquad \tmf{g}_j[\theta^\vee]=\{x\in\tmf{g}\,|\,[\theta^\vee,x]=jx\}.
\eeq
We denote by $\pi_j:\tmf{g}\to\tmf{g}_j[\theta^\vee]$ the corresponding projection. The gradation \eqref{230905-1} is compatible with the principal gradation, that is
\beq\label{230904-1}
\langle \theta^\vee,\alpha_i\rangle=-\langle \alpha_0^\vee,\alpha_i\rangle=-C_{0i}\geq 0,\qquad i\in\mr{I},
\eeq 
where we used \eqref{221116-1/0}. This in particular implies that
\beq\label{231004-1}
\bigoplus_{j\geq 1}\tmf{g}_j[\theta^\vee]\subset\tmf{n}^+.
\eeq
\begin{lemma}\label{lemmaspectrum}
\begin{enumerate}
\item The spectrum of the adjoint action of $\theta^\vee$ on the $\rmf{g}-$module $\tmf{g}$ is given in Table \ref{table:affine}.
\item Let $r>1$, and consider the $\rmf{g}-$module $V_\theta\subset \tmf{g}$. Then
\beq\label{230904-3}
V_\theta\cap\tmf{g}_2[\theta^\vee]=\bb{C}v_\theta.
\eeq
\end{enumerate}
\end{lemma}
\begin{proof}
(1) The weights of the $\rmf{g}-$module $\tmf{g}$ are the zero weight and the roots of $\rmf{g}$.  Let $\widetilde{\theta}$ be the highest root of $\rmf{g}$. A direct case-by-case computation shows that $\langle\theta^\vee,\widetilde{\theta}\rangle=2$ in all cases except when $\mf{g}=D_4^{(3)}$  (i.e. when $\tmf{g}=D_4$ and $\rmf{g}=G_2$), in which case $\langle\theta^\vee,\widetilde{\theta}\rangle=3$. Since the $\theta^\vee-$gradation is compatible with the principal gradation this value is the maximum of the spectrum. Moreover, the spectrum is symmetric with respect to $0$, and all intermediate values are attained except when $\mf{g}=A_1^{(1)}$  (i.e. $\tmf{g}=\rmf{g}=A_1$) and when $\mf{g}=D_{n+1}^{(2)}$ (i.e. $\tmf{g}=D_{n+1}$ and $\rmf{g}=B_n$), in which case $\pm 1$ are not part of the spectrum.\\
(2) The weights of the $\rmf{g}-$module $V_\theta$ are the zero weight and the short roots of $\rmf{g}$. The highest short root of $\rmf{g}$ is $\theta$, and $\langle\theta^\vee,\theta\rangle=2$. Reasoning as in part (1) we get that the spectrum is given by $\{-2,0,2\}$ if $\mf{g}=D_{n+1}^{(2)}$, and it is  $\{-2,-1,0,1,2\}$ otherwise. To prove \eqref{230904-3}, since the $\theta^\vee-$gradation is compatible with the principal gradation, it is sufficient to show that $f_iv_\theta\notin\tmf{g}_2[\theta^\vee]$ for every $i\in\mr{I}$. To prove that, first notice that
\begin{align*}
[\theta^\vee,f_iv_\theta]=(2-\langle\theta^\vee,\alpha_i\rangle)f_iv_\theta=(2+\langle\alpha_0^\vee,\alpha_i\rangle)f_iv_\theta=(2+C_{0i})f_iv_\theta.
\end{align*}
then $f_iv_\theta\in\tmf{g}_2[\theta^\vee]$ if and only if $C_{0i}=0$. Assume that $f_iv_\theta\in\tmf{g}_2[\theta^\vee]$, so that $C_{i0}=C_{0i}=0$. Then, for every $m\in\mr{I}$ we have 
\begin{align*}
e_mf_iv_\theta&=f_ie_mv_\theta+[e_m,f_i]v_\theta=\delta_{mi}\alpha^\vee_i=\delta_{mi}\langle\alpha^\vee_i,\theta\rangle v_\theta\\
&=-\delta_{mi}\langle\alpha^\vee_i,\alpha_0\rangle v_\theta=-\delta_{mi}C_{i0}v_\theta=0.
\end{align*}
Since $V_\theta$ is an irreducible $\rmf{g}-$module, we conclude that $f_iv_\theta=0$.
\end{proof}
We consider the following subspace of $\tmf{g}$:
\beq\label{u}
\mf{u}=\ad_{\theta^\vee}\rmf{n}^++\bb{C}v_\theta.
\eeq
Note that if $r=1$ then $v_\theta=e_\theta\in\rmf{n}^+$, and $\mf{u}=\ad_{\theta^\vee}\rmf{n}^+$. If $r>1$ the sum in \eqref{u} is a direct sum. Due to  \eqref{231004-1}, and considering the $\sigma$-invariant part only, it follows that
\beq\label{231005-1}
\ad_{\theta^\vee}\rmf{n}^+=\bigoplus_{j\geq 1}\rmf{g}_j[\theta^\vee].
\eeq
We introduce the following basis of $\mf{u}$. Let $\mr{\Delta}$ denote the set of roots of $\rmf{g}$, and define
$$\mr{\Delta}_{\mf{u}}=\{\alpha\in\mr{\Delta}|\langle\theta^\vee,\alpha\rangle\geq1\}.$$
Note the inclusions  $\mr{\Delta}_{\mf{u}}\subset \mr{\Delta}_+$, where $\mr{\Delta}_+$ is the set of positive roots of $\rmf{g}$ with respect to the Borel subalgebra $\rmf{b}^+$.  Let $\{e_\alpha\,|\,\alpha\in\mr{\Delta}_+\}$ be the basis of $\rmf{n}^+$ introduced above. Then
\beq\label{ubasis}
\mf{u}=\left(\bigoplus_{\alpha\in \mr{\Delta}_{\mf{u}}}\bb{C}e_\alpha\right)\oplus\bb{C}v_\theta.
\eeq

%


\subsection{Local conditions on trivial monodromy}
If $x=z-w_j$ is a local coordinate centered at $w_j$, the connection
\eqref{looprealization} admits an expansion of the form
\beq\label{genoper1}
\mc{L}=\partial_x+\frac{R}{x}+a+bx+cx^2+O(x^3),
\eeq
with $ R=-\theta^\vee+\eta$, with $\eta \in \tmf{n}^+$ and $a,b,c\in\tmf{g}$.

If $\eta_0=\pi_0(\eta)$ is the part of degree $0$, with respect to the
gradation induced by $\theta^\vee$, of $\eta$, then $R-\eta_0$ is
semi-simple and conjugated to $-\theta^\vee$, with spectrum provided in Table \ref{table:affine}.
Therefore $\tmf{g}$ splits into eigenspaces of $R-\eta_0$:
\beq\label{gradationR}
\tmf{g}=\bigoplus_i\tmf{g}_i[R],\qquad \tmf{g}_i[R]=\lbrace x \in \tmf{g}, [R-\eta_0,x]= ix
\rbrace.
\eeq
We denote by $\pi^{R}_{i}:\tmf{g}\to \tmf{g}_i[R]$ the corresponding projections.
\begin{definition}\label{def:trivialgeneral}
 We say that the connection \eqref{genoper1} has trivial monodromy at $0$ if for every finite dimensional
 $\tmf{g}$-module $V$, the differential equation $\mc{L}\psi(x)=0, \psi: D \to V$, with $D$ a punctured neighborhood of
 $x=0$, has trivial monodromy.
\end{definition}

\begin{proposition}\label{prop:zeromonodromy}
If $\rmf{g}$ is not of type $G_2$, the operator \eqref{genoper1} has trivial monodromy at $x=0$ if and only if
\begin{subequations}\label{zeromonodromy}
\begin{align}
&\eta_0=0,\label{zeromonodromya}\\
&\pi^{R}_{-1}(a)=0,\label{zeromonodromyb}\\
&\pi^{R}_{-2}(b)=[\pi^{R}_{-2}(a),\pi^{R}_{0}(a)].\label{zeromonodromyc}
\end{align}
\end{subequations}

If $\rmf{g}$ is of type $G_2$, the operator \eqref{genoper1} has trivial monodromy at $x=0$ if and only if
\begin{subequations}\label{zeromonodromyg2}
\begin{align}
&\eta_0=0,\label{zeromonodromyag2}\\
&\pi^{R}_{-1}(a)=0,\label{zeromonodromybg2}\\
&\pi^{R}_{-2}(b)=[\pi^{R}_{-2}(a),\pi^{R}_{0}(a)]+\frac54[\pi^{R}_{-3}(a),\pi^{R}_{1}(a)],\label{zeromonodromycg2}\\
&\pi^{R}_{-3}(c)=\frac12[\pi^{R}_{-3}(a),\pi^{R}_{0}(b)]+[\pi^{R}_{-2}(a),\pi^{R}_{-1}(b)]-[\pi^{R}_{0}(a),\pi^{R}_{-3}(b)]
\nonumber\\
& +\frac32[\pi^{R}_{-2}(a),[\pi^{R}_{1}(a),\pi^{R}_{-2}(a)]] +\frac32
[\pi^{R}_{0}(a),[\pi^{R}_{-3}(a),\pi^{R}_{0}(a)]]. \label{zeromonodromydg2}
\end{align}
\end{subequations}
\end{proposition}
\begin{proof}
According to the general theory \cite{bava83}, a connection of the form
\begin{equation*}
 \partial_x+\frac{R}{x}+ \sum_{j\geq0} d_j x^j ,
\end{equation*}
with $d_j\in\tmf{g}$, is Gauge equivalent to a connection of the form
\begin{equation*}
 \partial_x+\frac{R}{x}+\sum_{j\geq0} d'_j x^j , 
\end{equation*}
where $d'_j\in \tmf{g}_{-1-j}[R]$. The latter connection is said to be
in aligned form and it
has trivial monodromy at $0$ if and only if
$\eta_0=0$ and $d'_j=0$ for all $j \in \mbox{spec}(R-\eta_0), j\leq-1$.

Explicit computations of the aligned form of the connection lead to
the thesis; see
\cite[Theorem 8.4]{mara18} for details.
\end{proof}
It follows from the above proposition that if the connection \eqref{genoper1} has trivial monodromy then $\eta_0=0$.
This in turn implies that the coefficient $R$ is semisimple and conjugated to $-\theta^\vee$; explicitly, we have:
\beq\label{231004-3}
R=-e^{\ad\bar{\eta}}\theta^\vee,
\eeq
where
\begin{subnumcases}{\hspace{-80pt}\bar{\eta}=\label{etabar}}
\eta_1+\frac{1}{2}\eta_2\qquad & if $\rmf{g}\neq G_2,$\phantom{50pt}\label{etabarnotg2} \\
\eta_1+\frac{1}{2}\eta_2+\frac{1}{3}\eta_3-\frac{1}{12}[\eta_1,\eta_2]\qquad & if  $\rmf{g}= G_2$.\label{etabarg2}
\end{subnumcases}
In the above formula $\eta_i=\pi_i(\eta)$ and $\pi_i$  is the projection associated to the gradation \eqref{230905-1}.
The gradation \eqref{gradationR} is therefore conjugated to the gradation \eqref{230905-1}
$$\tmf{g}_i[R]=\lbrace e^{\ad\bar{\eta}}x | x\in \tmf{g}_{-i}[\theta^\vee], \rbrace$$
and the relation
$$\pi_i^R(x)=e^{\ad\bar{\eta}}(\pi_{-i}(e^{-\ad\bar{\eta}}x))$$
holds true for every $x\in\tmf{g}$. Using the last identity we can express the trivial monodromy conditions of Proposition \ref{prop:zeromonodromy} with respect to the $\theta^\vee-$gradation:
\begin{proposition}\label{prop:zeromonodromytheta}
If $\rmf{g}$ is not of type $G_2$, the operator  \eqref{genoper1} has trivial monodromy at $x=0$ if and only if
\begin{subequations}\label{zeromonodromytheta}
\begin{align}
&\eta_0=0,\label{zeromonodromyatheta}\\
&\pi_{1}(e^{-\ad\bar{\eta}}a)=0,\label{zeromonodromybtheta}\\
&\pi_{2}(e^{-\ad\bar{\eta}}b)=[\pi_{2}(e^{-\ad\bar{\eta}}a),\pi_{0}(e^{-\ad\bar{\eta}}a)],\label{zeromonodromyctheta}
\end{align}
\end{subequations}
where $\bar{\eta}$ is given by \eqref{etabarnotg2}. If $\rmf{g}$ is of type $G_2$, the operator \eqref{genoper1} has trivial monodromy at $x=0$ if and only if
\begin{subequations}\label{zeromonodromyg2theta}
\begin{align}
&\eta_0=0,\label{zeromonodromyag2theta}\\
&\pi_{1}(e^{-\ad\bar{\eta}}a)=0,\label{zeromonodromybg2theta}\\
&\pi_{2}(e^{-\ad\bar{\eta}}b)=[\pi_{2}(e^{-\ad\bar{\eta}}a),\pi_{0}(e^{-\ad\bar{\eta}}a)]+\frac54[\pi_{3}(e^{-\ad\bar{\eta}}a),\pi_{-1}(e^{-\ad\bar{\eta}}a)],\label{zeromonodromycg2theta}\\
&\pi_{3}(e^{-\ad\bar{\eta}}c)=\frac12[\pi_{3}(e^{-\ad\bar{\eta}}a),\pi_{0}(e^{-\ad\bar{\eta}}b)]+[\pi_{2}(e^{-\ad\bar{\eta}}a),\pi_{1}(e^{-\ad\bar{\eta}}b)]\nonumber\\
&-[\pi_{0}(e^{-\ad\bar{\eta}}a),\pi_{3}(e^{-\ad\bar{\eta}}b)]+\frac32[\pi_{2}(e^{-\ad\bar{\eta}}a),[\pi_{-1}(e^{-\ad\bar{\eta}}a),\pi_{2}(e^{-\ad\bar{\eta}}a)]] \notag\\
& +\frac32 [\pi_{0}(e^{-\ad\bar{\eta}}a),[\pi_{3}(e^{-\ad\bar{\eta}}a),\pi_{0}(e^{-\ad\bar{\eta}}a)]],\label{zeromonodromydg2theta}
\end{align}
\end{subequations}
where $\bar{\eta}$ is given by \eqref{etabarg2}.
\end{proposition}

\subsection{Monodromy equations for FFH connections}
After Lemma \ref{lem:AisoB}, the FFH connection \eqref{eq:ffopersourgauge} has trivial monodromy if and only if its loop realisation $\mc{L}(z,\la)$ has trivial monodromy at all $w_j$, $j \in J$, for every $\la$.
According to Proposition \ref{prop:zeromonodromytheta}, $\mc{L}(z,\la)$ has trivial monodromy at $w_j$ if and only
if its Laurent expansion at $w_j$ satisfies the constraints \eqref{zeromonodromytheta}  --or \eqref{zeromonodromyg2theta} in the case $\rmf{g}=\mf{g}_2$ --
for every $\la \in\bb{C}$.
Therefore, the FFH connection \eqref{eq:ffopersourgauge} has trivial monodromy if and only if the constraints \eqref{zeromonodromytheta}  -- or \eqref{zeromonodromyg2theta} -- are satisfied at all $w_j$, $j \in J$, and for every $\la \in \bb{C}$.
To proceed further, we write these constraints explicitly in terms of the connection $\mc{L}(z;\la)$, as in
\eqref{looprealization}.
Fix $j\in J$ and let $x=z-w_j$, the connection $\mc{L}(z;\la)$ admits an expansion of the form \eqref{genoper1}, say
\beq\label{genoper1j}
\mc{L}=\partial_x+\frac{R(j)}{x}+a(j)+b(j)x+c(j)x^2+O(x^3),
\eeq
with $ R(j)=-\theta^\vee+\eta(j)$, with $\eta(j) \in \tmf{n}^+$ and $a(j),b(j),c(j)\in\tmf{g}$. The explicit form of $\eta(j)$, $a(j)$ and $b(j)$ , for $j\in J$, is given by:
\begin{subequations}\label{etaab}
\begin{align}
\eta(j)=&X(j),\label{eta}\\
w_ja(j)=&\mr{f}+\ell+(w_j+w_j^k\la)v_\theta+\frac{r-1}{2}(-\theta^\vee+X(j))\notag\\
&+\sum_{i\neq j}\frac{rw_j^r}{w_j^r-w_i^r}\left(-\theta^\vee+X(i)\right),\label{a(j)}\\
w_j^2b(j)=&-\mr{f}-\ell+(k-1)w_j^k\la v_\theta+\frac{(r-1)(r-5)}{12}(-\theta^\vee+X(j))\notag\\
&-\sum_{i\neq j}\frac{rw_j^r\left((r-1)w_i^r+w_j^r)\right)}{w_j^r-w_i^r}(-\theta^\vee+X(i))
\end{align}
\end{subequations}
The term $c(j)$, which is relevant only in the $\mf{g}=D_4^{(3)}$ case, can be computed in a similar way. The coefficients $a(j)$, $b(j)$ and $c(j)$ depend on the loop variable $\lambda$, while $\eta(j)$ does not. For each $j\in J$ we consider the trivial monodromy system \eqref{zeromonodromytheta} (if $\rmf{g}$ is not of type $G_2$) or \eqref{zeromonodromyg2theta} (if $\rmf{g}$ is of type $G_2$). These equations have to be satisfied for every value of $\lambda$. For example, if $\mf{g}\neq D_4^{(3)}$, then equations \eqref{zeromonodromytheta} applied to the present case read:
\begin{subequations}\label{zeromonodromythetaj}
\begin{align}
&\pi_0(\eta(j))=0,\label{zeromonodromyathetaj}\\
&\pi_{1}(e^{-\ad\bar{\eta}(j)}a(j))=0,\label{zeromonodromybthetaj}\\
&\pi_{2}(e^{-\ad\bar{\eta}(j)}b(j))=[\pi_{2}(e^{-\ad\bar{\eta}(j)}a(j)),\pi_{0}(e^{-\ad\bar{\eta}(j)}a(j))],\label{zeromonodromycthetaj}
\end{align}
\end{subequations}
where $\bar{\eta}(j)$ is given by \eqref{etabarnotg2}:
$$\bar{\eta}(j)=\pi_1(\eta(j))+\frac12\pi_2(\eta(j)).$$
The unkowns of the trivial monodromy system \eqref{zeromonodromythetaj} are $w_j\in\bb{C}^\times$, and $X(j)\in\tmf{n}^+\cap \rmf{g}$, or equivalently by $\{w_j,\eta(j)|j\in J\}$, where $\eta(j)$ is given in \eqref{eta}. 
Recalling the definition \eqref{u} of the subspace $\mf{u}$, and using \eqref{zeromonodromyathetaj} together with \eqref{231005-1} we obtain
$$\eta(j)\in\mf{u}\cap\rmf{n}^+=\ad_{\theta^\vee}\rmf{n}^+,$$ 
and we can write $\eta(j)$ with respect to the basis \eqref{ubasis} of $\mf{u}$  as
$$\eta(j)=\sum_{\alpha\in\mr{\Delta}_{\mf{u}}}x^\alpha(j)e_\alpha,$$
for some $x^\alpha(j)\in\bb{C}$.  Given the connection \eqref{eq:ffopersourgauge}, the trivial monodromy system \eqref{zeromonodromythetaj} turns into a system of algebraic equations for the complex unknowns 
\beq\label{UJ}
U_J=
\lbrace w_j,x^{\alpha}(j)\,|\,\alpha\in\mr{\Delta}_{\mf{u}}\rbrace_{j \in J},\\
\eeq
which we analyse below. Note, incidentally, that $|U_J|=|J|\dim(\ad_{\theta^\vee}\rmf{n}^++1)$, where $|J|$ is the cardinality of $J$. It then follows that
\beq\label{|U_J|}
|U_J|=
\begin{cases}
|J|(\dim\mf{u}+1),\qquad &r=1,\\
|J|\dim\mf{u},\qquad &r>1.
\end{cases}
\eeq
\begin{theorem}\label{thm:trivialmonodromy}
Let $\mf{g}=\tmf{g}^{(r)}$ be as in Table \ref{table:affine}. 
\begin{enumerate}
\item If $r=1$, the connection  \eqref{eq:ffopersourgauge} has trivial monodromy if and only if the the $|J|(\dim\mf{u}+1)$ scalar unknowns $U_J=\{w_j,x^\alpha(j)|\alpha\in \mr{\Delta}_{\mf{u}},j\in J\}$ satisfy the system of  of $|J|(\dim\mf{u}+1)$ equations \eqref{zeromonodromythetaj} with $j\in J$. In this case, $\dim\mf{u}=2h-3$, where $h$ is the Coxeter number of $\mf{g}$.
\item If $r>1$,  the connection  \eqref{eq:ffopersourgauge} has trivial monodromy  if and only if $J=\emptyset$.
\end{enumerate}
\end{theorem}
\begin{proof}
(1) Proved in \cite{mara18}. Note that in this case $v_\theta=e_\theta\in\rmf{g}=\tmf{g}$ and $\theta$ is the highest root of $\rmf{g}$, so that $\mf{u}=\ad_{\theta^\vee}\rmf{n}^+$.\\
(2) We consider the loop realization \eqref{looprealization} of the oper \eqref{eq:ffopersourgauge}, and assume that  $\mf{g}\neq D_4^{(3)}$, so that the trivial monodromy equations are given by  \eqref{zeromonodromythetaj}. Since $r>1$ we have $\mf{u}=\ad_{\theta^\vee}\rmf{n}^+\oplus\bb{C}v_\theta$, and  the set of unknowns $U_J$ of the trivial monodromy sistem \eqref{zeromonodromythetaj} consists of $|J|\dim\mf{u}$ scalar variables $\lbrace w_j,x^{\alpha}(j)\,|\,\alpha\in\mr{\Delta}_{\mf{u}}\rbrace_{j \in J}$.  To count the number of equations we write \eqref{zeromonodromythetaj} in a more explicit form. We fix $j\in J$ and for simplicity we drop the dependence on $j$ in the following formulae and we denote $x_i=\pi_i(x)$ ($x=\eta,a,b$). For $\mf{g}$ of type $D_{n+1}^{(2)}$, the spectrum of $\theta^\vee$ is equal to $\{-2,0,2\}$ and system \eqref{zeromonodromythetaj} is equivalent to 
\begin{subequations}\label{zeromonodromyb2}
\begin{align}
&\eta_0=0,\label{zeromonodromyb2a}\\ \label{zeromonodromyb2c}
&b_2-\frac{1}{2}[\eta_2,b_{0}]+\frac{1}{8}[\eta_2,[\eta_2,b_{-2}]]  =
[a_2-\frac{1}{2}[\eta_2,a_{0}]+\frac{1}{8}[\eta_2,[\eta_2,a_{-2}]],
a_0-\frac{1}{2}[\eta_2,a_{-2}]],
\end{align}
\end{subequations}

Condition \eqref{zeromonodromyb2a} is satisfied by choosing $\eta(j)=X(j)\in\ad_{\theta^\vee}\rmf{n}^+\subsetneq\mf{u}$. Equation \eqref{zeromonodromyb2c} consists of a $\sigma$-invariant part, which is constant in $\lambda$ and provides $\dim(\ad_{\theta^\vee}\rmf{n}^+)$ equations, and a $\sigma$ skew invariant part, which depends linearly on $\lambda$ and it is given by
\beq\label{lambdaeq2}
(k-1)w_j^{k-2}\lambda v_\theta=\left[(w_j+w_j^{k-1}\lambda) v_\theta,a_0(j)-\frac{1}{2}[\eta_2(j),a_{-2}(j)]\right].
\eeq
Since $a_0(j),\eta_2(j)$ and $a_{-2}(j)$ all belong to $\rmf{g}\cap\tmf{g}_0[\theta^\vee]$, the right hand side in \eqref{lambdaeq2} belongs to $V_\theta\cap \tmf{g}_0[\theta^\vee]=\bb{C}v_\theta$, where we used \eqref{230904-3}. It follows that \eqref{lambdaeq2} can be written as 
\beq\label{eq:231129}
(k-1)w_j^{k-2}\lambda =(w_j+w_j^{k-1}\lambda)C(j), 
\eeq
for some scalar function $C(j)$, and a simple computation shows that
$$C(j)=-\ell(\theta)+1+\sum_{i\neq j}\frac{4w_j^2}{w_j^2-w_i^2}.$$
Condition \eqref{eq:231129} has to be satisfied for every value of $\lambda$. The constant part in $\lambda$ gives the condition $w_j C(j)=0$ which implies $C(j)=0$ since $w_j\in\bb{C}^*$.  Substituting back into \eqref{eq:231129} we obtain $k=1$ which is not acceptable since $k\in(0,1)$.\\
If $\mf{g}$ is of type $A_{2n-1}^{(2)}$ or $E_6^{(2)}$, the spectrum of $\theta^\vee$ is given by $\{-2,-1,0,1,2\}$ and -- by direct inspection at \eqref{etaab} -- we have $a_{-2}=b_{-2}=0$. It follows that \eqref{zeromonodromythetaj} is equivalent to
\begin{subequations}\label{gimonodromya}
\begin{align}
&\eta_0=0\label{gimonodromy0a}\\
&2a_1-2[\eta_1, a_0]+[\eta_1,[\eta_1,a_{-1}]]-[\eta_2,a_{-1}]=0,\label{gimonodromy1a}\\
&2b_2+[\eta_1,[\eta_1,b_0]]-2[\eta_1,b_{1}]-[\eta_2,b_0]-\frac13 [\eta_1,[\eta_1,[\eta_1,b_{-1}]]]+[\eta_1,[\eta_2,b_{-1}]]\notag\\
&\,\,\, -\left[2a_2-[\eta_2,a_0]+\frac{1}{3}[\eta_1,[\eta_1,a_0]+2[\eta_2,a_{-1}]-4a_1],a_0-[\eta_1,a_{-1}]\right]=0,\label{gimonodromy2a}
\end{align}
\end{subequations}
The $\sigma$ skew-invariant part is contained in \eqref{gimonodromy2a} only, and explicitely reads
\beq\label{lambdaeq2-2}
(k-1)w_j^{k-2}\lambda v_\theta=\left[(w_j+w_j^{k-1}\lambda) v_\theta,a_0(j)-\frac{1}{2}[\eta_1(j),a_{-1}(j)]\right].
\eeq
Reasoning as in the previous case, we obtain again condition \eqref{eq:231129}, which leads to a contraddiction.\\
The proof for  $D_4^{(3)}$ is similar, the only difference being the trivial monodromy equations are given by \eqref{zeromonodromyg2theta}.
\end{proof}

\subsection{Trivial monodromy for twisted parabolic Miura $\mf{g}$-opers}
Due to the negative result proved in Theorem \ref{thm:trivialmonodromy}, it is natural to ask whether there exists a more general class of twisted parabolic Miura $\mf{g}$-opers, which provide solutions to the $QQ$-system \eqref{QQsystem} and for which the trivial monodromy conditions are non-empty.
We thus consider connections of the form
\begin{align}\label{eq:ffopersourgauge-1}
 \mc{L}(z)= & \partial_z+\frac{1}{z}\left(f+\ell-k \scaling + zv_\theta+ \sum_{j \in J} \frac{ r z^{r}}{z^r-w_j^r}\left(-\theta^\vee +\sum_{m=0}^{r-1}\left(\frac{z}{w_j}\right)^mX_m(j),
 \right)\right).
 \end{align}
 where
 $$X_m(j)\in\tmf{g}_m,\qquad m=0,\dots,r-1.$$
 Their loop realization is given by
 \begin{align}\label{looprealization-1}
 \mc{L}(z;\lambda)= & \partial_z+\frac{1}{z}\left(f+\ell + (z+z^k\lambda)v_\theta+ \sum_{j \in J} \frac{ r z^{r}}{z^r-w_j^r}\left(-\theta^\vee +\sum_{m=0}^{r-1}\left(\frac{z}{w_j}\right)^mX_m(j),
 \right)\right).
 \end{align}
 These connections are not of the form suggested in \cite{FH16}; however they coincide with them when $J=\emptyset$ (the desribe the same ground-state oper), and they a provide solutions to the same $QQ$ system \eqref{QQsystem}, if the monodromy is trivial about all $w_j, j \in J$.
 In fact, since the additional terms are subdominant at $0$ and $\infty$ and satisfies the symmetry \eqref{Lsigma1/r0},
 the theory that we have developed in Section \ref{opers and ba} apply to these more general connections.
Moreover, as we will study in Section \ref{sec:opers}, they are still representatives of twisted parabolic Miura $\mf{g}$-opers  (see \eqref{tpmgoper-general}).
We now prove that even in this more general class of connections, the trivial monodromy equations have no solutions.

The construction is similar to the one described above: we consider the following subspace of $\tmf{g}$:
\beq\label{u-1}
\tmf{u}=\ad_{\theta^\vee}\tmf{n}^+,
\eeq
and due to  \eqref{231004-1} it follows that
\beq\label{231005-1-1}
\tmf{u}=\bigoplus_{j\geq 1}\tmf{g}_j[\theta^\vee].
\eeq
Since $\rmf{g}\subset\tmf{g}$, we have
$$\mf{u}\subset\tmf{u},$$
and the inclusion is strict if $r>1$. We introduce the following basis of $\tmf{u}$. Let $\mr{\Delta}=\mr{\Delta}_{\text{long}}\cup \mr{\Delta}_{\text{short}}$ be the decomposition of the set of roots of $\rmf{g}$ into long and short roots. Define
$$\mr{\Delta}_{\tmf{u}}=\{\alpha\in\mr{\Delta}|\langle\theta^\vee,\alpha\rangle\geq1\},$$
and 
$$\mr{\Delta}_{\tmf{u},\text{short}}=\mr{\Delta}_{\tmf{u}}\cap \mr{\Delta}_{\text{short}}$$
Note the inclusions  $\mr{\Delta}_{\tmf{u},\text{short}}\subset \mr{\Delta}_{\tmf{u}}\subset \mr{\Delta}_+$, where $\mr{\Delta}_+$ is the set of positive roots of $\rmf{g}$ with respect to the Borel subalgebra $\rmf{b}^+$.  Let $\{e_\alpha\,|\,\alpha\in\mr{\Delta}_+\}$ be the basis of $\rmf{n}^+$ introduced above. Note in particular that $\theta\in\mr{\Delta}_{\tmf{u},\text{short}}$, so the element $v_\theta$ appearing in \eqref{looprealization-1} belongs to $\tmf{u}$. We can, and will from now on, fix elements $\{v_\alpha\,|\,\alpha\in\mr{\Delta}_{\tmf{u},\text{short}}\}\subset \tmf{g}_1$ (if $r>1$) and $\{\bar{v}_\alpha\,|\,\alpha\in\mr{\Delta}_{\tmf{u},\text{short}}\}\subset \tmf{g}_2$ (if $r>2$) such that
\beq\label{ubasis-1}
\tmf{u}=
\begin{cases}
\bigoplus_{\alpha\in \mr{\Delta}_{\tmf{u}}}\bb{C}e_\alpha,\quad & r=1,\\
\left(\bigoplus_{\alpha\in \mr{\Delta}_{\tmf{u}}}\bb{C}e_\alpha\right)\oplus\left(\bigoplus_{\alpha\in \mr{\Delta}_{\tmf{u},\text{short}}}\bb{C}v_\alpha\right),\quad & r=2,\\
\left(\bigoplus_{\alpha\in \mr{\Delta}_{\tmf{u}}}\bb{C}e_\alpha\right)\oplus\left(\bigoplus_{\alpha\in \mr{\Delta}_{\tmf{u},\text{short}}}\bb{C}v_\alpha\right)\oplus\left(\bigoplus_{\alpha\in \mr{\Delta}_{\tmf{u},\text{short}}}\bb{C}\bar{v}_\alpha\right),\quad & r=3.
\end{cases}
\eeq
Fix now $j\in J$ and let $x=z-w_j$, the connection $\mc{L}(z;\la)$ admits an expansion of the form \eqref{genoper1}, say
\beq\label{genoper1j-1}
\mc{L}=\partial_x+\frac{R(j)}{x}+a(j)+b(j)x+c(j)x^2+O(x^3),
\eeq
with $ R(j)=-\theta^\vee+\eta(j)$, with $\eta(j) \in \tmf{n}^+$ and $a(j),b(j),c(j)\in\tmf{g}$. The explicit form of $\eta(j)$, $a(j)$ and $b(j)$ , for $j\in J$, is given by:
\begin{subequations}\label{etaab-1}
\begin{align}
\eta(j)=&\sum_{m=0}^{r-1}X_m(j),\label{eta-1}\\
w_ja(j)=&\mr{f}+\ell+(w_j+w_j^k\la)v_\theta-\frac{r-1}{2}\theta^\vee+\sum_{m=0}^{r-1}\frac{r-1+2m}{2}X_m(j)\notag\\
&+\sum_{i\neq j}\frac{rw_j^r}{w_j^r-w_i^r}\left(-\theta^\vee+\sum_{m=0}^{r-1}\left(\frac{w_j}{w_i}\right)^mX_m(i)\right),\label{a(j)-1}\\
w_j^2b(j)=&-\mr{f}-\ell+(k-1)w_j^k\la v_\theta-\frac{(r-1)(r-5)}{12}\theta^\vee\notag\\
&+\sum_{m=0}^{r-1}\frac{6m(m+r-2)+(r-1)(r-5)}{12}X_m(j)\\
&+\sum_{i\neq j}\frac{rw_j^r\left((r-1)w_i^r+w_j^r)\right)}{w_j^r-w_i^r}\theta^\vee\notag\\
&-\sum_{i\neq j}\frac{rw_j^r}{w_j^r-w_i^r}\left(\sum_{m=0}^{r-1}\frac{(rw_i^r+(m-1)(w_i^r-w_j^r))w_j^m}{w_i^m}X_m(i)\right)
\end{align}
\end{subequations}
The term $c(j)$, which is relevant only in the $\mf{g}=D_4^{(3)}$ case, can be computed in a similar way. The coefficients $a(j)$, $b(j)$ and $c(j)$ depend on the loop variable $\lambda$, while $\eta(j)$ does not. For each $j\in J$ we consider the trivial monodromy system \eqref{zeromonodromytheta} (if $\rmf{g}$ is not of type $G_2$) or \eqref{zeromonodromyg2theta} (if $\rmf{g}$ is of type $G_2$). These equations have to be satisfied for every value of $\lambda$.

The unknowns of the trivial monodromy system \eqref{zeromonodromythetaj} are $w_j\in\bb{C}^\times$, and $X_m(j)\in\tmf{n}^+\cap \tmf{g}_m$, with $m=0,\dots,r-1$, or equivalently by $\{w_j,\eta(j)|j\in J\}$, where $\eta(j)$ is given in \eqref{eta-1}.
Recalling the definition \eqref{u-1} of the subspace $\tmf{u}$, and using \eqref{zeromonodromyathetaj} together with \eqref{231005-1-1} we obtain
$$\eta(j)\in\tmf{u},$$ 
and we can write $\eta(j)$ with respect to the basis \eqref{ubasis-1} of $\tmf{u}$  as
$$\eta(j)=\sum_{\alpha\in\mr{\Delta}_{\mf{u}}}x^\alpha(j)e_\alpha+(r-1)\sum_{\alpha\in\mr{\Delta}_{\mf{u},\text{short}}}y^\alpha(j)v_\alpha+(r-1)(r-2)\sum_{\alpha\in\mr{\Delta}_{\mf{u},\text{short}}}\bar{y}^\alpha(j)\bar{v}_\alpha,$$
for some $x^\alpha(j),y^\alpha(j),\bar{y}^\alpha(j)\in\bb{C}$. Hence, \eqref{zeromonodromythetaj} turns into a system of algebraic equations for the complex unknowns 
\beq\label{UJ-1}
\widetilde{U}_J=
\begin{cases}
\lbrace w_j,x^{\alpha}(j)\,|\,\alpha\in\mr{\Delta}_{\mf{u}}\rbrace_{j \in J},\qquad& r=1,\\
\lbrace w_j,x^{\alpha}(j),y^\beta(j)\,|\,\alpha\in\mr{\Delta}_{\mf{u}},\beta\in \mr{\Delta}_{\mf{u},\text{short}}\rbrace_{j \in J},\qquad& r=2,\\
\lbrace w_j,x^{\alpha}(j),y^\beta(j),\bar{y}^\gamma(j)\,|\,\alpha\in\mr{\Delta}_{\mf{u}},\beta,\gamma\in \mr{\Delta}_{\mf{u},\text{short}}\rbrace_{j \in J},\qquad &r=3.
\end{cases}
\eeq
Note that the number of scalar unknowns is
$$|\widetilde{U}_J|=|J|(\dim\tmf{u}+1).$$
In order to obtain a complete system of equations, we expect to obtain the same number of equations in the same number of unknowns. If $r>1$ the first negative answer is provided by the following:
\begin{proposition}\label{dimensionsarewrong}
Consider the connection \eqref{eq:ffopersourgauge-1}, and let $M$ be the number of equations for the trivial monodromy system \eqref{zeromonodromytheta} (if $\mf{g}\neq D_4^{(3)}$) or \eqref{zeromonodromyg2theta} (if $\mf{g}=D_4^{(3)}$). If $r>1$ then $M>|J|(\dim\tmf{u}+1)$, so there are more equations than unknowns.
\end{proposition}
\begin{proof}
Let $\mf{g}$ be of type $D_{n+1}^{(2)}$, so that the trivial monodromy system  \eqref{zeromonodromytheta} is given by \eqref{zeromonodromyb2}. Condition \eqref{zeromonodromyb2a} is satisfied by choosing $\eta(j)\in\mf{u}$, or equivalently, due to \eqref{eta}, choosing $X_m(j)\in\tmf{u}\cap\tmf{g}_m$, $m=0,\dots,r-1$. Equation \eqref{zeromonodromyb2c} is linear in $\lambda$: the constant part provides $|J|\dim\mf{u}$ equations, while the linear part reads:
\beq\label{lambdaeq2-1}
(k+1)w_j^{k-2}v_\theta=\left[w_j^{k-1}v_\theta,a_0(j)-\frac{1}{2}[\eta_2(j),a_{-2}(j)]\right].
\eeq
Since $X_m(j)\in\tmf{u}$ we get from \eqref{a(j)-1} that $a_0(j)\in\rmf{g}$, and therefore $[v_\theta,a_0(j)]\in V_\theta\cap\tmf{g}_2[\theta^\vee]=\bb{C}v_\theta$. The projection onto $\bb{C}v_\theta$ of \eqref{lambdaeq2-1} thus provides additional $|J|$ equations, which brings the total number of equations to $|J|(\dim\tmf{u}+1)$.  However, the element $\eta_2(j)$  in \eqref{lambdaeq2-1} contains the term $w_j^{-1}y^\theta(j)v_\theta$, and since in this case $a_{-1}(j)=f_1$, the term $[w_j^{k-1}v_\theta,[w_j^{-1}y^\theta(j)v_\theta,f_1]]$ appears. For $\rmf{g}$ of type $B_n$, $\theta-\alpha_1$ is a short root and $\theta+(\theta-\alpha_1)=\tilde{\theta}$, the highest root of $B_n$. We therefore obtain the additional condition
\beq\label{additionalBn}
w_j^{k-2}y^\theta(j)e_{\tilde{\theta}}=0. 
\eeq
Let now $\mf{g}$ be of type $A_{2n-1}^{(2)}$, so $r=2$ and $\rmf{g}=C_n$. The trivial monodromy system  \eqref{zeromonodromythetaj} is given by \eqref{gimonodromya}. Condition \eqref{gimonodromy0a} is equivalent to $\eta\in\tmf{u}$, so it is identically satisfied. The dependence on $\lambda$ of the system \eqref{gimonodromya} is only through the terms $a_2$ and $b_2$, which are linear in $\la$. Equation \eqref{gimonodromy1a} and  the constant (in $\la$) term of \eqref{gimonodromy2a} take values in $\tmf{u}$, leading to $|J|\dim\tmf{u}$ equations. The linear term in $\la$ of equation \eqref{gimonodromy2a} reads
\beq\label{lambdaeq1}
(k-1)w_j^{k-2}v_\theta=\left[w_j^{k-1}v_\theta,a_0(j)-[\eta_1(j),a_{-1}(j)]\right].
\eeq
As before, $a_0(j)\in\rmf{g}$ and therefore $[v_\theta,a_0(j)]\in V_\theta\cap\tmf{g}_2[\theta^\vee]=\bb{C}v_\theta$. Projecting \eqref{lambdaeq2-1} (for $j\in J$) onto $\bb{C}v_\theta$  gives additional $|J|$ equations, which brings the total number of equations to $|J|(\dim\tmf{u}+1)$. However, for $\rmf{g}$ of type $C_n$ we have that $\alpha_1+\alpha_2$ is a short root satisfying $\langle\theta^\vee,\alpha_1+\alpha_2\rangle=1$ so that the element $\eta_1(j)$  contains the term $y^{\alpha_1+\alpha_2}(j)v_{\alpha_1+\alpha_2}$. In this case $a_{-1}(j)=f_2$ so that in \eqref{lambdaeq1} the term $\left[w_j^{k-1}v_\theta,[y^{\alpha_1+\alpha_2}(j)v_{\alpha_1+\alpha_2},f_2]\right]$ appears. In addition $\theta+(\alpha_1+\alpha_2-\alpha_2)=\tilde{\theta}$, the highest root  of $C_n$, and we obtain the additional condition
\beq\label{additionalCn}
w_j^{k-1}y^{\alpha_1+\alpha_2}(j)e_{\tilde{\theta}}=0.
\eeq 
\\
For $\mf{g}$ of type $E_6^{(2)}$ the trivial monodromy system is again \eqref{gimonodromya}. The linear term in $\lambda$ is given by \eqref{lambdaeq1}, and also in this case we get an additional condition, as follows: for $\rmf{g}$ of type $F_4$ the root $\beta=\alpha_1+\alpha_2+\alpha_3+\alpha_4$ is short and satisfies $\langle\theta^\vee,\beta\rangle=1$. Therefore, $\eta_1(j)$  contains the term $y^{\beta}(j)v_{\beta}$. In this case $a_{-1}(j)=f_1$ so that in \eqref{lambdaeq1} the term $\left[w_j^{k-1}v_\theta,[y^{\beta}(j)v_{\beta},f_1]\right]$ appears. In addition, $\beta-\alpha_2$ is a (short) root, and $\theta+(\beta-\alpha_2)=\tilde{\theta}$, the highest root  of $F_4$. We thus obtain the additional condition $w_j^{k-1}y^{\beta}(j)e_{\tilde{\theta}}=0$. \\
The proof for $D_4^{(3)}$ is similar.
\end{proof}

\begin{corollary}\label{cor:nogoBn}
If $\mf{g}$ is of type $D_{n+1}^{(2)}$, the connection  \eqref{eq:ffopersourgauge-1} has trivial monodromy  if and only if $J=\emptyset$.
\end{corollary}
\begin{proof}
From \eqref{additionalBn} we obtain $y^\theta(j)=0$, and substituting into \eqref{zeromonodromyb2c} we get $w_j=0$, which is not acceptable since $w_j\in\bb{C}^\times$.
\end{proof}

The same result holds for $\mf{g}$ of type $A_5^{(2)}$, but the computations are more involved and will be omitted.  The idea is that from the additional condition \eqref{additionalCn} one obtains $y^{\alpha_1+\alpha_2}(j)=0$:  this time the substitution into \eqref{zeromonodromyb2c} does not give an immediate contradiction,  but considering the full system \eqref{zeromonodromyb2} eventually one obtains that there are no solutions. Due to Proposition \ref{dimensionsarewrong} and Corollary \ref{cor:nogoBn} we expect the same negative result to hold for all other $r>1$ cases.

\section{Dynkin diagram automorphisms of affine algebras}\label{atpmo}
In this section we introduce an untwisted affine Kac-Moody algebra $\bmf{g}$ and we realize the algebras $\mf{g}$, $\tmf{g}$
and $\rmf{g}$ introduced in Section \ref{sec:liealgebra1} as subalgebras of $\bmf{g}$. The algebra $\bmf{g}$ admits a Dynkin diagram automorphism which we denote as $\sigma$, which is the same letter we used for the Dynkin diagram automorphism of $\tmf{g}$ considered in Section \ref{sec:liealgebra1}. The restriction of $\sigma\in\Aut{\bmf{g}}$ to the subalgebra $\tmf{g}$ coincides with Dynkin diagram automorphism previously defined. The content of this section will be used in the next one, where we show that FFH connections are representatives of a particular class of  twisted parabolic Miura $\mf{g}-$opers.

Schematically, the Lie algebras  $\bmf{g}, \tmf{g},\mf{g}$ and $\rmf{g}$ are related as follows.
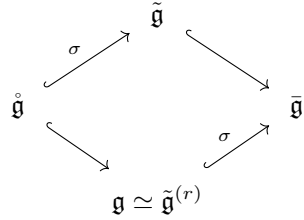
\begin{center}
\begin{tikzcd}
				& \tmf{g} \arrow[dr,hook] & \\
\rmf{g} \arrow[ur,hook,"\sigma"]\arrow[dr,hook] &		       & \bmf{g}\\
				& \mf{g}\simeq\tmf{g}^{(r)} \arrow[ur,hook,"\sigma"]
\end{tikzcd}
\end{center}
\begin{itemize}
\item The algebra $\bmf{g}$ is  one of the untwisted affine Kac-Moody algebra listed in Table \ref{table:affine}. We index the nodes of the Dynkin diagram of $\bmf{g}$ by the set
$$\bar{I}=\{0,1,\dots,\bar{n}\}.$$
The algebra $\bmf{g}$ admits a Dynkin diagram automorphism $\sigma$, of order $r$. The trivial ($r=1$) cases are shown in Table \ref{table:dynkin(r=1)}, the non-trivial ($r>1$) cases in Table \ref{table:dynkin(r>1)}. 
\item  The Lie algebra $\tmf{g}$ is the simple finite dimensional Lie algebra introduced in Section \ref{sec:liealgebra1}. The Dynkin diagram of $\tmf{g}$ can obtained from that of $\bmf{g}$ by removing the $r$ nodes in the $\sigma-$orbit $\bar{I}_0$ of the $0$-th node. Due to the numbering of the Dynkin diagrams we have chosen, the  set $\tilde{I}$ of nodes of the Dynkin diagram of $\tmf{g}$ can be obtained as
$$\tilde{I}=\bar{I}\setminus \bar{I}_0=\{1,\dots,\bar{n}+1-r\}.$$  
\item  The Lie algebra $\mf{g}=\tmf{g}^{(r)}$ is the affine Kac-Moody algebra introduced in Section \ref{sec:liealgebra1}. The Dynkin diagram of $\mf{g}$ can be obtained by folding the Dynkin diagram of $\bmf{g}$ under $\sigma$. If $r=1$ then $\sigma$ is the identity and $\mf{g}=\bmf{g}$, if $r>1$ then $\mf{g}$ is a twisted affine Kac-Moody algebra.  Denoting by $n+1$ the number of orbits, and due to the numbering of nodes in the Dynking diagram, then the set $I$ of nodes of the Dynkin diagram of $\mf{g}$ can be obtained as a subset of $\bar{I}$:
$$I=\{0,1,\dots,n\}\subset \bar{I}.$$
\item The Lie algebra $\rmf{g}$ is simple Lie algebra introduced in Section \ref{sec:liealgebra1}. As already shown, its Dynkin diagram can be obtained by removing the $0$th-vertex from the Dynkin diagram of $\mf{g}$ or, equivalently, by folding the Dynkin diagram of $\tmf{g}$ by $\sigma$. The nodes of the Dynkin diagram of $\rmf{g}$ are indexed by the set
$$\mr{I}=I\setminus\{0\}=\{1,\dots,n\}.$$
\end{itemize}


\begin{center}
\begin{table}[h]
\caption{Dynkin diagram of $\bmf{g}=\mf{g}$ (case $r=1$).}\label{table:dynkin(r=1)}
\begin{center}
\begin{tabular}{c c}

\begin{tikzpicture}
\coordinate (N1) at (0,0);
\coordinate (N2) at (1,0) ; 
\coordinate (Nn-1) at (3,0) ;
\coordinate (Nn) at (4,0);
\foreach \i in {1,2,n-1,n}{\filldraw (N\i) circle[radius=2pt] node[below]{\small{$\i$}};}

\coordinate (N0) at ($(N1)!.5!(Nn)+(0,1)$);

\draw[thick] (N2)--(N1) -- (N0) -- (Nn) -- (Nn-1);
\draw[thick] (N2) -- +(.7,0);
\draw[dotted,thick ] (N2)--(Nn-1);
\draw[thick] (Nn-1) -- +(-.7,0);

\draw[draw=black, fill=white] (N0) circle[radius=2pt] node[above]{\small{$0$}};

\node at ($(N1)!.5!(Nn)+(0,2)$) {$\bmf{g}=\mf{g}=A^{(1)}_{n}$, $(n\geq 1)$};
\end{tikzpicture}

\\
\\
\\

\begin{tikzpicture}
\coordinate (N2) at (0,0);
\coordinate (N3) at ($(N2)+(1,0)$) ; 
\coordinate (Nn-3) at ($(N3)+(2,0)$);
\coordinate (Nn-2) at ($(Nn-3)+(1,0)$);

\coordinate (N0) at ($(N2)+(135:1)$);
\coordinate (N1) at ($(N2)+(225:1)$);
\coordinate (Nn-1) at ($(Nn-2)+(45:1)$);
\coordinate (Nn) at ($(Nn-2)+(-45:1)$);

\foreach \i in {1,2,3,n-3,n-2,n}{\filldraw (N\i) circle[radius=2pt] node[below]{\small{$\i$}};}
\foreach \i in {n-1}{\filldraw (N\i) circle[radius=2pt] node[above]{\small{$\i$}};}
\draw[draw=black, fill=white] (N0) circle[radius=2pt] node[above]{\small{$0$}};

\draw[thick] (N0)--(N2);
\draw[thick] (N1)--(N2) -- (N3);
\draw[dotted,thick] (N3)--(Nn-3);
\draw[thick] (N3) -- +(.7,0);
\draw[thick] (Nn-3) -- +(-.7,0);
\draw[thick] (Nn-3)--(Nn-2)--(Nn-1);
\draw[thick] (Nn-2)--(Nn);

\node at ($(N2)!.5!(Nn-2)+(0,2)$){$\bmf{g}=\mf{g}=D^{(1)}_n$, $(n\geq 4)$};

\end{tikzpicture}

\\
\\

\begin{tikzpicture}
\coordinate (N1) at (0,0);
\coordinate (N2) at ($(N1)+(1,0)$);
\coordinate (N3) at ($(N2)+(1,0)$);
\coordinate (N4) at ($(N3)+(0,1)$);
\coordinate (N5) at ($(N3)+(1,0)$);
\coordinate (N6) at ($(N5)+(1,0)$);
\coordinate (N0) at ($(N4)+(0,1)$);

\foreach \i in {1,2,...,6}{\filldraw (N\i) circle[radius=2pt] node[below right]{$\i$};}
\draw[draw=black, fill=white] (N0) circle[radius=2pt] node[below right]{\small{$0$}};

\draw[thick] (N1)--(N2) -- (N3) -- (N5) --(N6);
\draw[thick] (N3) -- (N4) --(N0);

\node at ($(N1)!.5!(N6)+(0,3)$) {$\bmf{g}=\mf{g}=E^{(1)}_{6}$};
\end{tikzpicture}

\\
\\

\begin{tikzpicture}
\coordinate (N1) at (0,0);
\coordinate (N2) at ($(N1)+(1,0)$);
\coordinate (N3) at ($(N2)+(1,0)$);
\coordinate (N4) at ($(N3)+(1,0)$);
\coordinate (N5) at ($(N4)+(0,1)$);
\coordinate (N6) at ($(N4)+(1,0)$);
\coordinate (N7) at ($(N6)+(1,0)$);
\coordinate (N0) at ($(N7)+(1,0)$);

\foreach \i in {1,2,...,7}{\filldraw (N\i) circle[radius=2pt] node[below right]{$\i$};}
\draw[draw=black, fill=white] (N0) circle[radius=2pt] node[below right]{\small{$0$}};

\draw[thick] (N1)--(N2) -- (N3) -- (N4) --(N6) -- (N7) -- (N0);
\draw[thick] (N4) -- (N5);

\node at ($(N1)!.5!(N0)+(0,2)$) {$\bmf{g}=\mf{g}=E^{(1)}_{7}$};
\end{tikzpicture}

\\
\\

\begin{tikzpicture}
\coordinate (N0) at (0,0);
\coordinate (N1) at ($(N0)+(1,0)$);
\coordinate (N2) at ($(N1)+(1,0)$);
\coordinate (N3) at ($(N2)+(1,0)$);
\coordinate (N4) at ($(N3)+(1,0)$);
\coordinate (N5) at ($(N4)+(1,0)$);
\coordinate (N6) at ($(N5)+(0,1)$);
\coordinate (N7) at ($(N5)+(1,0)$);
\coordinate (N8) at ($(N7)+(1,0)$);

\foreach \i in {1,2,...,8}{\filldraw (N\i) circle[radius=2pt] node[below right]{$\i$};}
\draw[draw=black, fill=white] (N0) circle[radius=2pt] node[below right]{\small{$0$}};

\draw[thick] (N0) -- (N1)--(N2) -- (N3) -- (N4) --(N5) -- (N7) -- (N8);
\draw[thick] (N6) -- (N5);

\node at ($(N0)!.5!(N8)+(0,2)$) {$\bmf{g}=\mf{g}=E^{(1)}_{8}$};
\end{tikzpicture}

\end{tabular}
\end{center}
\end{table}
\end{center}


\begin{center}
\begin{table}[h]
\caption{Dynkin diagram of $\bmf{g}$ with the automorphism $\sigma$ (of order $r>1$), and the Dynkin diagram of the folded algebra $\mf{g}$.}\label{table:dynkin(r>1)}
\begin{center}
\begin{tabular}{c c}

\begin{tikzpicture}
\coordinate (N1) at (0,0);
\coordinate (N2) at (1,0) ; 
\coordinate (Nn_2) at (3,0) ;
\coordinate (Nn_1) at (4,0);
\filldraw (N1) circle[radius=2pt] node[above right]{\small{$1$}};
\filldraw (N2) circle[radius=2pt] node[above]{\small{$2$}};
\filldraw (Nn_2) circle[radius=2pt] node[below]{\small{$n-2$}};
\filldraw (Nn_1) circle[radius=2pt] node[above]{\small{$n-1$}};
\draw[thick] (N1)--(N2);
\draw[thick] (N2) -- +(.7,0);
\draw[dotted,thick ] (N2)--(Nn_2);
\draw[thick] (Nn_2) -- +(-.7,0);
\draw[thick] (Nn_2) -- (Nn_1);

\coordinate (N0) at ($(N1)+(135:1)$);
\draw[draw=black, fill=white] (N0) circle[radius=2pt] node[above]{\small{$0$}};
\draw[thick] (N0)--(N1);

\coordinate (Nn+2) at ($(N1)+(225:1)$);
\filldraw (Nn+2) circle[radius=2pt] node[below]{\small{$n+2$}};
\draw[thick] (Nn+2)--(N1);

\coordinate (Nn) at ($(Nn_1)+(45:1)$);
\filldraw (Nn) circle[radius=2pt] node[above]{\small{$n$}};
\draw[thick] (Nn_1)--(Nn);

\coordinate (Nn+1) at ($(Nn_1)+(-45:1)$);
\filldraw (Nn+1) circle[radius=2pt] node[below]{\small{$n+1$}};
\draw[thick] (Nn_1)--(Nn+1);

\draw [<->, >=stealth, dashed, thin](N0) .. controls ($(N0)!.5!(Nn+2)+(-.2,0)$)   .. (Nn+2) ;
\node[left] at ($(N0)!.5!(Nn+2)+(-.2,0)$){$\sigma$};
\draw [<->, >=stealth, dashed, thin](Nn) .. controls ($(Nn)!.5!(Nn+1)+(.2,0)$)   .. (Nn+1) ;
\node at ($(N1)!.5!(Nn_1)+(0,1.5)$) {$\bmf{g}=D^{(1)}_{n+2}$, $(n\geq 2)$};
\end{tikzpicture}

&

\begin{tikzpicture}
\coordinate (N0) at (-1,0);
\coordinate (N1) at (0,0);
\coordinate (N2) at (1,0) ; 
\coordinate (Nn_2) at (3,0) ;
\coordinate (Nn_1) at (4,0);
\coordinate (Nn) at (5,0);
\node at ($(N1)!.5!(Nn_1)+(0,1.5)$) {$\mf{g}=D^{(2)}_{n+1}$, $(n\geq 2)$};
\draw[draw=black, fill=white] (N0) circle[radius=2pt] node[above]{\small{$0$}};
\filldraw (N1) circle[radius=2pt] node[above]{\small{$1$}};
\filldraw (N2) circle[radius=2pt] node[above]{\small{$2$}};
\filldraw (Nn_2) circle[radius=2pt] node[above]{\small{$n-2$}};
\filldraw (Nn_1) circle[radius=2pt] node[below]{\small{$n-1$}};
\filldraw (Nn) circle[radius=2pt] node[above]{\small{$n$}};
\draw[thick] (N1)--(N2);
\draw[thick] (N2) -- +(.7,0);
\draw[dotted,thick ] (N2)--(Nn_2);
\draw[thick] (Nn_2) -- +(-.7,0);
\draw[thick] (Nn_2) -- (Nn_1);

\draw[->,-stealth,thick,double] (Nn_1) -- (Nn);
\draw[->,-stealth,thick,double] (N1) -- (N0);
\end{tikzpicture}
\\
\\
\\

\begin{tikzpicture}
\coordinate (Nn) at (0,0);
\coordinate (Nn+1) at ($(Nn)+(120:1)$);
\coordinate (N2n_3) at ($(Nn+1)+(-2,0)$);
\coordinate (N2n_2) at ($(N2n_3)+(-1,0)$);
\coordinate (Nn_1) at ($(Nn)+(240:1)$);
\coordinate (N3) at ($(Nn_1)+(-2,0)$);
\coordinate (N2) at ($(N3)+(-1,0)$);
\filldraw (Nn) circle[radius=2pt] node[above right]{\small{$n$}};
\filldraw (Nn+1) circle[radius=2pt] node[above]{\small{$n+1$}};
\filldraw (N2n_3) circle[radius=2pt] node[above]{\small{$2n-3$}};
\filldraw (N2n_2) circle[radius=2pt] node[above]{\small{$2n-2$}};
\filldraw (Nn_1) circle[radius=2pt] node[below]{\small{$n-1$}};
\filldraw (N3) circle[radius=2pt] node[below]{\small{$3$}};
\filldraw (N2) circle[radius=2pt] node[below]{\small{$2$}};

\draw[thick] (Nn) -- (Nn_1);
\draw[thick] (Nn_1) -- +(-.7,0);
\draw[dotted,thick] (Nn_1) -- (N3);
\draw[thick] (N3) -- +(.7,0);
\draw[thick] (N2) -- (N3);
\draw[thick] (Nn) -- (Nn+1);
\draw[thick] (Nn+1) -- +(-.7,0);
\draw[dotted,thick] (Nn+1) -- (N2n_3);
\draw[thick] (N2n_3) -- +(.7,0);
\draw[thick] (N2n_2) -- (N2n_3);

\coordinate (N0) at ($(N2)+(150:1)$);
\draw[draw=black, fill=white] (N0) circle[radius=2pt] node[below]{\small{$0$}};
\draw[thick] (N0)--(N2);

\coordinate (N1) at ($(N2)+(210:1)$);
\filldraw (N1) circle[radius=2pt] node[below]{\small{$1$}};
\draw[thick] (N1)--(N2);

\coordinate (N2n_1) at ($(N2n_2)+(150:1)$);
\filldraw (N2n_1) circle[radius=2pt] node[above]{\small{$2n-1$}};
\draw[thick] (N2n_1)--(N2n_2);

\coordinate (N2n) at ($(N2n_2)+(210:1)$);
\filldraw (N2n) circle[radius=2pt] node[above]{\small{$2n$}};
\draw[thick] (N2n)--(N2n_2);

\draw [<->, >=stealth, dashed, thin](Nn+1) .. controls ($(Nn+1)!.5!(Nn_1)+(-.2,0)$)   .. (Nn_1) ;
\draw [<->, >=stealth, dashed, thin](N2n_3) .. controls ($(N2n_3)!.5!(N3)+(-.2,0)$)   .. (N3) ;
\node[right] at ($(N2n_3)!.5!(N3)+(-.2,0)$){$\sigma$};
\draw [<->, >=stealth, dashed, thin](N2n_2) .. controls ($(N2n_2)!.5!(N2)+(-.2,0)$)   .. (N2) ;
\draw [<->, >=stealth, dashed, thin](N2n) .. controls ($(N2n)!.5!(N0)+(-.2,0)$)   .. (N0) ;
\draw [<->, >=stealth, dashed, thin](N2n_1) .. controls ($(N2n_1)!.5!(N1)+(-.7,0)$)   .. (N1) ;

\node at ($($(N2n_1)!.5!(N1)+(-.7,0)$)!.5!(Nn)+(0,2)$) {$\bmf{g}=D^{(1)}_{2n}$, $(n\geq 3)$};

\end{tikzpicture}

&

\begin{tikzpicture}
\coordinate (Nn) at (0,0);
\coordinate (Nn_1) at ($(Nn)+(-1,0)$);
\coordinate (N3) at ($(Nn_1)+(-2,0)$);
\coordinate (N2) at ($(N3)+(-1,0)$);
\filldraw (Nn) circle[radius=2pt] node[above]{\small{$n$}};
\filldraw (Nn_1) circle[radius=2pt] node[above]{\small{$n-1$}};
\filldraw (N3) circle[radius=2pt] node[above]{\small{$3$}};
\filldraw (N2) circle[radius=2pt] node[above right]{\small{$2$}};

\draw[->,-stealth,thick,double] (Nn) -- (Nn_1);
\draw[thick] (Nn_1) -- +(-.7,0);
\draw[dotted,thick] (Nn_1) -- (N3);
\draw[thick] (N3) -- +(.7,0);
\draw[thick] (N2) -- (N3);

\coordinate (N0) at ($(N2)+(150:1)$);
\draw[draw=black, fill=white] (N0) circle[radius=2pt] node[above right]{\small{$0$}};
\draw[thick] (N0)--(N2);

\coordinate (N1) at ($(N2)+(210:1)$);
\filldraw (N1) circle[radius=2pt] node[below right]{\small{$1$}};
\draw[thick] (N1)--(N2);

\node at ($($(N0)!.5!(N1)$)!.5!(Nn)+(0,2)$) {$\mf{g}=A^{(2)}_{2n-1}$, $(n\geq 3)$};

\end{tikzpicture}
\\
\\

\begin{tikzpicture}
\coordinate (N3) at (0,0);
\coordinate (N4) at ($(N3)+(1,0)$);
\coordinate (N2) at ($(N3)+(135:1)$);
\coordinate (N1) at ($(N2)+(-1,0)$);
\coordinate (N0) at ($(N1)+(-1,0)$);
\coordinate (N5) at ($(N3)+(225:1)$);
\coordinate (N6) at ($(N5)+(-1,0)$);
\coordinate (N7) at ($(N6)+(-1,0)$);
\draw[draw=black, fill=white] (N0) circle[radius=2pt] node[above right]{\small{$0$}};
\filldraw (N1) circle[radius=2pt] node[above]{\small{$1$}};
\filldraw (N2) circle[radius=2pt] node[above]{\small{$2$}};
\filldraw (N3) circle[radius=2pt] node[above]{\small{$3$}};
\filldraw (N4) circle[radius=2pt] node[below]{\small{$4$}};
\filldraw (N5) circle[radius=2pt] node[below]{\small{$5$}};
\filldraw (N6) circle[radius=2pt] node[below]{\small{$6$}};
\filldraw (N7) circle[radius=2pt] node[below]{\small{$7$}};

\draw[thick] (N0) -- (N1);
\draw[thick] (N1) -- (N2);
\draw[thick] (N2) -- (N3);
\draw[thick] (N3) -- (N4);
\draw[thick] (N3) -- (N5);
\draw[thick] (N5) -- (N6);
\draw[thick] (N6) -- (N7);

\draw [<->, >=stealth, dashed, thin](N2) .. controls ($(N2)!.5!(N5)+(-.2,0)$)   .. (N5) ;
\draw [<->, >=stealth, dashed, thin](N1) .. controls ($(N1)!.5!(N6)+(-.2,0)$)   .. (N6) ;
\node[left] at ($(N1)!.5!(N6)+(-.2,0)$){$\sigma$};
\draw [<->, >=stealth, dashed, thin](N0) .. controls ($(N0)!.5!(N7)+(-.2,0)$)   .. (N7) ;

\node at ($($(N0)!.5!(N7)+(-.8,0)$)!.5!(N4)+(0,1.5)$) {$\bmf{g}=E^{(1)}_{7}$};

\end{tikzpicture}

&

\begin{tikzpicture}
\coordinate (N3) at (0,0);
\coordinate (N4) at ($(N3)+(1,0)$);
\coordinate (N2) at ($(N3)+(-1,0)$);
\coordinate (N1) at ($(N2)+(-1,0)$);
\coordinate (N0) at ($(N1)+(-1,0)$);

\foreach \i in {1,...,4}{\filldraw (N\i) circle[radius=2pt] node[above]{\small{$\i$}};}
\draw[draw=black, fill=white] (N0) circle[radius=2pt] node[above]{\small{$0$}};

\draw[thick] (N0) -- (N1);
\draw[thick] (N1) -- (N2);
\draw[->,-stealth,thick,double] (N3) -- (N2);
\draw[thick] (N3) -- (N4);

\node at ($(N0)!.5!(N4)+(0,1.5)$) {$\mf{g}=E^{(2)}_{6}$};

\end{tikzpicture}
 \\
 \\

\begin{tikzpicture}
\coordinate (N2) at (0,0);
\coordinate (N4) at ($(N2)+(-1,0)$);
\coordinate (N0) at ($(N4)+(-1,0)$);
\coordinate (N3) at ($(N2)+(120:1)$);
\coordinate (N6) at ($(N3)+(-1,0)$);
\coordinate (N1) at ($(N2)+(240:1)$);
\coordinate (N5) at ($(N1)+(-1,0)$);

\foreach \i in {2,3,4,6}{\filldraw (N\i) circle[radius=2pt] node[above right]{\small{$\i$}};}
\draw[draw=black, fill=white] (N0) circle[radius=2pt] node[above right]{\small{$0$}};
\foreach \i in {1,5}{\filldraw (N\i) circle[radius=2pt] node[below right]{\small{$\i$}};}

\draw[thick] (N2) -- (N3);
\draw[thick] (N3) -- (N6);
\draw[thick] (N2) -- (N4);
\draw[thick] (N4) -- (N0);
\draw[thick] (N1) -- (N2);
\draw[thick] (N5) -- (N1);

\draw [->, >=stealth, dashed, thin](N6) .. controls ($(N6)!.5!(N0)+(-.2,0)$)   .. (N0) ;
\draw [->, >=stealth, dashed, thin](N0) .. controls ($(N0)!.5!(N5)+(-.2,0)$)   .. (N5) ;
\draw [<-, >=stealth, dashed, thin](N6) .. controls ($(N6)!.5!(N5)+(.3,0)$)   .. (N5) ;
\draw [->, >=stealth, dashed, thin](N3) .. controls ($(N3)!.5!(N4)+(-.2,0)$)   .. (N4) ;
\draw [->, >=stealth, dashed, thin](N4) .. controls ($(N4)!.5!(N1)+(-.2,0)$)   .. (N1) ;
\draw [<-, >=stealth, dashed, thin](N3) .. controls ($(N3)!.5!(N1)+(.3,0)$)   .. (N1) ;

\node[above left] at ($(N6)!.5!(N0)+(-.2,0)$){$\sigma$};

\node at ($(N0)!.5!(N2)+(0,1.7)$) {$\bmf{g}=E^{(1)}_{6}$};

\end{tikzpicture}

&

\begin{tikzpicture}
\tikzset{Rightarrow/.style={double equal sign distance,>={Implies},->},
triple/.style={-,preaction={draw,Rightarrow}},
quadruple/.style={preaction={draw,Rightarrow,shorten >=0pt},shorten >=1pt,-,double,double
distance=0.2pt}}

\coordinate (N2) at (0,0);
\coordinate (N1) at ($(N2)+(-1,0)$);
\coordinate (N0) at ($(N1)+(-1,0)$);

\draw[draw=black, fill=white] (N0) circle[radius=2pt] node[above]{\small{$0$}};
\filldraw (N1) circle[radius=2pt] node[above]{\small{$1$}};
\filldraw (N2) circle[radius=2pt] node[above]{\small{$2$}};

\draw[thick] (N0) -- (N1);
\draw[thick,->,-stealth,triple] (N2) -- (N1);

\node at ($(N0)!.5!(N2)+(0,1.7)$) {$\mf{g}=D^{(3)}_{4}$};
\end{tikzpicture}

\end{tabular}
\end{center}
\end{table}
\end{center}

\subsection{The untwisted affine Kac-Moody algebra $\bmf{g}$} 
Let $\bmf{g}$ one of the Kac--Moody algebra of  affine type in Table \ref{table:affine}, with Dynkin diagram labeled as in Table \ref{table:dynkin(r=1)} and Table \ref{table:dynkin(r>1)}. Note that $\bmf{g}$ is always an untwisted affine Lie algebra and it is simply-laced, namely its  Cartan matrix is symmetric.  Let $\bar{I}=\{0,\dots,\bar{n}\}$ denote the set of nodes of the Dynkin diagram of $\bmf{g}$. Let $\{\bar{f}_i,\bar{\alpha}^\vee_i,\bar{e}_i\,|\,i\in\bar{I}\}$  be Chevalley generators of $\bmf{g}$, satisfying the relations ($i,j\in\bar{I}$):
$$[\bar{\alpha}^\vee_i,\bar{\alpha}^\vee_j]=0,\quad[\bar{\alpha}^\vee_i,\bar{e}_j]=\bar{C}_{ij}\bar{e}_j,\quad [\bar{\alpha}^\vee_i,\bar{f}_j]=-\bar{C}_{ij}\bar{f}_j,\quad [\bar{e}_i,\bar{f}_j]=\delta_{ij}\bar{\alpha}^\vee_i,$$
where $\bar{C}=(\bar{C}_{ij})_{i,j\in\bar{I}}$ is the Cartan matrix of $\bmf{g}$. We define
\beq\label{220717-3}
f=\sum_{i\in\bar{I}}\bar{f}_i\in\bmf{g}.
\eeq
Fix a Cartan subalgebra $\bmf{h}$ of $\bmf{g}$ such that $\{\bar{\alpha}^\vee_i\,|\,i\in\bar{I}\}\subset\bmf{h}$ is the set of simple coroots and denote by $\{\bar{\alpha}_i\,|\,i\in\bar{I}\}\subset\bmf{h}^\ast$ the corresponding set of simple roots, such that $\langle\bar{\alpha}^\vee_i,\bar{\alpha}_j\rangle=\bar{C}_{ij}$. Fix a scaling element $\bar{\scaling}\in\bmf{h}$ satisfying $\langle\bar{\scaling},\bar{\alpha}_i\rangle=\delta_{i0}$. Then 
\beq\label{220718-3}
\bmf{h}=\langle\bar{\alpha}^\vee_0,\dots,\bar{\alpha}^\vee_{\bar{n}},\bar{\scaling}\rangle. 
\eeq
Let $\{\bar{a}_i\,|\,i\in\bar{I}\}$ be Kac's lables of $\bmf{g}$ (these are equal to the dual Kac's labels, since $\bmf{g}$ is simply laced), namely a set of positive integers defined by the relations
\beq\label{220729-1}
\sum_{j\in\bar{I}}\bar{C}_{ij}\bar{a}_j=0,\qquad i\in\bar{I},
\eeq  
and which we normalize by setting $\bar{a}_0=1$. The element of $\bmf{h}$ given by
\beq\label{220718-1}
K=\sum_{i\in\bar{I}}\bar{a}_i\bar{\alpha}^\vee_i
\eeq
is the canonical central element of $\bmf{g}$. Denote by $\bmf{n}^+$ (resp. $\bmf{n}^-$) the subalgebra of $\bmf{g}$ generated by $\{\bar{e}_i\,|\,i\in\bar{I}\}$ (resp. $\{\bar{f}_i\,|\,i\in\bar{I}\}$), so that $\bmf{g}$ admits the triangular decomposition
\beq\label{220717-6}
\bmf{g}=\bmf{n}^-\oplus\bmf{h}\oplus\bmf{n}^+.
\eeq
Denote by $\bmf{b}_\pm=\bmf{h}\oplus\bmf{n}_\pm$ the corresponding Borel subalgebras. Let $\bar{\Delta}\subset\bmf{h}^\ast$ be the set of roots of $\bmf{g}$, so that
\beq\label{220728-2}
\bmf{g}=\bmf{h}\oplus\bigoplus_{\alpha\in\bar{\Delta}}\bmf{g}_\alpha,
\eeq
with $\bmf{g}_\alpha$ the corresponding root space. Let $\{\bar{\omega}^\vee_i\,|\,i\in\bar{I}\}\subset\bmf{h}$ be fundamental coweights of $\bmf{g}$, satisfying $\langle \bar{\omega}^\vee_i,\bar{\alpha}_j\rangle=\delta_{ij}$ ($i,j\in\bar{I}$). In particular, we set $\bar{\omega}^\vee_0=\bar{\scaling}$. Consider the element
\beq\label{220717-5}
\rho^\vee=\sum_{i\in\bar{I}}\bar{\omega}^\vee_i\in P^\vee({\bmf{g}}).
\eeq
The adjoint action of $\rho^\vee$ defines the principal gradation on $\bmf{g}$:
\beq\label{220722-2}
\bmf{g}=\bigoplus_{j\in\bb{Z}}\bmf{g}^j,\qquad \bmf{g}^j=\{x\in\bmf{g}\,|\,[\rho^\vee,x]=jx\}.
\eeq
In particular,
\beq\label{2207118-5}
\bmf{g}^{-1}=\bigoplus_{j\in\bar{I}}\bb{C}\bar{f}_j,
\eeq
so that $f$ defined in \eqref{220717-3} belongs to $\bmf{g}^{-1}$. For $k\in\bb{Z}$ we denote
\beq\label{220722-3}
\left(\bmf{g}\right)^{\geq k}=\bigoplus_{j\geq k}\bmf{g}^j,
\eeq
from which it follows that $\left(\bmf{g}\right)^{\geq 1}=\bmf{n}^+$, while $\left(\bmf{g}\right)^{\geq 0}=\bmf{h}\oplus\bmf{n}^+=\bmf{b}^+$  and
\beq\label{220718-2}
\left(\bmf{g}\right)^{\geq-1}=\left(\bigoplus_{j\in\bar{I}}\bb{C}\bar{f}_j\right)\oplus \bmf{b}^+.
\eeq

\subsection{Dynkin diagram automorphism of $\bmf{g}$ and the folded affine algebra $\mf{g}$}
Let $\sigma$ be a Dynkin diagram automorphism of $\bmf{g}$, namely a permutation of the set $\bar{I}$ such that $\bar{C}_{\sigma(i)\sigma(j)}=\bar{C}_{ij}$. Extend $\sigma$ to a Lie algebra automorphism (still denoted by $\sigma$) $\sigma\in\Aut(\bmf{g})$ defined on Chevalley generators by ($i\in\bar{I}$) 
\beq\label{220723-2}
\sigma(\bar{e}_i)=\bar{e}_{\sigma(i)},\quad \sigma(\bar{\alpha}^\vee_i)=\bar{\alpha}^\vee_{\sigma(i)},\quad \sigma(\bar{f}_i)=\bar{f}_{\sigma(i)}.
\eeq
The automorphism $\sigma$ induces the following gradation on $\bmf{g}$:
\beq\label{220715-2}
\bmf{g}=\bigoplus_{\ell=0}^{r-1}\bmf{g}_\ell,\qquad \bmf{g}_\ell=\left\{x\in \bmf{g}\,|\,\sigma(x)=\e^\ell x\right\}
\eeq
where
\beq\label{epsilon}
\e=e^\frac{2\pi i}{r}
\eeq
and $r$ is the order of $\sigma$.
\begin{remark}
Dynkin diagram automorphisms of affine Kac-Moody algebras, whose folded diagram is the Dynkin diagram of an affine Kac-Moody algebra were considered by Liu, Wu, Zhang and Zhou \cite{lwzz20}. It was proved in the same paper that the fixed point subalgebra $\bmf{g}_0$ does not need to be isomorphic to the affine algebra $\mf{g}$ associated to the folded diagram. Indeed, one has $\bmf{g}_0=\mf{g}\oplus\mc{H}$, where $\mc{H}$ is either $\{0\}$ or a proper vector subspace of the principal Heisenberg subalgebra of $\bmf{g}$.  The subspace $\mc{H}$ plays no role in the present construction, since (as shown below) the $\bmf{g}_0-$part of the opers we are interested in lies in $\mf{g}$.
\end{remark}
 The folded affine Lie algebra $\mf{g}$ obtained from the Dynkin diagram of $\bmf{g}$ is an affine algebra of type $\tmf{g}^{(r)}$ and therefore it is  twisted if $r>1$, see Table \ref{table:dynkin(r>1)}. The structure of $\mf{g}$ can be obtained from that of $\bmf{g}$ as follows. We denote by $\bar{I}_j\subset\bar{I}$ the $\sigma$-orbit of $j\in\bar{I}$, and let $n+1$ be the number of orbits. Due to the numbering of the Dynkin diagrams in Table \ref{table:dynkin(r>1)} we have that the nodes $0,\dots,n$ lie in different orbits, so we can represent the set of orbits by $I=\{0\dots,n\}\subset \bar{I}$. For $i\in\bar{I}$ we denote by $\langle i\rangle\in\bb{Z}^+$ the cardinality of the $\sigma$-orbit $\bar{I}_i$. The Cartan matrix $C=(C_{ij})_{i,j\in I}$ of $\mf{g}$ can be obtained summing over the columns of $\bar{C}$ along the orbits of $\sigma$:
$$C_{ij}=\sum_{\ell=1}^{\langle i\rangle}\bar{C}_{\sigma^\ell(i)j},\qquad i,j\in I.$$
 The elements
\beq\label{220717-4}
e_i=\sum_{\ell=1}^{\langle i\rangle}\bar{e}_{\sigma^\ell(i)},\quad \alpha^\vee_i=\sum_{\ell=1}^{\langle i\rangle}\bar{\alpha}^\vee_{\sigma^\ell(i)},\quad f_i=\sum_{\ell=1}^{\langle i\rangle}\bar{f}_{\sigma^\ell(i)},\qquad i\in I.
\eeq
satisfy the relations ($i,j\in I$)
$$[\alpha^\vee_i,\alpha^\vee_j]=0,\quad[\alpha^\vee_i,e_j]=C_{ij}e_j,\quad [\alpha^\vee_i,f_j]=-C_{ij}f_j,\quad [e_i,f_j]=\delta_{ij}\alpha^\vee_i,$$
together with the Serre relations \cite{lwzz20}, and are therefore Chevalley generators of $\mf{g}$. In addition, the elements
\beq\label{221026-1}
\alpha_i=\frac{1}{\langle i\rangle}\sum_{\ell=1}^{\langle i\rangle}\bar{\alpha}_{\sigma^\ell(i)},\qquad i\in I
\eeq
are simple roots of $\mf{g}$, namely they satisfy $\langle \alpha_i^\vee,\alpha_j\rangle=C_{ij}$ ($i,j\in I$).  The following elementary result will be useful later:
\begin{lemma}\label{lemma220729-2}
Recall the Kac's labels $\{\bar{a}_i\,|\,i\in\bar{I}\}$ of $\bmf{g}$. The quantities
$$a_i=
\frac{\langle i\rangle}{r}\bar{a}_i,\qquad a^\vee_i=
\bar{a}_i,\qquad i\in I$$
are, respectively, Kac's labels and dual Kac's labels of $\mf{g}$, satisfying the relations \eqref{220729-2/0}. In particular, $a_0=a^\vee_0=1$. 
\end{lemma}
\begin{proof}
From the relation \eqref{220729-1} and the invariance of $\bar{C}$ under $\sigma$ it follows that the $\bar{a}_i$'s are $\sigma-$invariant: $\bar{a}_{\sigma(i)}=\bar{a}_i$ ($i\in\bar{I}$). Moreover, from \eqref{220729-1} for $i\in\bar{I}$ one has
$$0=\sum_{j\in\bar{I}}\bar{C}_{ij}\bar{a}_j=\sum_{j\in I}\sum_{\ell=1}^{\langle j \rangle}\bar{C}_{i,\sigma^\ell(j)}\bar{a}_{\sigma^\ell(j)}=\sum_{j\in I}\left(\sum_{\ell=1}^{\langle j \rangle}\bar{C}_{i,\sigma^\ell(j)}\right)\bar{a}_{j},$$
and summing over the $i$-th orbit we get ($i\in I$):
\begin{align*}
0&=\sum_{m= 1}^{\langle i\rangle}\sum_{j\in\bar{I}}\bar{C}_{\sigma^m(i),j}\bar{a}_j=\sum_{m= 1}^{\langle i\rangle}\sum_{j\in I}\left(\sum_{\ell=1}^{\langle j \rangle}\bar{C}_{\sigma^m(i),\sigma^\ell(j)}\right)\bar{a}_{j}\\
&=\sum_{j\in I}\sum_{\ell=1}^{\langle j \rangle}\left(\sum_{m= 1}^{\langle i\rangle}\bar{C}_{\sigma^{m-\ell}(i),j}\right)\bar{a}_{j}=\sum_{j\in I} C_{ij}\langle j \rangle\bar{a}_{j}.
\end{align*}
By a direct (case-by-case) inspection we obtain that $\langle 0\rangle=r$, so that the quantities $a_j=\langle j \rangle \bar{a}_j/r$ ($j\in I$) satisfy \eqref{220729-2/0} and $a_0=1$. For the $a^\vee_i$'s the proof is similar.
\end{proof}
We set
\beq\label{221021-1}
D_i=\frac{a_i}{a^\vee_i},\qquad i\in I.
\eeq
Due to the above lemma, this is consistent with \eqref{Di/0}. We now show that certain elements of $\bmf{g}$ introduced above are $\sigma$-invariant, and - moreover - they belong to $\mf{g}\subset\bmf{g}$. 
\begin{enumerate}[i)]
\item The element $f\in\bmf{g}$ defined in \eqref{220717-3} is $\sigma$-invariant (i.e. $\sigma(f)=f$), and recalling the elements $f_i\in\mf{g}$ introduced in \eqref{220717-4} we have
\beq\label{220721-1}
f=\sum_{i\in\bar{I}}\bar{f}_i=\sum_{j\in I}\left(\sum_{\ell=1}^{\langle j\rangle}\bar{f}_{\sigma^\ell(j)}\right)=\sum_{j\in I}f_j
\eeq
from which it follows that $f\in\mf{g}$. In particular, it coincides with \eqref{f}.
\item Similarly, the canonical central element $K\in\bmf{h}$ of $\bmf{g}$ given by \eqref{220718-1} is $\sigma$-invariant ($\sigma(K)=K$) and, using Lemma \ref{lemma220729-2}, it turns out to be equal to the canonical central element $K\in\mf{h}$ of $\mf{g}$:
\beq
K=\sum_{i\in I}a^\vee_i \alpha^\vee_i.
\eeq
\item  The fundamental coweights $\{\omega^\vee_i\,|\,i\in I\}$ of $\mf{g}$ can be written in terms of those of $\bmf{g}$ as
\beq\label{221026-4}
\omega^\vee_i=\sum_{\ell=1}^{\langle i\rangle}\bar{\omega}^\vee_{\sigma^\ell(i)},\qquad i\in I,
\eeq
and they satisfy $\langle \omega_i^\vee,\alpha_j\rangle=\delta_{ij}$ ($i,j\in I$). We denote the coweight lattice and the coroot lattice of $\mf{g}$, respectively, as 
$$P^\vee(\mf{g})=\bigoplus_{i\in I}\bb{Z}\omega^\vee_i,\qquad Q^\vee(\mf{g})=\bigoplus_{i\in I}\bb{Z}\alpha^\vee_i.$$
Recalling that $\langle 0\rangle=r$, we denote 
\beq\label{220718-4}
\scaling=\omega^\vee_0=\sum_{\ell=1}^{r}\bar{\omega}^\vee_{\sigma^\ell(0)}=\bar{\scaling}+\sum_{\ell=2}^{r}\bar{\omega}^\vee_{\sigma^\ell(0)},
\eeq
where $\bar{\scaling}$ is the scaling element of $\bmf{g}$. The element $\scaling\in P^\vee(\mf{g})$, is manifestly $\sigma-$invariant, satisfies $\langle\scaling,\alpha_j\rangle=\delta_{0j}$ ($j\in I$), and it is therefore a scaling element of $\mf{g}$. Due to \eqref{220718-4} we can write the Cartan subalgebra \eqref{220718-3} of $\bmf{h}$ equivalently as
\beq\label{220718-5}
\bmf{h}=\langle\bar{\alpha}^\vee_0,\dots,\bar{\alpha}^\vee_{\bar{n}},\scaling\rangle. 
\eeq
The Cartan subalgebra $\mf{h}=\bmf{h}\cap\mf{g}$ of $\mf{g}$ is thus given by 
\beq\label{220720-7}
\mf{h}=\langle\alpha^\vee_0,\dots,\alpha^\vee_n,\scaling\rangle,
\eeq
where the $\alpha^\vee_i$'s are given by \eqref{220717-4}.
\item The element $\rho^\vee\in \bmf{h}$ defined by \eqref{220717-5} is fixed by $\sigma$ (i.e. $\sigma(\rho^\vee)=\rho^\vee$) and it can be written as the sum of the fundamental coweights of $\mf{g}$:
\beq\label{220721-4}
\rho^\vee=\sum_{i\in I}\omega^\vee_i.
\eeq
In particular, $\rho^\vee\in P^\vee(\mf{g})\subset \mf{h}$. 
\end{enumerate}

\subsection{The simple Lie algebras $\tmf{g}$ and $\rmf{g}$ as subalgebras of $\bmf{g}$.}
We can recover the simple Lie algebras  $\tmf{g}$ and $\rmf{g}$ introduced in Section \ref{sec:liealgebra1} as subalgebras of $\bmf{g}$ and $\mf{g}$. First, note that we can write the set $\tilde{I}$, the nodes of $\tmf{g}$ as $\tilde{I}=\bar{I}\setminus\bar{I}_0$. In other words, $\tilde{I}$ is the subset of nodes in the Dynkin diagram of $\bmf{g}$ with the orbit of the $0-$th node removed.  The subalgebra of $\bmf{g}$ generated by $\{\bar{e}_i,\bar{\alpha}^\vee_i,\bar{f}_i\,|\,i\in\tilde{I}\}$ is isomorphic to the simple finite dimensional Lie algebra $\tmf{g}$; in particular, $\bar{C}_{ij}=\tilde{C}_{ij}$ ($i,j\in\tilde{I}$). We fix this isomorphims by identifying $\tilde{e}_i=\bar{e}_i$, $\tilde{\alpha_i}^\vee=\bar{\alpha}_i^\vee$, $\tilde{f}_i=\bar{f}_i$ ($i\in\tilde{I}$). Since we obtained $\rmf{g}$ as a subalgebra of $\tmf{g}$ we obtain the following inclusions of Lie subalgebras:
$$\rmf{g}\subseteq \tmf{g}\subset\bmf{g},$$
where $\rmf{g}= \tmf{g}$ in the $r=1$ case. On the other hand, the Dynkin diagram of the simple Lie algebra $\rmf{g}$ can also be obtained from the Dynkin diagram of the affine algebra $\mf{g}$ by removing the $0-$th vertex ($\mr{I}=I\setminus\{0\}$) so that we can view   $\rmf{g}$ as a subalgebra of $\mf{g}$. Since $\mf{g}$ is a subalgebra of $\tmf{g}$, we obtain another set of inclusions
$$\rmf{g}\subset \mf{g}\subseteq\bmf{g},$$
where $\mf{g}=\bmf{g}$ in the $r=1$ case. The above inclusions are consistent with the action of $\sigma$ and with the triangular decomposition: the restriction of the Dynkin diagram automorphism $\sigma\in\Aut(\bmf{g})$ to the subalgebra $\rmf{g}$ coincides with the Dynkin diagram automorphism $\sigma\in\Aut(\tmf{g})$ introduced in Section \ref{sec:liealgebra1}, and we have the inclusions
\begin{gather}
\rmf{n}^\pm\subset\tmf{n}^\pm\subset \bmf{n}^\pm,\qquad \rmf{h}\subset\tmf{h}\subset \bmf{h},\label{231009-3}\\
\rmf{n}^\pm\subset\mf{n}^\pm\subset \bmf{n}^\pm,\qquad \rmf{h}\subset\mf{h}\subset \bmf{h}.
\end{gather}
It is a result of Kostant \cite{Kos59} that the Borel subalgebra $\rmf{b}^+=\rmf{h}\oplus \rmf{n}^+$ of $\rmf{g}$ admits the decomposition
\beq\label{230913-4}
\rmf{b}^+=[\mr{f},\rmf{n}^+]\oplus \mf{s},
\eeq
where $\mf{s}$ is a $\ad_{\mr{\rho}^\vee}-$invariant subspace, known as transversal subspace, of dimension $n=\rank{\rmf{g}}$. The choice of $\mf{s}$ is not unique, and we will consider in the present paper a fixed but otherwise arbitrary transversal subspace. Recalling the unipotent group \eqref{221108-1/0}, there exists an isomorphism of affine varieties
\begin{equation}\label{eq:231114-2}
\mr{\mc{N}}\times(\mr{f}+\mf{s})\longrightarrow \mr{f}+\rmf{b}^+.
\end{equation}
We finally recall that if $\{\mr{\omega}^\vee_i | i\in\mr{I}\}$ be fundamental coweights of $\rmf{g}$, defined by the relations $\langle\mr{\omega}^\vee_j,\alpha_j\rangle=\delta_{ij}$ ($i,j\in\mr{I}$), then for $i\in\mr{I}$ one has $\omega^\vee_i=\mr{\omega}^\vee_i+a_i\scaling$ where the $\omega^\vee_i$s are fundamental coweights of $\mf{g}$ given in \eqref{221026-4} and $\scaling$ is the scaling element \eqref{220718-4} of $\mf{g}$.  The relation between the elements \eqref{mrrho} and \eqref{220721-4} is therefore given by
\beq\label{220721-4/1}
\rho^\vee=\mr{\rho}^\vee+h\scaling,
\eeq
where $h$ is the Coxeter number of $\mf{g}$.

\subsection{A parabolic subalgebra of $\bmf{g}$.}
Let $\tmf{n}^+$ be the nilpotent Lie algebra defined by \eqref{triangulartildeg}. Due to \eqref{231009-3} we can view $\tmf{n}^+$ as a subalgebra of $\bmf{n}^+$, and more precisely as the subalgebra generated by the elements $\{\tilde{e}_i=\bar{e}_i|i\in\tilde{I}\}$. Given the triangular decomposition \eqref{220717-6} of $\bmf{g}$, and recalling that $\bmf{b}^-=\bmf{n}^-\oplus\bmf{h}$, we define the following subalgebra of $\bmf{g}$:
\beq\label{220717-1}
\bmf{p}=\bmf{b}^-\oplus\tmf{n}^+.
\eeq
This is a parabolic subalgebra of $\bmf{g}$, standard with respect to the Borel subalgebra $\bmf{b}^-$, and it is stable under the action of $\sigma$:
$$x\in\bmf{p}\quad\Longleftrightarrow\quad \sigma x\in \bmf{p}.$$
In view of the application to the theory of opers, we are interested only in the intersection of $\bmf{p}$ with the subspace \eqref{220718-2} of $\bmf{g}$.
\begin{lemma}\label{220720-2}
\begin{enumerate}[i)]
\item  We have
\beq\label{220722-1}
\bmf{p}\cap\left(\bmf{g}\right)^{\geq-1}=\bmf{g}^{-1}\oplus\bmf{h}\oplus\tmf{n}^+,
\eeq
where $\bmf{g}^{-1}$ is given by \eqref{2207118-5} and $\bmf{h}$ is the Cartan subalgebra \eqref{220718-5} of $\bmf{g}$. In particular, $\bmf{p}\cap\left(\bmf{g}\right)^{\geq-1}$ is a finite dimensional $\bb{C}-$vector space. 
\item The Cartan subalgebra $\bmf{h}$ of $\bmf{g}$ decomposes as
\beq\label{220720-3}
\bmf{h}=\bb{C}\scaling\oplus\bb{C}\tmf{h}\oplus \bmf{h}',
\eeq
where 
\beq\label{221004-1}
\bmf{h}'=\langle\bar{\alpha}^\vee_{\sigma^\ell(0)}\,|\,\ell=1,\dots,r\rangle,
\eeq
and $\tmf{h}$ is the Cartan subalgebra \eqref{220718-6/0} of $\tmf{g}$.
\end{enumerate}
\end{lemma}
\begin{proof}
\begin{enumerate}[i)]
\item Obvious from \eqref{220718-2} and \eqref{2207118-5}, and the fact that $\tmf{n}^+\subset\bmf{b}^+$.
\item The Cartan subalgebra $\bmf{h}$ has a basis given by \eqref{220718-5}.  By definition, the coroots in the $\sigma$-orbit of $0$ are a basis of $\bmf{h}'$, while from \eqref{220718-6/0} the remainig coroots form a basis of $\tmf{h}$.
\end{enumerate}
\end{proof}
\begin{proposition}\label{220720-4}
The subalgebra $\bmf{h}\oplus\tmf{n}^+$ of $\bmf{g}$ decomposes with respect to the gradation \eqref{220715-2} as
\beq\label{gradml}
\bmf{h}\oplus\tmf{n}^+=\bigoplus_{\ell=0}^{r-1}\mf{m}_\ell,
\eeq
where $\mf{m}_\ell=(\bmf{h}\oplus\tmf{n}^+)\cap\bmf{g}_\ell$ is given by
$$
\mf{m}_\ell=
\begin{cases}
\mf{h}\oplus\rmf{n}^+,& \ell=0\\
\bmf{h}'_\ell\oplus\tmf{b}^+_\ell,& \ell\geq 1.
\end{cases}
$$
In the formula above, $\mf{h}$ is the Cartan subalgebra \eqref{220720-7} of $\mf{g}$ and $\rmf{n}^+$ is given by \eqref{220720-1/0}.  Moreover,   $\bmf{h}'_\ell=\bmf{h}'\cap\bmf{g}_\ell$, where $\bmf{h}'$ is given by \eqref{221004-1}, and $\tmf{b}^+_\ell=\tmf{b}^+\cap\tmf{g}_\ell$. 
\end{proposition}
\begin{proof}
Since $\bmf{h}$ decomposes as in \eqref{220720-3}, we can write $\bmf{h}\oplus\tmf{n}^+=\bb{C}\scaling\oplus\bmf{h}'\oplus\tmf{b}^+$. The Borel subalgebra $\tmf{b}^+$ can be written as
$$\tmf{b}^+=\bigoplus_{\ell=0}^{r-1}\tmf{b}^+_\ell$$
where in particular $\tmf{b}^+_0=\rmf{b}^+$. On the other hand, the subspace $\bmf{h}'$ given by \eqref{221004-1} decomposes as
$$\bmf{h}'=\bigoplus_{\ell=0}^{r-1}\bmf{h}'_\ell,$$
where each space $\bmf{h}'_\ell$ is $1$-dimensional, generated by
$$\kappa_\ell=\sum_{m=0}^{r-1}\e^{-m\ell}\bar{\alpha}^\vee_{\sigma^m(0)},\qquad \ell=0,\dots,r-1.$$
In particular, $\kappa_0=\alpha^\vee_0.$ The scaling element $\scaling$ satisfies $\sigma(\scaling)=\scaling$, so that $\bb{C}\scaling\in\mf{m}_0$. Thus for $\ell\geq 1$ we get $\mf{m}_\ell=(\bmf{h}\oplus\tmf{n}^+)\cap\bmf{g}_\ell=\bmf{h}'_\ell\oplus\tmf{b}^+_\ell$. On the other hand, $\mf{m}_0=(\bmf{h}\oplus\tmf{n}^+)\cap\bmf{g}_0=(\bb{C}\scaling\oplus\bmf{h}'\oplus\tmf{b}^+)\cap\bmf{g}_0=\bb{C}\scaling\oplus\bb{C}\alpha^\vee_0\oplus(\tmf{b}^+\cap\bmf{g}_0)$. Since $\tmf{b}^+\subset\tmf{g}$ and $\bmf{g}_0\cap\tmf{g}=\tmf{g}_0$, then $\tmf{b}^+\cap\bmf{g}_0=\tmf{b}^+\cap\tmf{g}_0=\rmf{b}^+=\rmf{h}\oplus\rmf{n}^+$, where we used \eqref{220720-1/0}. So $\mf{m}_0=\bb{C}\scaling\oplus\bb{C}\alpha^\vee_0\oplus\rmf{h}\oplus\rmf{n}^+=\langle\scaling,\alpha^\vee_0,\alpha^\vee_1,\dots,\alpha^\vee_n\rangle\oplus\rmf{n}^+=\mf{h}\oplus\rmf{n}^+$. 
\end{proof}
The following results are an immediate consequence of the above proposition.
\begin{corollary}\label{220721-3}
Let $f\in\mf{g}\subset\bmf{g}$ be given by \eqref{220721-1}.
\begin{enumerate}[i)]
\item The affine subspace $f+\mf{m}_0=f+\mf{h}\oplus\rmf{n}^+$ of $\bmf{g}$ is contained in $\mf{g}$. In particular, each element of $f+\mf{m}_0$ is fixed by $\sigma$. Moreover, $f+\mf{m}_0$ is stable under the adjoint action of $\rmf{n}^+$.
\item For each $\ell\geq 1$ the subspace $\mf{m}_\ell=\bmf{h}'_\ell\oplus \tmf{b}^+_\ell$ is a  $\mf{m}_0$-module. 
\item The vector space $\bmf{g}^{-1}\oplus\bmf{h}\oplus\tmf{n}^+$ is a $\mf{m}_0$-module.
\item The affine space $f+\bmf{h}\oplus\tmf{n}^+$ is stable under the adjoint action of $\rmf{n}^+$.
\end{enumerate}
\end{corollary}

\section{Opers}\label{sec:opers}
In this section we consider a class of opers, which we call \emph{affine  twisted parabolic Miura opers}, we describe their properties, and we show that the FFH connections considered in \eqref{eq:ffopersourgauge} are representative of this class, thus proving a conjecture by Hernandez and Frenkel \cite{FH16}. Moreover, we deduce a normal form for FFH connections.

The notion of affine  twisted parabolic Miura opers, to be defined below, builds upon several other notions of classes of opers, which we briefly recall. The original notion of opers (with values in a finite dimensional simple Lie algebra) appeared for the first time Drinfeld and Sokolov \cite{DS85} and was later formalized by Beilinson and Drinfeld \cite{bedr02}. Affine opers  were introduced by Ben-Zvi and Frenkel \cite{frebenzvi01} (see also \cite{FF11}) as suitable analogues of opers with values in an (affine) Kac--Moody algebra.  Parabolic Miura opers were first considered in \cite{FF11} as generalization to the parabolic case of the notion of Miura opers \cite{frebenzvi01}.  Finally, twisted opers were defined -- in the finite dimensional case -- by Frenkel and Gross \cite{frgr09}.

\subsection{Twisted parabolic Miura $\mf{g}$-opers over $\bb{C}^\times$} 
We consider twisted parabolic Miura $\mf{g}$-opers for the class of affine algebras $\bmf{g}$ listed in Table \ref{table:affine} and choosing $\bmf{p}\subset\bmf{g}$ given in \eqref{220717-1} as parabolic subalgebra. The construction can be generalized in the obvious way to other Lie algebras and different choices of parabolic subalgebras.

Let $\mc{K}$ be the field of meromorphic functions on $\bb{C}^\times$, and for any vector space $V$ denote by  $V(\mc{K})$ the space of  meromorphic functions from $\bb{C}^\times$ to $V$. Recall the principal gradation \eqref{220722-2} of $\bmf{g}$, the associated filtration \eqref{220722-3}, and the element $f\in\bmf{g}^{-1}$ given by \eqref{220717-3} or equivalently by \eqref{220721-1}. Let moreover $\bmf{p}$ be the parabolic subalgebra of $\bmf{g}$ defined in \eqref{220717-1}. We introduce a set of operators taking values in $(\bmf{g})^{\geq-1}\cap\bmf{p}$ or equivalently, due to  \eqref{220722-1}, with values in $\bmf{g}^{-1}\oplus\bmf{h}\oplus\tmf{n}^+$.
\begin{definition}\label{220717-2}
We denote by $\op_{\bmf{g},\bmf{p}}(\mc{K})$  the set of operators
\beq\label{220714-1}
L(z)=\partial_z+\frac{1}{z}\left(f+b(z)\right),\qquad b\in (\bmf{h}\oplus \tmf{n}^+)(\mc{K}).
\eeq
\end{definition}
Following \cite{frgr09}, we define twisted opers by considering a class of operators in $\op_{\bmf{g},\bmf{p}}(\mc{K})$ which are fixed under a transformation which involves both the automorphism $\sigma$ and a rotation in $\bb{C}^\times$. We first consider the two actions separately:
\begin{lemma}\label{220721-5}
The set $\op_{\bmf{g},\bmf{p}}(\mc{K})$ is stable under the automorphism of $\bb{C}^\times$ given by $z\mapsto az$ ($a\in\bb{C}^\times)$. Indeed, this map induces the transformation 
\beq\label{220721-2}
\partial_z+\frac{1}{z}\left(f+b(z)\right)\longrightarrow \partial_z+\frac{1}{z}\left(f+b(az)\right).
\eeq
in $\op_{\bmf{g},\bmf{p}}(\mc{K})$.
\end{lemma}
\begin{proof}
Under a local change of coordinates $z=\varphi(x)$, the operator \eqref{220714-1} becomes
$$(\varphi^\ast L)(x)=\partial_x+\frac{\varphi'(x)}{\varphi(x)}(f+b(\varphi(x))).$$
Choosing $\varphi(x)=ax$ ($a\in\bb{C}^\times$) we get $(\varphi^\ast L)(x)=\partial_x+\frac{1}{x}\left(f+b(ax)\right)\in \op_{\bmf{g},\bmf{p}}(\mc{K})$.
\end{proof}
We are interested in the special case $a=e^{2\pi it}$, with $t\in\bb{R}$.
\begin{definition}
Let $L(z)\in \op_{\bmf{g},\bmf{p}}(\mc{K})$ be given by \eqref{220714-1}. For every $t\in\bb{R}$ we define the rotated operator $L_t(z)\in \op_{\bmf{g},\bmf{p}}(\mc{K})$ as
\beq\label{221116-3}
L_t(z)=\partial_z+\frac{1}{z}\left(f+b(e^{2\pi it}z)\right).
\eeq
\end{definition}
The requirement that $L(z)\in \op_{\bmf{g},\bmf{p}}(\mc{K})$ is meromorphic at $z=0$ is thus expressed by the condition
\beq\label{L1L}
L_1(z)=L(z).
\eeq
We now consider the action of $\sigma$. 
\begin{definition}
Given an operator $L(z)=\partial_z+A(z)$ with $A\in\bmf{g}(\mc{K})$ we define the twisted operator as $L^\sigma(z)=\partial_z+\sigma(A)(z)$.
In particular, for $L(z)\in \op_{\bmf{g},\bmf{p}}(\mc{K})$ as in \eqref{220714-1}, and if $b=\sum_{\ell=0}^{r-1}b_\ell$ according to the decomposition \eqref{gradml}, with $b_\ell\in\mf{m}_\ell$, then
\beq\label{tpmgoper-general}
L^\sigma(z)=\partial_z+\frac{1}{z}\left(f+\sum_{\ell=0}^{r-1}\e^\ell b_\ell(z)\right),
\eeq
where $\e$ is given by \eqref{epsilon}.
\end{definition}
We extend the action of $\sigma$ from $\bmf{g}$ to $\bb{C}^\times\times \bmf{g}$ by considering the map 
$$(z,x)\mapsto(\e^{-1}z,\sigma(x)),\qquad z\in\bb{C}^\times,x\in\bmf{g},$$
which induces the transformation
$$L(z)\mapsto L^\sigma_{-\frac1r}(z)$$
on $\op_{\bmf{g},\bmf{p}}(\mc{K})$. We restrict to the class of operators \eqref{220714-1} which are pointwise fixed:
\beq\label{Lsigma1/r}
L^\sigma_{-\frac1r}(z)=L(z).
\eeq
Using Proposition \ref{220720-4}, we obtain that these invariant operators are of the form
\beq\label{220720-9}
L(z)=\partial_z+\frac{1}{z}\left(f+\sum_{\ell=0}^{r-1}z^\ell b_\ell(z^r)\right),
\eeq
where 
\begin{align*}
b_0&\in \mf{m}_0(\mc{K})=\mf{h}(\mc{K})\oplus\rmf{n}^+(\mc{K})\\
b_\ell&\in \mf{m}_\ell(\mc{K})= \bmf{h}'_\ell(\mc{K})\oplus \tmf{b}^+_\ell(\mc{K}),\qquad \ell\geq 1.
\end{align*}
\begin{definition}\label{221114-1}
We denote by $\op^{\sigma}_{\bmf{g},\bmf{p}}(\mc{K})$  the set of operators \eqref{220720-9}.
\end{definition}
From Corollary \ref{220721-3}, the affine subspace $f+\bmf{h}\oplus \tmf{n}^+$ is stable under the adjoint action of $\rmf{n}^+$. 
\begin{definition}\label{221115-1}
We consider the Gauge group $\mr{\mc{N}}^{\sigma}(\mc{K})=\{\exp(y(z^r))\,|\,y\in\rmf{n}^+(\mc{K})\}$.  Let $L(z)\in \op^{\sigma}_{\bmf{g},\bmf{p}}(\mc{K})$ be given by $L(z)=\partial_z+(f+b(z))/z$. For every $y\in\rmf{n}^+(\mc{K})$ the action of $\exp(\ad y)\in\mr{\mc{N}}^{\sigma}(\mc{K})$ on $L(z)$ is defined as
\beq\label{220707-2}
\exp(\ad y(z^r))L(z)=\partial_z+\frac{1}{z}\left(\sum_{k\geq 0}\frac{(\ad y)^k}{k!}\left(f+b(z)\right)-rz^{r}\sum_{k\geq 0}\frac{(\ad y)^k}{(k+1)!}\frac{dy}{dz}\right),
\eeq
\end{definition}
Notice that $\op^{\sigma}_{\bmf{g},\bmf{p}}(\mc{K})$ is stable under the action of $\mr{\mc{N}}^{\sigma}(\mc{K})$. We now define meromorphic twisted parabolic Miura $\mf{g}-$opers.
\begin{definition}
The space of meromorphic twisted $(\bmf{g},\bmf{p})-$opers (or twisted parabolic Miura $\mf{g}-opers$) on $\bb{C}^\times$ is defined as 
$$\Op^{\sigma}_{\bmf{g},\bmf{p}}(\mc{K})
=\op^{\sigma}_{\bmf{g},\bmf{p}}(\mc{K})/\mr{\mc{N}}^{\sigma}(\mc{K}).$$
We denote by $[L(z)]$ the oper corresponding to the operator $L(z)$.
\end{definition}
It is often useful \cite{bedr02,lacroix18} to enlarge the class of operators $\op^{\sigma}_{\bmf{g},\bmf{p}}(\mc{K})$ by allowing the action of a more general Gauge group. To this aim, provide the following
\begin{definition}
Let $\mc{H}^{\sigma}(\mc{K})$ the abelian group generated by $\varphi(z^r)^\mu$ with $\varphi\in\mc{K}\setminus\{0\}$ and $\mu\in P^\vee({\mf{g}})$, the coweight lattice of $\mf{g}$. Let $L(z)\in \op^{\sigma}_{\bmf{g},\bmf{p}}(\mc{K})$ be given by $L(z)=\partial_z+(f+b(z))/z$, and decompose $b$ as $b=b_0+\sum_{\alpha\in\bar{\Delta}}b_\alpha\in\bmf{g}$ with $b_0\in\bmf{h}(\mc{K})$ and $b_\alpha\in\bmf{g}_\alpha(\mc{K})$. We define
\beq\label{220715-1}
\varphi(z^r)^{\ad\mu} L(z)=\partial_z+\frac{1}{z}\left(\sum_{i\in I}\varphi(z^r)^{-\langle\mu,\alpha_i\rangle}f_i-\frac{rz^{r}\varphi'(z^r)}{\varphi(z^r)}\mu+\sum_{\alpha\in\bar{\Delta}}\varphi(z^r)^{\langle\mu,\alpha\rangle}b_\alpha(z)\right).
\eeq
\end{definition}
Note that the only element $\mc{H}^{\sigma}(\mc{K})$ leaving $\op^{\sigma}_{\bmf{g},\bmf{p}}(\mc{K})$ invariant is the identity.
\begin{definition}
We set $\widetilde{\op}^{\sigma}_{\bmf{g},\bmf{p}}(\mc{K})=\mc{H}^{\sigma}(\mc{K})\op^{\sigma}_{\bmf{g},\bmf{p}}(\mc{K})$. In other words, each element in $\widetilde{\op}^{\sigma}_{\bmf{g},\bmf{p}}(\mc{K})$ is of the form \eqref{220715-1} for some $\varphi(z^r)^\mu$ in $\mc{H}^{\sigma}(\mc{K})$ and $L(z)\in \op^{\sigma}_{\bmf{g},\bmf{p}}(\mc{K})$.
\end{definition}
Note that each element in $\widetilde{\op}^{\sigma}_{\bmf{g},\bmf{p}}(\mc{K})$ takes values in $(\bmf{g})^{\geq-1}\cap\bmf{p}$. The action of $\mc{H}^{\sigma}(\mc{K})$ on $\bar{\mf{g}}(\mc{K})$ induces an action on $\mr{\mc{N}}^{\sigma}(\mc{K})$, given by
$$\varphi^{\ad\mu}\exp(y)=\exp(\varphi^{\ad\mu} y),\qquad y\in\rmf{n}^+(\mc{K}),\varphi\in\mc{K}\setminus\{0\},\mu\in P^\vee(\mf{g}).$$
\begin{definition}
We denote by $\mc{B}^{\sigma}(\mc{K})=\mc{H}^{\sigma}(\mc{K})\rtimes\mr{\mc{N}}^{\sigma}(\mc{K})$ the Gauge group obtained as a semidirect product of $\mc{H}^{\sigma}(\mc{K})$ and $\mr{\mc{N}}^{\sigma}(\mc{K})$, with multiplication given by
$$(\exp(y_1),\phi_1^{\mu_1})(\exp(y_2),\phi_2^{\mu_2})=(\exp(y_1)\exp(\phi_1^{\ad\mu_1}y_2),\phi_1^{\mu_1}\phi_2^{\mu_2}).$$
\end{definition}
Note that $\mc{H}^\sigma(\mc{K})$ is normal in $\mc{B}^\sigma(\mc{K})$.  Moreover, for each $L(z)\in \widetilde{\op}^\sigma_{\bmf{g},\bmf{p}}(\mc{K})$ there is a unique element in $\mc{H}^\sigma(\mc{K})$ mapping $L(z)$ to $\op^\sigma_{\bmf{g},\bmf{p}}(\mc{K})$. It thus follows that there is a bijection between the sets  $\op^\sigma_{\bmf{g},\bmf{p}}(\mc{K})/\mr{\mc{N}}^\sigma(\mc{K})$ and $\widetilde{\op}^\sigma_{\bmf{g},\bmf{p}}(\mc{K})/\mc{B}^\sigma(\mc{K})$.
%

%

\begin{remark}
Recall the element $\rho^\vee\in P^\vee(\mf{g})$ given by \eqref{220721-4/1}. Applying the Gauge $z^{-\rho^\vee}\in\mc{H}^\sigma(\mc{K}) $ to an operator of type \eqref{220714-1} one gets an operator of the form $\partial_z+f+b(z)$, for some $b\in\mf{h}(\mc{K})\oplus\rmf{n}^+(\mc{K})$. This Gauge is often useful when dealing with opers \cite{bedr02}; we prefer however to use the Gauge \eqref{220714-1} since it is invariant under the map $z\mapsto az$, see Lemma \ref{220721-5}. The class of operators of the form \eqref{220714-1} looks similar to the notion of $(\leq 1)$-singular opers, introduced by Beilinson and Drinfeld \cite[Sec. 3.8.8]{bedr02}. Note, however, that in the present construction we allow the functions appearing in operators as well as the elements of the Gauge group to be meromorphic (over $\bb{C}^\times$).
\end{remark}

\subsection{A special class of twisted parabolic Miura $\mf{g}$-opers.}
We restrict to a special class of twisted parabolic Miura $\mf{g}$-opers. The FFH connections \eqref{eq:ffopersourgauge} turn out to be representatives of opers of this class.
\begin{definition}
We denote by $\op^{\sigma}_{\mf{g},v_\theta}(\mc{K})$ the set of meromorphic operators of the form
\beq\label{231418-1} 
L(z)=\partial_z+\frac{1}{z}\left(f+\xi(z^r)\scaling +b(z^r)+zc(z^r)v_\theta\right),
\eeq
with $\xi\in\mc{K}$ and $b\in\rmf{b}^+(\mc{K})$, and  $v_\theta\in\tmf{g}$ is the element introduced in Definition \ref{221108-2/0}. 
\end{definition}
Since $\mf{h}\oplus\rmf{n}^+=\bb{C}\scaling\oplus\bb{C}K\oplus\rmf{b}^+$ and $\sigma v_\theta=\e v_\theta$, it follows that the operators \eqref{231418-1} are a special class of operators \eqref{220720-9}, or equivalently that $\op^{\sigma}_{\mf{g},v_\theta}(\mc{K})\subset \op^{\sigma}_{\bmf{g},\bmf{p}}(\mc{K}).$ In particular, they satisfy the invariant condition \eqref{Lsigma1/r}.  Moreover, recalling the Gauge group $\mr{\mc{N}}^\sigma(\mc{K})$  (cf. Definition \ref{221115-1}), one can easily show that $g.(L(z)-zc(z^r)v_\theta)+zc(z^r)v_\theta\in \op^{\sigma}_{\mf{g},v_\theta}(\mc{K})$, for every $g\in\mr{\mc{N}}^\sigma(\mc{K})$. In addition, due to \eqref{221115-2/0}, then $g.v_\theta=v_\theta$, ($g\in\mr{\mc{N}}^\sigma(\mc{K})$), so that  $\op^{\sigma}_{\mf{g},v_\theta}(\mc{K})$ is stable under the adjoint action of $\mr{\mc{N}}^\sigma(\mc{K})$.
\begin{definition}
We define $\Op^{\sigma}_{\mf{g},v_\theta}(\mc{K})=\op^{\sigma}_{\mf{g},v_\theta}(\mc{K})/\mr{\mc{N}}^{\sigma}(\mc{K})$, and we denote by $[L]$ the equivalence class of $L\in \op^{\sigma}_{\mf{g},v_\theta}(\mc{K})$.
\end{definition}
\begin{definition}
Let $\mf{s}\subset\rmf{b}^+$ be a transversal subspace, as in \eqref{230913-4}. The operator \eqref{231418-1} is said to be in canonical form (with respect to $\mf{s}$) if $b\in\mf{s}(\mc{K})$.
\end{definition}
\begin{proposition}\label{prop: 231114}
Fixed a transversal subspace $\mf{s}\subset\rmf{b}^+$, each oper $[L(z)]\in \Op^{\sigma}_{\mf{g},v_\theta}(\mc{K})$ admits a unique canonical form (with respect to $\mf{s}$).
\end{proposition}
\begin{proof}
Since $v_\theta$ and $\scaling$ are $\mr{\mc{N}}$-invariant, it is sufficient to prove the proposition for an operator of the form $L(z;\la)=\partial_z+\left(\mr{f}+b(z^r)\right)/z$. For these operators the result is well-known \cite{DS85}.
\end{proof}
We now show how to represent the above operators (and opers) in more concrete terms by means of the loop realization \eqref{loopg} of the Kac-Moody algebra $\mf{g}$.    
Let $\tilde{\xi}(z)$ satisfy $\tilde{\xi}'(z)=\frac{\xi(z)}{rz}$, recall that $\scaling\in P^\vee(\mf{g})$, and define 
\beq\label{L(z,la)defalg}
L(z;\la)=e^{\tilde{\xi}(z^r)\ad\scaling}L(z),
\eeq
where we used the Gauge \eqref{220715-1}. Explicitly, we have:
\beq\label{L(z,la)}
L(z;\la)=\partial_z+\frac{1}{z}\left(\mr{f}+b(z^r)+\left(zc(z^r)+e^{-\tilde{\xi}(z^r)}\la\right)v_\theta\right),
\eeq
we call the above the \emph{loop realization} of the operator $L(z)$. The Gauge group $\mr{\mc{N}}^{\sigma}(\mc{K})$ acts on the loop realization \eqref{L(z,la)}, so that we can define the oper $[L(z;\la)]$ corresponding to \eqref{L(z,la)} as its equivalence class.  Moreover, since $[\scaling,\rmf{n}^+]=0$ we have
$$[L(z;\la)]=[e^{\tilde{\xi}(z^r)\ad\scaling}L(z)]=e^{\tilde{\xi}(z^r)\ad\scaling}[L(z)],$$
which provides a bijection between the equivalence classes, namely between the opers.
From Proposition \ref{prop: 231114} the following result follows.
\begin{corollary}
Let $\mf{s}\subset\rmf{b}^+$ be a transversal subspace, as in \eqref{230913-4}. For each $L(z;\la)$ as in \eqref{L(z,la)} the oper $[L(z,\la)]$ admits a unique canonical form 
$$L_{\mf{s}}(z;\la)=\partial_z+\frac{1}{z}\left(\mr{f}+s(z^r)+\left(zc(z^r)+e^{-\tilde{\xi}(z^r)}\la\right)v_\theta\right),$$
for some $s\in\mf{s}(\mc{K})$.
\end{corollary}
Note that the term $e^{-\tilde{\xi}(z^r)}$ in \eqref{L(z,la)} does not need to be meromorphic. In particular, under the map $z\mapsto e^{2\pi i}z$ we have
\beq\label{230419-1}
e^{-\tilde{\xi}(z^r)}\mapsto e^{-2\pi i\xi(0)}e^{-\tilde{\xi}(z^r)}.
\eeq
In order to deal with this, for an operator $L(z;\la)$ of the form \eqref{L(z,la)} and $t\in\bb{R}$ we define the rotated operator as
the operator induced by the map
$$(z,\la)\mapsto (e^{2\pi it}z,e^{2\pi i\xi(0)t}\la),$$
that is, we set:
\beq\label{L_t(z,la)}
L_t(z;\la)=\partial_z+\frac{1}{z}\left(\mr{f} +b(e^{2\pi irt}z^r)+\left(e^{2\pi it}zc(e^{2\pi irt}z^r)+e^{-\tilde{\xi}(e^{2\pi irt}z^r)}e^{2\pi i\xi(0)t}\la\right)v_\theta\right).
\eeq
With this choice, the following property holds: the operator   \eqref{231418-1} is meromorphic at $z=0$ if and only if its loop realization \eqref{L(z,la)} satisfies
\beq\label{doreytateo}
L_1(z;\la)=L(z;\la).
\eeq
We now consider in more detail the case when $\xi(x^r)$ is a nonzero constant, say $\xi(x^r)=-k$.  This is the case of interest for the ODE/IM correspondence, see \eqref{looprealization}. Then, \eqref{231418-1} reads
$$
L(z)=\partial_z+\frac{1}{z}\left(f-k\scaling +b(z^r)+zc(z^r)v_\theta\right),
$$
and it is meromorphic over $\bb{C}$, but its loop realization
$$
L(z;\la)=\partial_z+\frac{1}{z}\left(\mr{f}+b(z^r)+\left(zc(z^r)+z^{k}\la\right)v_\theta\right),
$$
is meromorphic in $\bb{C}^\times$ only. In this case, the rotated operator \eqref{L_t(z,la)} takes the form
$$
L_t(z;\la)=\partial_z+\frac{1}{z}\left(\mr{f} +b(e^{2\pi irt}z^r)+\left(e^{2\pi it}zc(e^{2\pi irt}z^r)+z^{k}\la\right)v_\theta\right),
$$
and it obviously satisfies the identity \eqref{doreytateo}, which in the ODE/IM correspondence setting is known as Dorey-Tateo symmetry (cf. \eqref{eq:symanzik}).


\subsection{Feigin-Frenkel-Hernandez $\mf{g}$-opers}
We define FFH $\mf{g}$-opers as a class of twisted parabolic Miura $\mf{g}$-opers \eqref{231418-1} satisfying certain local assumptions, to be specified below. This is essentially the approach originally given in \cite{FH16}. We show  that each FFH $\mf{g}$-oper admits
an (essentially unique) representative as a FFH connection \eqref{eq:ffopersourgauge}. In this discussion we fix 
a transversal subspace $\mf{s}\subset\rmf{g}$. Recall that due to \eqref{eq:231114-2}, for every $\ell\in\rmf{h}$ there exists a unique pair $(N,\bar{\ell})\in\mr{\mc{N}}\times \mf{s}$ such that $N(\mr{f}+\ell)=\mr{f}+\bar{\ell}$. It follows that in the case $J= \emptyset$ the FFH connection \ref{eq:ffopersourgauge} has a canonical form
\beq\label{231009-1}
 \mc{L}_{G,\mf{s}}(z)= \partial_z+
 \frac{1}{z}\left(f+\bar{\ell} -k \, \scaling  + zv_\theta \right).
\eeq

 \begin{definition}
 Fixed $\bar{\ell} \in \mathfrak{s}$ and $k \in (0,1)$, a FFH oper is a  twisted parabolic  Miura
 $\mf{g}$-oper with the canonical form
\begin{equation}\label{eq:KdVopers}
\mc{L}_{\mf{s}}(z)=\mc{L}_{G,\mf{s}}(z;\la)+s(z^r)/z,
\end{equation}
where $s$ is a meromorphic funciton, with values in  $\mf{s}$, which satisfies four assumptions :
\begin{asu}\label{asu1} 
The (dominant term of the) asymptotic behaviour of solutions in a neighborhood of $0$ is independent on $s$:
$$s(z^r)=O(z),\qquad  z \to 0.$$
\end{asu}
\begin{asu}\label{asu2}
The (dominant term of the) asymptotic behaviour of solutions in a neighborhood of $\infty$ is independent on $s$:
$$s(z^r)=O(1),\qquad z \to \infty.$$ 
\end{asu}
\begin{asu}\label{asu3}
If $w \in \mathbb{C}^*$ is a singularity of $s$ then locally
$$s(z^r)=\sum_{i=1}^n \sum_{l\geq 0}^{d_i} \frac{zs^{d_i}_l}{(z-w)^{d_i+1-l}}, $$
where $s_l^{d_i}\in\mf{s}\cap\mf{g}^{d_i}$, $l=0,\dots,d_i$, and where $\bar{s}=\sum_i s_0^{d_i} \in \mathfrak{s}$ is the unique element in
 $\mathfrak{s}$ such that the  $\mr{f}/w-\mr{\rho}^\vee+\bar{s}$ and  $\mr{f}/w-\mr{\rho}^\vee-\theta^\vee \in \mathring{\mf{g}} $ are conjugated. In other words,  every additional singularity is a regular singularity and close to the singularity the oper is locally Gauge equivalent to
 $$\partial_z+ \frac{\mr{f}/w-\mr{\rho}^\vee -\theta^\vee}{z-w}+O(1).$$
\end{asu}
\begin{asu}\label{asu4}
If $w \in \mathbb{C}^*$ is a singularity of $s$, the monodromy at $w$ is trivial  for every $\la \in \bb{C}$.
\end{asu}
\end{definition}
The above definition was given in \cite{FF11} in the case $r=1$ and in \cite{FH16} in the general case.

\begin{theorem}\label{pro:quasinormal}
Any affine parabolic twisted Miura $\mf{g}$ oper with the canonical form \eqref{eq:KdVopers}, where the function
 $s$ satisfies the Assumptions 1,2, and 3 above admits the representation \eqref{eq:ffopersourgauge} for some $J \subset \bb{N}$, with $(w_j ,X(j),y(j)) \in  \mathbb{C}^*\times \mathring{\mf{n}}^+ \times \bb{C}$, $j\in J$. Moreover, the above representation is essentially unique:
Fixed a $\ell \in \mathring{\mf{h}}$ such that $f+\bar{\ell}$ and $f+\ell \in \mathring{\mf{g}} $ are conjugated, two such opers coincide if and only if
their~ representation (\ref{eq:ffopersourgauge}) coincide.
\end{theorem}
\begin{proof}
Recall that $f=\mr{f}+f_0$.  Since the elements $\scaling$, $f_0$ and $v_\theta$ are invariant under the action of the Gauge group $\mr{\mc{N}}(\mc{K})$, the thesis is equivalent to the existence and uniqueness of a representation of the form
  \begin{align}\label{FFoper}
 \mathring{\mc{L}}= & \partial_z+ \frac{1}{z}\left(\mathring{f}+\ell  +
 \sum_{j \in J} \frac{r z^{r}}{z^r-w_j^r}\left(-\theta^\vee + X(j) \right)\right),
 \end{align}
 for a $\mathring{\mf{g}}$-oper
\beq\label{230511-1}
 \mathring{\mc{L}}=  \partial_z+ \frac{1}{z}\left(\mathring{f}+b(z)\right) ,
\eeq
with $b \in \rmf{b}^+(\mc{K})$ and $b(z)=b(\e z)$,  whose canonical form satisfies Assumptions 1,2,3. We prove it as follows. Due to  \cite[Proposition 4.7]{mara18}, the canonical form of an oper \eqref{230511-1} (not necessarily $\sigma-$invariant)   which satisfies  Assumptions 1,2, and 3,  is
 \begin{align}\label{eq:canproof}
  \mathring{\mc{L}}_{\mf{s}}= \partial_z+ \frac{\mathring{f}+\bar{\ell} }{z} +
   \sum_{j =1}^M \sum_{i=1}^n \sum_{l=0}^{d_i} \frac{s^{d_i}_l(j)}{(z-\tilde{w}_j)^{d_i+1-l}},
 \end{align}
for some $M \in \bb{N}$, $\tilde{w}_j \in \bb{C}^*$, and $s^{d_i}_l(j) \in \mf{s}\cap\rmf{g}^{d_i}$, where $\bar{\ell} \in \mathfrak{s}$ 
is the unique element such  that the  $\mr{f}+\bar{\ell}$ and $\mr{f}+\ell \in \mathring{\mf{g}} $ are conjugated, and for $j=1,\dots,M$ the element $\bar{s}(j)=\sum_i s^{d_i}(j) \in \mathfrak{s}$ is  the unique element  such that   $\mr{f}/\tilde{w}_j-\mathring{\rho}^\vee+\bar{s}(j)$ and $\mr{f}/\tilde{w}_j-\mathring{\rho}^\vee-\theta^\vee \in  \mathring{\mf{g}} $ are conjugated.
 In addition, according to  \cite[Theorem 6.1]{mara18}, if we fix $\ell$, an oper of the form (\ref{eq:canproof}) can uniquely be represented as
 \begin{align}\label{eq:proofnormpart}
 \mathring{\mc{L}}= & \partial_z+  \frac{\mathring{f}+\ell }{z} + 
 \sum_{k =1}^M \frac{1}{z-\tilde{w}_j} \left(-\theta^\vee + \tilde{X}(j) \right),
 \end{align}
 for some  $\tilde{X}(j) \in \mathring{\mf{g}}$  via a Gauge transform belonging to $\mathring{\mc{N}}(\mc{K})$.  
 Imposing the $\sigma-$invariance it follows that  the set of additional singularities must be invariant under rotations by $\e$ (namely that $\tilde{w}_j$ is an additional singularity if and only if $\e \tilde{w}_j$ is another additional singularity) and that \eqref{eq:proofnormpart} has to be of the form \eqref{FFoper}, with $|J|=M/3$ and for certain, uniquely determined, coefficients $X(j)$, $j\in J$.
\end{proof}

\def\cprime{$'$} \def\cydot{\leavevmode\raise.4ex\hbox{.}} \def\cprime{$'$}

\end{document}